\documentclass{amsart}
\usepackage{amsmath,amssymb,amsrefs}
\usepackage{graphicx}
\usepackage{bm}
\usepackage{hyperref} 
\usepackage{enumerate}

\numberwithin{equation}{section}

\newtheorem{theorem}{Theorem}[section]
\newtheorem{corollary}[theorem]{Corollary}
\newtheorem{proposition}[theorem]{Proposition}
\newtheorem{lemma}[theorem]{Lemma}
\theoremstyle{definition}
\newtheorem{definition}[theorem]{Definition}
\theoremstyle{remark}
\newtheorem{remark}[theorem]{Remark}

{\unskip\nobreak\hskip 1em plus 1fil\nobreak$\Box$
\parfillskip=0pt%
\endtrivlist}

\newcommand{\R}{\mathbb{R}}
\newcommand{\C}{\mathbb{C}}
\newcommand{\Z}{\mathbb{Z}}
\newcommand{\N}{\mathbb{N}}
\newcommand{\B}{\mathcal{B}}
\newcommand{\brill}{{\mathcal{B}}}

\newcommand{\bk}{{\bf k}}
\newcommand{\bfm}{{\bf m}}
\newcommand{\bn}{{\bf n}}
\newcommand{\bK}{{\bf K}}
\newcommand{\bKp}{{\bf K'}}
\newcommand{\bv}{{\bf v}}
\newcommand{\bw}{{\bf w}}
\newcommand{\bx}{{\bf x}}
\newcommand{\bX}{{\bf X}}
\newcommand{\by}{{\bf y}}

\newcommand{\vtilde}{{\bm{\mathfrak{v}}}}
\newcommand{\ktilde}{{\bm{\mathfrak{K}}}}
\newcommand{\smallktilde}{{\mathfrak{K}}}
\newcommand{\kpar}{{k_{\parallel}}}
\newcommand{\kparpi}{{\kpar=2\pi/3}}
\newcommand{\kparv}{{\kpar=\bK\cdot\vtilde_1}}

\newcommand{\abs}[1]{\left\lvert#1\right\rvert}
\newcommand{\norm}[1]{\left\lVert#1\right\rVert}
\newcommand{\inner}[1]{\left\langle#1\right\rangle}
\newcommand{\D}{\partial}

\newcommand{\eps}{\varepsilon}
\newcommand{\nit}{\noindent}
\newcommand{\nn}{\nonumber}

\newcommand{\lamsharp}{{\lambda_{\sharp}}}
\newcommand{\thetasharp}{{\vartheta_{\sharp}}}
\newcommand{\exponent}{\nu}
\newcommand{\tsigma}{{\tilde{\sigma}}}
\newcommand{\eig}{E}
\newcommand{\supp}{\text{supp}}
\newcommand{\nofold}{no-fold }
\newcommand{\pit}{(\lambda^2+\eps)(\lambda^2+\delta^2)}

\begin{document}

\title{Edge states in honeycomb structures}

\author{C. L. Fefferman}
\address{Department of Mathematics, Princeton University, Princeton, NJ, USA}
\email{cf@math.princeton.edu}
\thanks{The first author was supported in part by NSF grant DMS-1265524.}

\author{J. P. Lee-Thorp}
\address{Department of Applied Physics and Applied Mathematics, Columbia University, New York, NY, USA}
\email{jpl2154@columbia.edu}

\author{M. I. Weinstein}
\address{Department of Applied Physics and Applied Mathematics and Department of Mathematics, Columbia University, New York, NY, USA}
\email{miw2103@columbia.edu}
\thanks{The second and third authors were supported in part by NSF grants: DMS-10-08855, DMS-1412560, DGE-1069420 and Simons Foundation Math + X Investigator grant \#376319 (MIW)}

\date{\today}

\subjclass[2010]{Primary 35J10 35B32;\\Secondary 35P  35Q41 37G40}

\keywords{Schr\"odinger equation, Dirac equation, Dirac point, Floquet-Bloch spectrum, Topological insulator,  Edge states, Domain wall, Honeycomb lattice}

\begin{abstract} 
An edge state is a time-harmonic solution of a conservative wave system, {\it e.g.} Schr\"odinger, Maxwell,  which is propagating (plane-wave-like) parallel to, and  localized transverse to, a line-defect or  ``edge". Topologically protected edge states are edge states which are stable against spatially localized (even strong) deformations of the edge.  First studied in the context of the quantum Hall effect, protected edge states have attracted huge interest  due to their role in the field of topological insulators. Theoretical understanding of topological protection has mainly come from discrete (tight-binding) models and direct numerical simulation. In this paper we consider a rich family of \underline{continuum} PDE models for which we rigorously study regimes where topologically protected edge states exist.
 
Our model is  a class of Schr\"odinger operators on $\R^2$ with a background two-dimensional honeycomb potential perturbed by an ``edge-potential''. The edge potential is a domain-wall interpolation, transverse to a prescribed ``rational'' edge,  between two distinct periodic structures. General conditions are given for the bifurcation of a branch of topologically protected edge states from {\it Dirac points} of the background honeycomb structure. The bifurcation is seeded by the zero mode of a one-dimensional effective Dirac operator. A key condition is a spectral \nofold condition for the prescribed edge.
 We then use this result to prove the existence of topologically protected edge states along zigzag edges of certain honeycomb structures.
Our results are consistent with the physics literature and appear to be the first rigorous results on the existence of topologically protected edge states for continuum 2D PDE systems describing waves in a non-trivial periodic medium.
We also show that the family of Hamiltonians we study contains cases where zigzag edge states exist, but which are not topologically protected. 
\end{abstract}

\maketitle


\section{Introduction and Outline}\label{intro}

This paper is motivated by a remarkable physical observation. When two distinct 2-dimensional materials with favorable crystalline structures are joined along an edge, there exist propagating modes, {\it e.g.} electronic or photonic, whose energy remains localized in a neighborhood of the edge without spreading into the ``bulk''. Furthermore, these modes and their properties persist in the presence of  arbitrary local, even large,  perturbations of the edge.   An understanding  of such ``protected edge states" in periodic structures
 has so far mainly  been obtained by analyzing discrete ``tight-binding" models
 and from numerical simulations.  In this paper we prove that edge states arise from the Schr\"odinger equation for a class of potentials that have many features (not all) in common with the relevant experiments.
  A central role is played by 
  a spectral ``no-fold" condition. In the case of small amplitude (low-contrast) honeycomb potentials, this reduces to a sign condition of a particular Fourier coefficient of the potential.
 A combination of numerical simulation and heuristic argument suggests
that if the ``no-fold" condition fails, then edge
states need not be topologically protected.  Let us explain these ideas in more detail.

Wave transport in periodic structures with honeycomb symmetry has been an area of intense activity  catalyzed by the study of graphene, a single atomic layer two-dimensional honeycomb structure of carbon atoms. The  remarkable electronic properties exhibited by graphene  \cites{geim2007rise, RMP-Graphene:09, Katsnelson:12, zhang2005experimental} have inspired the study  of waves in general honeycomb structures or  ``artificial graphene''   in electronic \cites{artificial-graphene} and  photonic \cites{HR:07,RH:08,Chen-etal:09,bahat2008symmetry,lu2014topological,BKMM_prl:13} contexts. 
One such property, observed in electronic and photonic systems with honeycomb symmetry is the existence of 
topologically protected {\it edge states}.   Edge states are modes which are
(i) pseudo-periodic (plane-wave-like  or propagating) parallel to a line-defect, and (ii)  localized transverse to the line-defect; see Figure \ref{fig:mode_schematic}. 
{\it Topological protection} refers to the persistence of these modes and their properties, even when the line-defect is subjected to strong local or random perturbations.  In applications, edge states are of great interest due to their potential as robust vehicles for channeling energy. 

The extensive physics literature on topologically robust edge states goes back to investigations of the quantum Hall effect; see, for example, \cites{H:82, TKNN:82, Hatsugai:93, wen1995topological} and the rigorous mathematical articles \cites{Macris-Martin-Pule:99,EG:02, EGS:05,Taarabt:14}.  
In \cites{HR:07, RH:08} a proposal for realizing {\it photonic edge states}  in periodic electromagnetic structures which exhibit the magneto-optic effect was made. In this case, the edge is realized via a domain wall across which 
the Faraday axis is reversed. 
Since the magneto-optic effect breaks time-reversal symmetry, as does the magnetic field in the Hall effect, the resulting edge states are unidirectional.

Other realizations of edges in photonic and electromagnetic systems, {\it e.g.} between periodic dielectric and conducting structures, between periodic structures and free-space, have been explored through experiment and numerical simulation; see, for example \cites{Soljacic-etal:08,Fan-etal:08,Rechtsman-etal:13a,Shvets-PTI:13,Shvets:14}.
     In the context of tight-binding models, the existence and robustness of edge states has been related to topological invariants (Chern index or  Berry / Zak phase \cites{delplace2011zak}) associated with the ``bulk'' (infinite periodic honeycomb) band-structure.
      
We are interested in exploring these phenomena in general energy-conserving wave equations in continuous media. We consider the case of the Schr\"odinger equation  on $\R^2$,  $i\D_t\psi=H\psi$, and  study the existence and robustness of edge states of time-harmonic form: $\psi=e^{-iEt}\Psi$. Our model consists of a honeycomb background potential, the ``bulk'' structure, and a perturbing ``edge-potential''.  The  edge-potential interpolates between two distinct asymptotic periodic structures, via a {\it domain wall} which varies transverse to a specified line-defect (``edge'') in the direction
 of some element of the period lattice, $\Lambda_h$. 
 In the context of honeycomb structures, the most frequently studied edges are the ``zigzag'' and ``armchair'' edges; see Figure  \ref{fig:edges}.

Our model of an edge is motivated by the domain-wall construction of \cites{HR:07, RH:08}. A difference is that 
we break spatial-inversion symmetry, while preserving time-reversal symmetry.  Hence, the edge states -- though topologically robust -- may travel in either direction along the edge.  In \cites{FLW-PNAS:14,FLW-MAMS:15} we proved that a one-dimensional variant of such edge-potentials  gives rise to topologically protected edge states in periodic structures with symmetry-induced linear band crossings, the analogue in one space dimension of Dirac points (see below).  We explore a photonic realization of such states in coupled waveguide arrays  in \cites{Thorp-etal:15}.

Our goal is to clarify the  underlying mechanisms for the existence of topologically protected edge states.  In Theorem \ref{thm-edgestate} we give general conditions for a topologically protected  bifurcation of edge states from {\it Dirac points} of the background (bulk) honeycomb structure.  
The bifurcation is seeded by the robust zero mode of a one-dimensional effective Dirac equation.
A key hypothesis  is a {\it spectral \nofold condition} for the prescribed edge, assumed to be a {\it rational edge}.
In one-dimensional continuum models \cites{FLW-MAMS:15}, this condition is a  consequence of monotonicity properties of dispersion curves.  For continuous $d$-dimensional structures, with $d\ge2$, the spectral \nofold condition may or may not hold; see Section \ref{zz-gap}. Moreover, by varying a parameter, such as the lattice scale of a periodic structure, one can continuously tune between cases where the condition holds or does not hold; see Appendix \ref{V11-section}.
In Theorem \ref{SGC!}  and Theorem \ref{Hepsdelta-edgestates} we verify the spectral \nofold condition for the zigzag edge, for a family of Hamiltonians with weak (low-contrast) potentials, and obtain the existence of {\it zigzag edge states} in this setting.

In a forthcoming article \cites{FLW-sb:16}, we study the strong binding regime (deep potentials) for a large  class of honeycomb Schr\"odinger operators. We prove 
that the two lowest energy dispersion surfaces, after a rescaling by the potential well's depth, converge uniformly to those of the celebrated Wallace (1947) 
 \cites{Wallace:47} tight-binding model of graphite. A corollary of this result is that the spectral no-fold condition, as stated in the present article, 
 is satisfied for sufficiently deep potentials (high contrast) for a very large classes of edge directions in $\Lambda_h$ (including the zigzag edge). In fact, we believe that the analysis of the present article can be extended and together with \cites{FLW-sb:16} will yield the existence of edge states
 which are localized, transverse to {\it arbitrary} edge directions $\vtilde_1\in\Lambda_h$.  This is work in progress.
  For a detailed discussion of examples and motivating numerical simulations, see \cites{FLW-2d_materials:15}.
 
The types of edge states which exist for edges generated by domain walls stand in contrast to those which exist in the case of ``hard edges'', {\it i.e.} edges defined by the tight-binding bulk Hamiltonian on one side of an edge with Dirichlet (zero) boundary condition imposed on the edge; see parenthetical remark in Figure \ref{fig:edges}. In this case, it is well-known that 
zigzag (hard) edges support edge states, while armchair (hard) edges do not support edge states; see, for example, \cites{Graf-Porta:13}. 

Finally, we believe that failure of the spectral \nofold condition implies that there are no topologically protected edge states, although there is evidence that there are meta-stable edge states, which are localized near the edge for a long time; see Section \ref{meta-stable?}.

%
%
%
\begin{figure}
\centering
\includegraphics[width=0.8\textwidth]{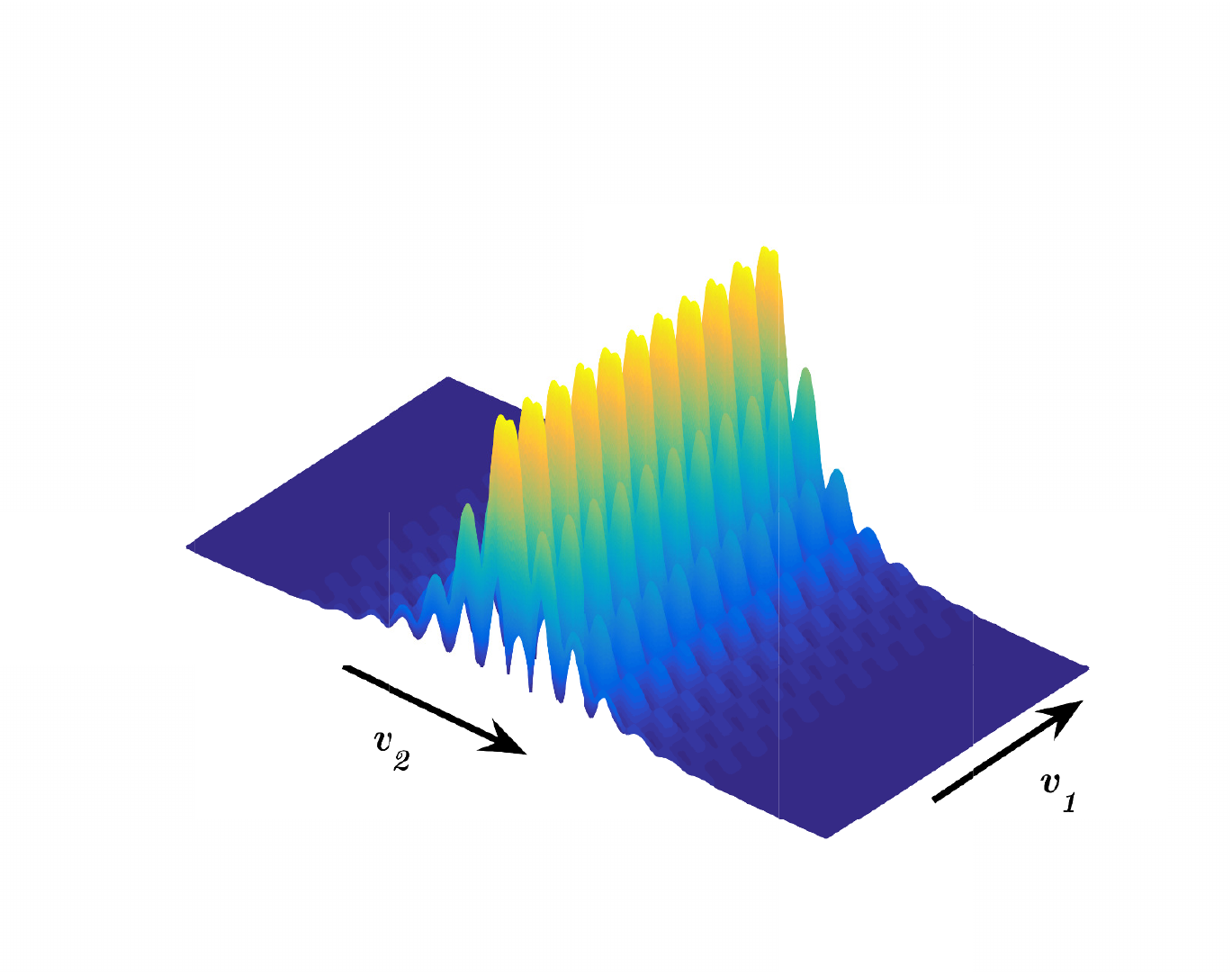}
\caption{\footnotesize
Edge state -- propagating (plane-wave like) parallel to a zigzag edge ($\R\bv_1$) and localized transverse to the edge.
\label{fig:mode_schematic}
}
\end{figure}

\begin{figure}
\centering
\includegraphics[width=0.8\textwidth]{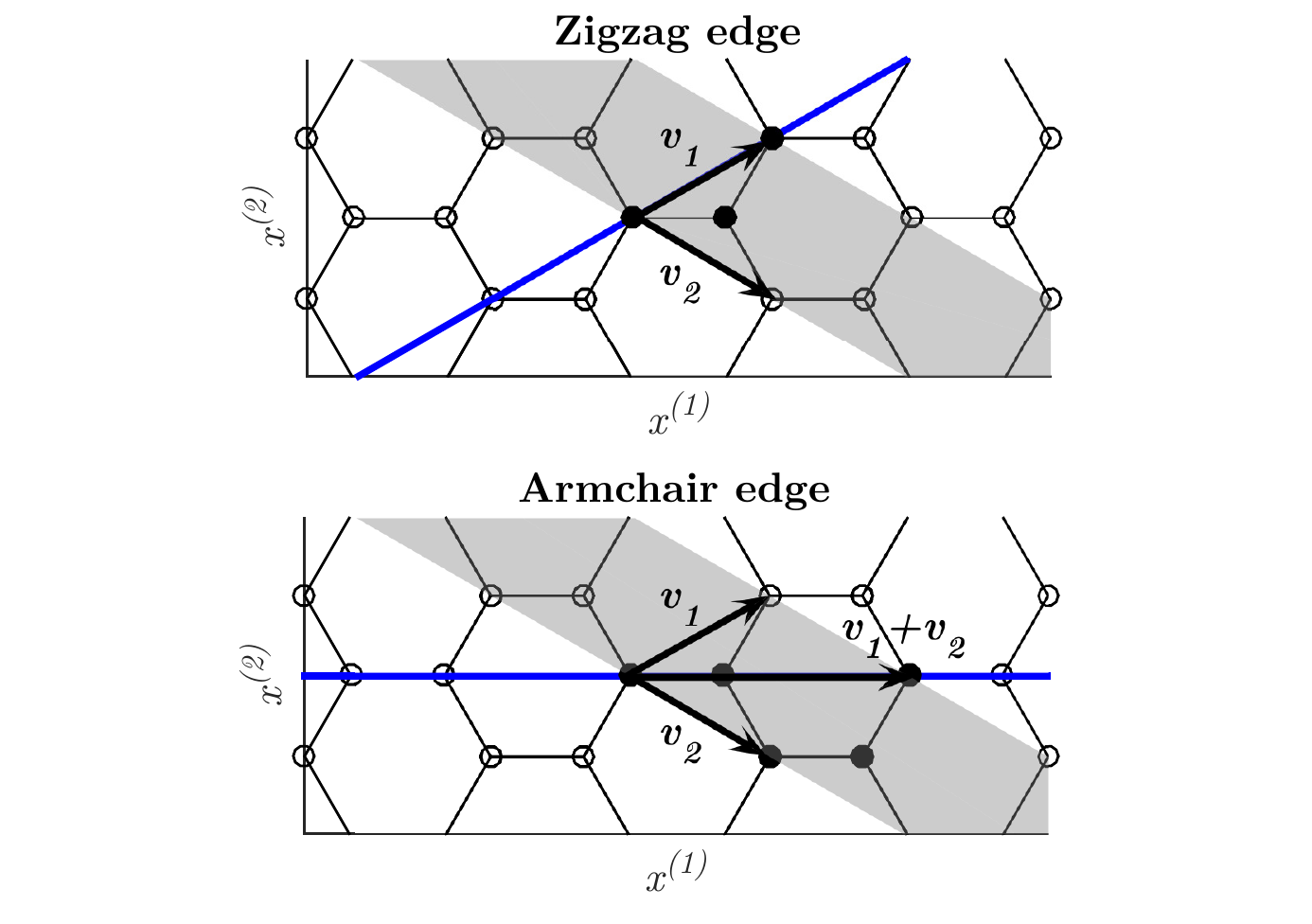}
\caption{\footnotesize
Bulk honeycomb structure, ${\bf H} =  ({\bf A } + \Lambda_h) \cup   ( {\bf B} + \Lambda_h)$.
{\bf Top panel}: Zigzag edge (blue line), $\R\bv_1 = \{\bx : \bk_2\cdot\bx=0\}$. 
Shaded region is the fundamental domain of the cylinder, $\Sigma_{ZZ}$, corresponding to the zigzag edge.
{\bf Bottom panel}: Armchair edge (blue line), $\R\left(\bv_1+\bv_2\right) = \{\bx : (\bk_1-\bk_2)\cdot\bx=0\}$. 
Fundamental domain of the cylinder, $\Sigma_{AC}$, corresponding to the armchair edge, also indicated.
(Darkened vertices are sites at which zero-boundary conditions are imposed in tight-binding models of
 ``hard'' edges.)
\label{fig:edges}
}
\end{figure}

\subsection{Detailed discussion of main results}\label{detailed-intro}

Let $\Lambda_h = \Z\bv_1\oplus \Z\bv_2$ denote the regular (equilateral) triangular lattice
and $\Lambda_h^* = \Z\bk_1\oplus \Z\bk_2$ denote the associated dual lattice, with relations  $\bk_l\cdot\bv_m=2\pi \delta_{lm},\ l,m=1,2$. The expressions for $\bk_l$ and $\bv_m$ are displayed in Section \ref{sec:honeycomb}. The honeycomb structure, ${\bf H}$, is the union of two interpenetrating triangular lattices: ${\bf A} + \Lambda_h$ and ${\bf B} + \Lambda_h$; see Figures \ref{fig:edges} and \ref{fig:lattices}.
\begin{figure}
\centering 
\includegraphics[width=0.8\textwidth]{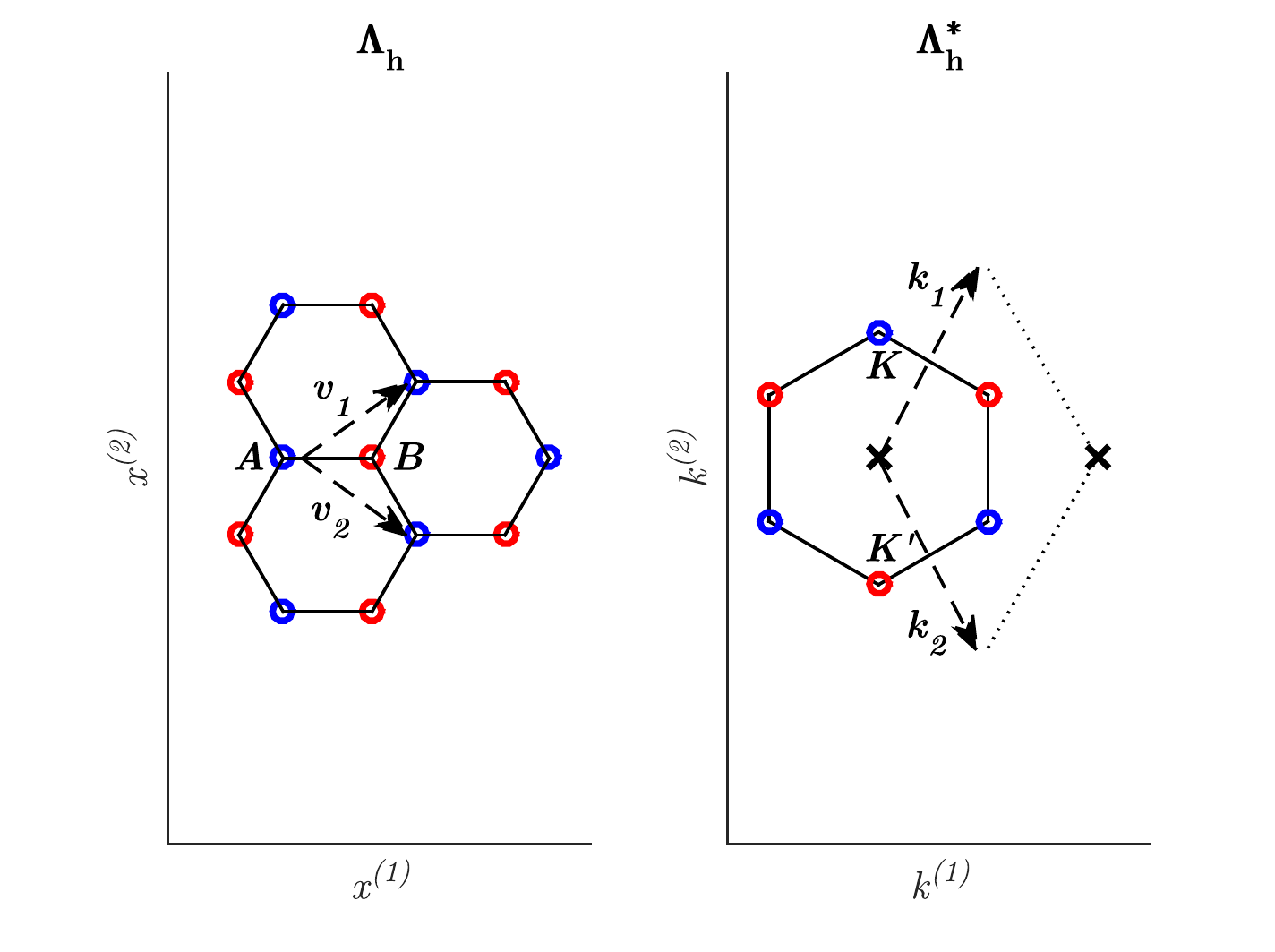}
\caption{\footnotesize
{\bf Left panel:} ${\bf A}=(0,0)$, ${\bf B}=(\frac{1}{\sqrt3},0)$. 
The honeycomb structure, ${\bf H}$  is the union of two interpenetrating sublattices:  $\Lambda_{\bf A}={\bf A}+\Lambda_h$ (blue) 
and $\Lambda_{\bf B}={\bf B}+\Lambda_h$ (red). The lattice vectors  $\{\bv_1,\bv_2\}$ generate $\Lambda_h$. 
Colors designate sublattices; in graphene the atoms occupying $\Lambda_{\bf A}-$ and $\Lambda_{\bf B}-$ sites are identical. 
{\bf Right panel:}
Brillouin zone, $\brill_h$, and dual basis $\{\bk_1,\bk_2\}$. $\bK$ and $\bK'$ are labeled. Other vertices of $\brill_h$ obtained via application of $R$, a rotation by $2\pi/3$.
\label{fig:lattices}
}
\end{figure}

A {\it honeycomb lattice potential}, $V(\bx)$,  is a real-valued,  smooth function, which is $\Lambda_h-$ periodic and, relative to some origin of coordinates,  inversion symmetric  (even) and invariant under a $2\pi/3$ rotation; see Definition \ref{honeyV}.  A choice of period cell is $\Omega_h$, the parallelogram in $\R^2$ spanned by $\{\bv_1, \bv_2\}$.
   
We begin with the Hamiltonian for the unperturbed honeycomb structure:
\begin{align*}
H^{(0)}  &= -\Delta +V(\bx). \label{H0}
\end{align*}
The {\it band structure} of the $\Lambda_h-$ periodic Schr\"odinger operator, $H^{(0)}$, is obtained by considering the
family of eigenvalue problems, parametrized by $\bk\in\mathcal{B}_h$, the Brillouin zone: 
 $(H^{(0)}-E)\Psi=0,\ \Psi(\bx+\bv)=e^{i\bk\cdot\bv}\Psi(\bx),\ \ \bx\in\R^2,\ \bv\in\Lambda_h$.
 Equivalently,   $\psi(\bx)=e^{-i\bk\cdot\bx}\Psi(\bx)$, satisfies the periodic eigenvalue problem:
$ \left(H^{(0)}(\bk)-E(\bk)\right)\psi=0$ and $\psi(\bx+\bv)=\psi(\bx)$ for all $\bx\in\R^2$ and 
$\bv\in\Lambda_h$, where   $H^{(0)}(\bk)=-(\nabla+i\bk)^2+V(\bx)$.
 For each $\bk\in\brill_h$, the spectrum is real and consists of discrete eigenvalues  $E_b(\bk),\ b\ge1, $ where $E_j(\bk)\le E_{j+1}(\bk)$. The maps $\bk\mapsto E_b(\bk)\in\R$ are called the dispersion surfaces of $H^{(0)}$. The collection of these surfaces constitutes the {\it band structure} of $H^{(0)}$. As $\bk$ varies over $\mathcal{B}_h$, each map $\bk\to E_b(\bk)$ is Lipschitz continuous and sweeps out a  closed interval in $\R$. The union of these intervals is the $L^2(\R^2)-$ spectrum of $H^{(0)}$.
A more detailed discussion is presented in Section \ref{honeycomb_basics}.
 
A central role is played by the {\it Dirac points}  of $H^{(0)}$. 
These are quasi-momentum / energy pairs, $(\bK_\star,E_\star)$,  in the band structure of $H^{(0)}$ at which neighboring dispersion surfaces touch conically at a point \cites{RMP-Graphene:09,Katsnelson:12,FW:12}. The existence of Dirac points, located at the six vertices of the Brillouin zone, $\mathcal{B}_h$ (regular hexagonal dual period cell) for generic honeycomb structures was proved in \cites{FW:12,FLW-MAMS:15}; see also \cites{Grushin:09,berkolaiko-comech:15}.
The quasi-momenta of Dirac points partition into two equivalence classes; the $\bK-$ points consisting of $\bK, R\bK$ and $R^2\bK$, where $R$ is a rotation by $2\pi/3$ and  $\bK'-$ points  consisting of $\bK'=-\bK, R\bK'$ and $R^2\bK'$. 
The time evolution of a  wavepacket, with data spectrally localized near a Dirac point, is governed by a massless two-dimensional Dirac system \cites{FW:14}.  
  
\begin{figure}
\centering
\includegraphics[width=0.8\textwidth]{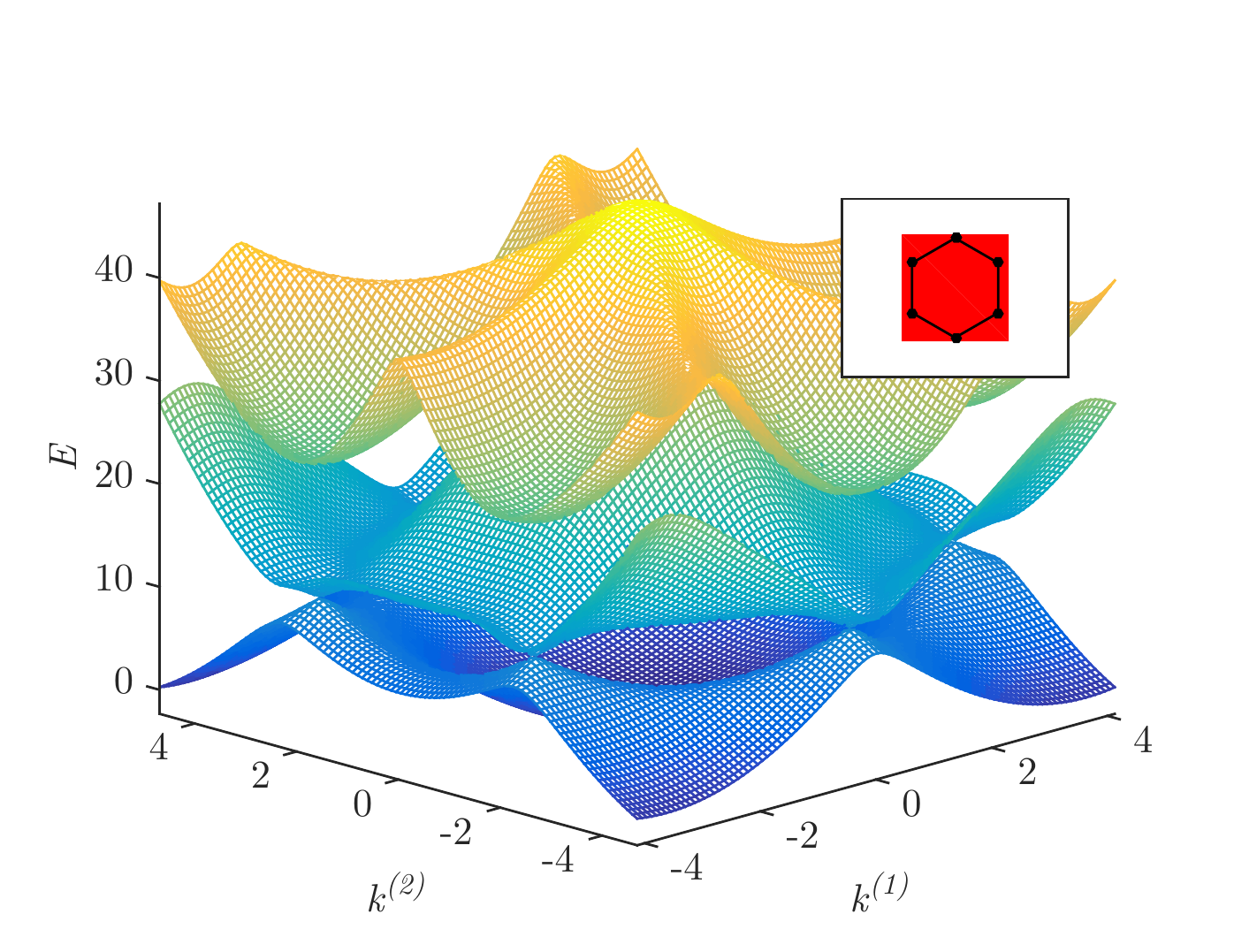}
 \caption{\footnotesize
Lowest three dispersion surfaces $\bk\equiv(k^{(1)},k^{(2)})\in\mathcal{B}_h\mapsto E(\bk)$ of the band structure of  $H^{(0)}\equiv -\Delta + V(\bx)$, where 
$V$ is the honeycomb potential: $V(\bx) = 10 \left(\cos(\bk_1 \cdot\bx)+\cos(\bk_2 \cdot\bx)+\cos((\bk_1+\bk_2)\cdot\bx)\right)$.
Dirac points occur at the intersection of the lower two dispersion surfaces, at the six vertices of the Brillouin zone, $\mathcal{B}_h$. 
\label{fig:E_mesh_3}
}
\end{figure}
  
Figure \ref{fig:E_mesh_3} displays the first three dispersion surfaces of $H^{(0)}$ for a honeycomb potential. The lowest two of these surfaces touch conically at the six vertices of $\mathcal{B}_h$ (inset). Associated with the Dirac point $(\bK_\star,E_\star)$
is a two-dimensional eigenspace of $\bK_\star-$ pseudo-periodic states,\ ${\rm span}\{\Phi_1,\Phi_2\}$:  
\[ H^{(0)}\Phi_j(\bx)=E_\star \Phi_j(\bx),\ \bx\in\R^2,\ \  j=1,2\ ,\ \textrm{ where}\ \  \Phi_j(\bx+\bv)=e^{i\bK_\star\cdot\bv}\Phi_j(\bx),\ \ \bv\in\Lambda_h ;\]
see Definition \ref{dirac-pt-defn}.
It is also shown in \cites{FW:12} that a $\Lambda_h-$ periodic perturbation of $V(\bx)$, which breaks inversion or time-reversal
symmetry lifts the eigenvalue degeneracy; a (local) gap is opened about the Dirac points and the perturbed dispersion surfaces are locally smooth. The perturbation of $H^{(0)}$ by an edge potential (see \eqref{schro-domain}) takes advantage of this instability of Dirac points with symmetry breaking perturbations.
 
To construct our  Hamiltonian, perturbed by an edge-potential, we first choose a vector $\vtilde_1\in\Lambda_h$, the period lattice,  and consider the  line 
$\R\vtilde_1$, the ``edge''. Choose $\vtilde_2$ such that  $\Lambda_h=\Z\vtilde_1\oplus\Z\vtilde_2$. Also introduce dual basis  vectors, $\ktilde_1$ and $ \ktilde_2$,  satisfying  $\ktilde_l\cdot\vtilde_m=2\pi\delta_{lm},\ l,m=1,2$; see Section \ref{ds-slices} for a detailed discussion. 
The choice $\vtilde_1=\bv_1$ (or equivalently $\bv_2$) is  a  {\it zigzag edge} and the choice $\vtilde_1=\bv_1+\bv_2$ is an {\it armchair edge}; see Figure \ref{fig:edges}.
Introduce the perturbed Hamiltonian:
\begin{equation}
H^{(\delta)} \equiv -\Delta + V(\bx) + \delta\kappa(\delta\ktilde_2\cdot \bx)W(\bx) \ =\ H^{(0)} + \delta\kappa(\delta\ktilde_2\cdot \bx)W(\bx) .
\label{schro-domain}
\end{equation}
Here, $\delta$ is real and will be taken to be sufficiently small, and $W(\bx)$ is $\Lambda_h-$ periodic
and odd. The function $\kappa$, defines a {\it domain wall}. We choose $\kappa$ to be sufficiently smooth and to satisfy $\kappa(0)=0$ and    $\kappa(\zeta)\to\pm\kappa_\infty\ne0$
as $\zeta\to\pm\infty$. Without loss of generality, we assume $\kappa_\infty>0$, {\it e.g.} $\kappa(\zeta)=\tanh(\zeta)$. We refer to the line $\R\vtilde_1$ as a $\vtilde_1-$ edge. 

Note that $H^{(\delta)}$ is invariant under translations  parallel to the  $\vtilde_1-$ edge, $\bx\mapsto\bx+\vtilde_1$, and hence there is a well-defined {\it parallel quasi-momentum}, denoted $\kpar$.  Furthermore, 
$H^{(\delta)}$ transitions adiabatically between the  asymptotic Hamiltonian $H_-^{(\delta)}=H^{(0)}\ -\ \delta\kappa_\infty W(\bx)$ as $\ktilde_2\cdot\bx\to-\infty$ to the asymptotic Hamiltonian $H_+^{(\delta)}=H^{(0)}\ +\ \delta\kappa_\infty W(\bx)$
as $\ktilde_2\cdot\bx\to\infty$.  In the case where $\kappa$ changes sign once across $\zeta=0$, the  domain wall modulation of $W(\bx)$ realizes a phase-defect
 across the edge (line-defect) $\R\vtilde_1$. A variant of this construction was used in \cites{FLW-MAMS:15} to insert a phase defect between asymptotic dimer periodic potentials. 
 \bigskip
 \bigskip
 
Suppose $H^{(0)}$ has a Dirac point at $(\bK_\star,E_\star)$. It is important to note that while $H^{(0)}$ is inversion symmetric,  $H^{(\delta)}_\pm$ is not. 
For $\delta\ne0$,  $H^{(\delta)}_\pm$ does not have Dirac points; its dispersion surfaces are locally smooth and for quasi-momenta $\bk$ such that if $|\bk-\bK_\star|$ is sufficiently small, there is an open neighborhood of $E_\star$ not contained in the  $L^2(\R^2/\Lambda_h)-$ spectrum of $H^{(\delta)}_\pm(\bk)$. This  ``spectral gap'' about $E=E_\star$ may however only be local about $\bK_\star$  \cites{FW:12}.  If there is a real open neighborhood of $E_\star$, not contained in the  spectrum of $H^{(\delta)}_\pm(\bk)=-(\nabla+i\bk)^2+V\pm\delta\kappa_\infty W$ for \emph{all} $\bk\in\B_h$, then 
$H_\pm^{(\delta)}$ is said to have a  (global) omni-directional spectral gap about $E=E_\star$.
We'll see, in our discussion of the {\it spectral \nofold condition}, 
 that it is a ``directional spectral gap'' that plays a key role in the existence of edge states;
  see  Section \ref{intro-nofold} and  Definition \ref{SGC}.

Under suitable hypotheses, we shall construct  {\it $\vtilde_1-$ edge states} of $H^{(\delta)}$, which are spectrally localized near the Dirac point, $(\bK_\star,E_\star)$. 
These are non-trivial solutions $\Psi$, with energies $E\approx E_\star$,  of the $\kpar-$ eigenvalue problem:
\begin{align}
 H^{(\delta)}\Psi\ &=\ E\Psi,\label{edge-evp}\\
  \ \Psi(\bx+\vtilde_1)\ &=\ e^{i\kpar}\Psi(\bx)\  (\textrm{propagation parallel to}\ \R\vtilde_1), \label{edge-bc1}\\
|\Psi(\bx)|\ &\to\ 0,\ \ {\rm as}\ \ |\ktilde_2\cdot\bx|\to\infty\quad  (\textrm{localization transverse to}\ \R\vtilde_1), \label{edge-bc2}
\end{align}
for $\kpar\approx\bK_\star\cdot\vtilde_1$. 
To formulate the eigenvalue problem in an appropriate Hilbert space, we introduce the cylinder $\Sigma\equiv \R^2/ \Z\vtilde_1$. If $f(\bx)$ satisfies the pseudo-periodic boundary condition \eqref{edge-bc1}, then $f(\bx)e^{-i\frac{\kpar}{2\pi}\ktilde_1\cdot\bx}$ is well-defined on the cylinder $\Sigma$.  Denote by  $H^s(\Sigma),\ s\ge0$, the  Sobolev spaces of functions defined on $\Sigma$. The pseudo-periodicity and decay conditions \eqref{edge-bc1}-\eqref{edge-bc2} are encoded by requiring $ \Psi \in H^s_\kpar(\Sigma)$, for some $s\ge0$,  where
\begin{equation*}
  H^s_\kpar=H^s_\kpar(\Sigma)\ \equiv \ \left\{f : f(\bx)e^{-i\frac{\kpar}{2\pi}\ktilde_1\cdot\bx}\in H^s(\Sigma) \right\} .\label{Hs-kpar}
  \end{equation*}
Thus we formulate the EVP \eqref{edge-evp}-\eqref{edge-bc2} as:
\begin{equation}
H^{(\delta)}\Psi\ =\  E\Psi,\ \ \Psi\in H^2_{\kpar}(\Sigma).
\label{EVP}\end{equation}

\begin{remark}[Symmetry relation among $\bK-$ and $\bK'-$ points] 
Note that if $\Psi(\bx)=e^{i\bk\cdot\bx}Z(\bx)$ is a solution of the eigenvalue problem \eqref{EVP}, then $\psi_{\bK}=e^{-i(E t-\bK\cdot\bx)} Z(\bx),$
where $Z(\bx+\vtilde_1)=Z(\bx)$ and $Z(\bx)\rightarrow0$ as $|\ktilde_2\cdot\bx| \rightarrow \infty$,
is a propagating edge state of the time-dependent Schr\"odinger equation:
$ i\D_t\psi(\bx,t) = H^{(\delta)} \psi(\bx,t)$ 
with parallel quasi-momentum $\kpar=\bK\cdot\vtilde_1$.
Since the time-dependent Schr\"odinger equation has the invariance $\psi(\bx,t)\mapsto\overline{\psi(\bx,-t)}$, it follows  that
\[  \overline{\psi_\bK(\bx,-t)} = e^{-i(Et+\bK\cdot\bx)} \overline{Z(\bx)} = e^{-i(Et-\bK'\cdot\bx)} \overline{Z(\bx)} =  \psi_{\bK'}(\bx,t) . \]
Thus $\psi_{\bK'}(\bx,t)$ is a  counterpropagating edge state with parallel quasi-momentum, $k_\parallel=\bK'\cdot\vtilde_1=-\bK\cdot\vtilde_1$.
Due to these  symmetry considerations and the equivalence of $\bK-$ points: $\{\bK,R\bK,R^2\bK\}$, without loss of generality, we henceforth  restrict our attention to the Dirac point $(\bK,E_\star)$.
\end{remark}

\subsection{Summary of main results}\label{results-summary}

\subsubsection{General conditions for the existence of topologically protected edge states; Theorem \ref{thm-edgestate} and Corollary \ref{vary_k_parallel}} 

In {\bf Theorem \ref{thm-edgestate}} we formulate hypotheses  on the honeycomb potential,  $V$, domain wall function, $\kappa(\zeta)$, and asymptotic periodic structure, $W(\bx)$, which
imply  the existence of  topologically protected $\vtilde_1-$ edge states, constructed as   non-trivial  eigenpairs $\delta\mapsto (\Psi^\delta, E^\delta)$ of \eqref{EVP} with $\kpar=\bK\cdot\vtilde_1$, 
defined for all $|\delta|$ sufficiently small. This branch of non-trivial  states    bifurcates  from the trivial solution branch $E\mapsto(\Psi\equiv0,E)$ at $E=E_\star$, the energy of the Dirac point. 
Key among the hypotheses is the spectral \nofold condition, discussed below in Section \ref{intro-nofold}.
At leading order in $\delta$, the edge state, $\Psi^\delta(\bx)$, is a slow modulation of the degenerate nullspace of $H^{(0)}-E_\star$:
\begin{align}
  \Psi^\delta(\bx) &\approx \alpha_{\star,+}(\delta\ktilde_2\cdot\bx)\Phi_+(\bx) + \alpha_{\star,-}(\delta\ktilde_2\cdot\bx)\Phi_-(\bx) \ \ \text{in} \ \ H_{\kpar=\bK\cdot\vtilde_1}^2(\Sigma) , \label{multiscale-formal0} \\
  E^\delta &= E_\star + \mathcal{O}(\delta^2),\ \ 0<|\delta|\ll1,\label{multiscale-formal1}
\end{align} 
where $\Phi_+$ and $\Phi_-$ are the appropriate linear combinations of $\Phi_1$ and $\Phi_2$,
defined in \eqref{Phi_pm-def}.
 The envelope amplitude-vector, $\alpha_\star(\zeta)=(\alpha_{\star,+}(\zeta),\alpha_{\star,-}(\zeta))^T$,  is a zero-energy eigenstate, $\mathcal{D}\alpha_\star=0$,  of the one-dimensional Dirac operator (see also \eqref{multi-dirac-op}): 
 \[ \mathcal{D} \equiv -i|\lambda_\sharp||\ktilde_2|\sigma_3\D_\zeta + \vartheta_\sharp\kappa(\zeta)\sigma_1,\] 
 where the Pauli matrices $\sigma_j$ are displayed in \eqref{Pauli-sigma}.
Here $\lambda_\sharp\in\C$ (see \eqref{lambda-sharp2}) depends on the unperturbed honeycomb potential, $V$, and is non-zero for generic $V$. The constant
  $\vartheta_\sharp\equiv\left\langle \Phi_1,W\Phi_1\right\rangle_{L^2(\Omega_h)}$ is real and is also generically nonzero.  
$\mathcal{D}$ has a spatially localized zero-energy eigenstate for any $\kappa(\zeta)$ having asymptotic limits of opposite sign at $\pm\infty$. Therefore, the zero-energy eigenstate, which seeds the bifurcation, persists   for   {\it localized} perturbations of $\kappa(\zeta)$.
In this sense, the bifurcating branch of edge states is topologically protected against a class of local perturbations of the edge.

Section \ref{formal-multiscale} gives an account of a formal multiple scale expansion, to any order in the small parameter, $\delta$, of a solution to the eigenvalue problem \eqref{EVP}. The expression in  \eqref{multiscale-formal0} is the leading order term in this expansion. Our methods can be used to prove the validity of the multiple scale expansion, at any finite order.

{\bf  Corollary \ref{vary_k_parallel}} ensures, under the conditions of  Theorem \ref{thm-edgestate}, the existence of edge states, 
$\Psi(\bx;\kpar)\in H^2_{\kpar}(\Sigma)$ for all  $\kpar$ in a neighborhood of $\kpar=\bK\cdot\vtilde_1$, and by symmetry (see Remark \ref{kpar-symmetry}) for all $\kpar$ in a neighborhood of $\kpar=-\bK\cdot\vtilde_1=\bK'\cdot\vtilde_1$.
Thus,  
by taking a continuous superposition of states given by Corollary \ref{vary_k_parallel}, one obtains states that remain localized about (and dispersing along) the zigzag edge for all time.

\begin{remark}\label{key-hyp}
A key hypothesis in Theorem \ref{thm-edgestate}  is a  {\it spectral  no-fold} condition at $(\bK,E_\star)$ for the $\vtilde_1-$ edge of the band-structure of $-\Delta+V$. This (essentially) ensures  the existence of a $L^2_{\kpar=\bK\cdot\vtilde_1}(\Sigma)-$ spectral gap containing $E_\star$ for the perturbed  Hamiltonian, $H^{(\delta)}$; see Definition \ref{SGC} and the discussion in Section \ref{intro-nofold}. 
\end{remark}

\subsubsection{Theorem \ref{Hepsdelta-edgestates}; Existence of topologically protected zigzag edge states}\label{zigzag-summary}

We consider the case of zigzag edges corresponding to the choice $\vtilde_1=\bv_1$, $\vtilde_2=\bv_2$, and $\ktilde_1=\bk_1$, $\ktilde_2=\bk_2$.  Recall that  $\Lambda_h=\Z\bv_1\oplus\Z\bv_2$. The choice $\vtilde_1=\bv_2$ would lead to equivalent results.

We consider the zigzag edge state eigenvalue problem
\begin{equation}
H^{(\eps,\delta)}\Psi\ =\  E\Psi, \quad \Psi\in H^2_{\kpar}(\Sigma) \qquad \text{(see also \eqref{EVP})} ,
\label{EVP-1}\end{equation}
with Hamiltonian
\begin{equation}
H^{(\eps,\delta)} \equiv -\Delta + \eps V(\bx) + \delta\kappa(\delta\bk_2\cdot \bx)W(\bx) \ =\ H^{(\eps)} + \delta\kappa(\delta\bk_2\cdot \bx)W(\bx) .
\label{Ham-ZZ}
\end{equation}
Here, $\eps$ and $\delta$ are chosen to satisfy
\begin{equation}
0<|\delta|\lesssim \eps^2 \ll1 .
\label{small-eps-delta}\end{equation}

  There are two cases, which are delineated by the sign of the distinguished Fourier coefficient, $\eps V_{1,1}$, of the unperturbed (bulk) honeycomb potential, $\eps V(\bx)$. Here,
  \begin{equation*}
V_{1,1}\ \equiv\ 
\frac{1}{|\Omega_h|} \int_{\Omega_h} e^{-i(\bk_1+\bk_2)\cdot\by}\ V(\by)\ d\by,
\label{V11eq0-intro}
\end{equation*}
is assumed to be non-zero. We designate these cases:
\[ \textrm{\bf Case (1)}\qquad  \eps V_{1,1}>0\ \ \ \textrm{ and}\ \ \   \textrm{\bf Case (2)}\qquad \eps V_{1,1}<0.\]
 In Appendix \ref{V11-section} we give two explicit families of potentials, a superposition of ``bump-functions'' concentrated, respectively, on a triangular lattice, $\Lambda_h^{(a)}$, and a  honeycomb structure, ${\bf H}$, that  can be tuned between these two cases 
by variation of a lattice scale parameter.

Under the condition $\eps V_{1,1}>0$ (Case (1)) and \eqref{small-eps-delta}, we verify the spectral \nofold condition for the zigzag edge in {\bf Theorem \ref{SGC!}}. The existence of zigzag edge states 
({\bf Theorem \ref{Hepsdelta-edgestates}}) then follows from Theorem \ref{thm-edgestate} and Corollary \ref{vary_k_parallel}. In particular,  for all $\eps$ and $\delta$ satisfying \eqref{small-eps-delta} and for each $\kpar$ near $\bK\cdot\bv_1=2\pi/3$, 
the zigzag edge state eigenvalue problem \eqref{EVP}
has topologically protected edge states with energies sweeping out a neighborhood of $E_\star^\eps$, where $(\bK,E_\star^\eps)$ is a Dirac point. 

\begin{remark}[Directional versus omnidirectional spectral gaps]\label{BS-conj}
 While the regime of weak potentials, implied by \eqref{small-eps-delta},  would at first seem to be a simplifying assumption, we wish  to remark on a subtlety for $H^{(\eps,\delta)}_\pm = -\Delta+\eps V \pm \delta\kappa_\infty W$ ($\eps, \delta$ small), which arises precisely in this regime.     It is well-known that for sufficiently weak periodic potentials  on $\R^d, d\ge2$, that there are no spectral gaps; this is related to the  
 ``Bethe-Sommerfeld conjecture'' \cites{Bethe-Sommerfeld:33,Skriganov:79,Dahlberg-Trubowitz:82}. Nevertheless,
 if  $\eps V_{1,1}>0$, and $\eps$ and $\delta$ are related as in  \eqref{small-eps-delta},  then a {\it directional} spectral gap, {\it i.e.} an $L_\kpar^2(\Sigma)-$ spectral gap  exists; see Theorem \ref{delta-gap} and Section \ref{intro-nofold}. 
 \end{remark}

Figure \ref{fig:spectra_vary_delta}
 and  Figure \ref{fig:k_parallel3} are illustrative of Cases (1) and (2).
The  simulations were done for the Hamiltonian
 $H^{(\eps,\delta)}$ with $\eps=\pm10$ and $0\le\delta\le10$:
\begin{equation}
\label{VW-numerics}
\begin{split}
H^{(\eps,\delta)}&= -\Delta +\eps V(\bx) + \delta\kappa(\delta\bk_2\cdot\bx)W(\bx),\ \ \kappa(\zeta)=\tanh(\zeta),\\
 V(\bx)&= \sum_{j=0}^2\cos(R^j\bk_1 \cdot\bx),\ \ 
W(\bx)= \sum_{j=0}^2(-1)^{\delta_{j2}}\sin(R^j\bk_1 \cdot\bx).
\end{split}
\end{equation}
Here, $R$ is the $2\pi/3-$ rotation matrix displayed in \eqref{Rdef}.
Figure \ref{fig:spectra_vary_delta} displays, for fixed $\eps$,  the $L^2_{\kparpi}(\Sigma)-$ spectra (plotted horizontally) of $H^{(\eps,\delta)}$ corresponding to a  range of  $\delta$ values (strength / scale of domain wall -perturbation) for 
Cases (1) $\eps V_{1,1}>0$ (top panel) and (2) $\eps V_{1,1}<0$ (middle and bottom panels). 
Figure \ref{fig:k_parallel3} displays, for these cases, the $L^2_{\kpar}(\Sigma)-$ spectra (plotted vertically) 
for a range of parallel-quasi-momentum, $\kpar$.  

\begin{remark}[Symmetries of $\kpar\mapsto E(\kpar)$] \label{kpar-symmetry}
 Figure \ref{fig:k_parallel3} exhibits some elementary symmetries.
Since the boundary condition for the EVP \eqref{EVP-1}, $\Psi(\bx+\bv_1)=e^{i\kpar}\Psi(\bx)$ 
is   $2\pi-$ periodicity in $\kpar$, the mapping  $k_\parallel\mapsto E(k_\parallel)$ is  $2\pi-$ periodic. Furthermore, invariance under complex conjugation, implies symmetry of  $k_\parallel\mapsto E(k_\parallel)$ about $\kpar=0$ and $\kpar=\pi$.
\end{remark}

\subsubsection{Non-topologically protected bifurcations of edge states}\label{unprotected}

In  Case (2), where $\eps V_{1,1}<0$, {\bf Theorem \ref{NO-directional-gap!}} implies that  the spectral \nofold condition fails and we do not obtain a bifurcation from the Dirac point.
 However, through a combination of formal asymptotic analysis and numerical computations, we do find bifurcating branches of  edge states. These branches do not emanate  from Dirac points (the \nofold condition fails), but rather from a spectral band edge. Moreover, as we discuss below, these states are  \underline{not} topologically protected; they may be destroyed by an appropriate localized perturbation of the edge.  Case (2) ($\eps V_{1,1}<0)$ is illustrated by Figures \ref{fig:spectra_vary_delta} (middle and bottom panels) and Figure \ref{fig:k_parallel3} (bottom panel).  

\begin{figure}
\centering
\includegraphics[width=0.8\textwidth]{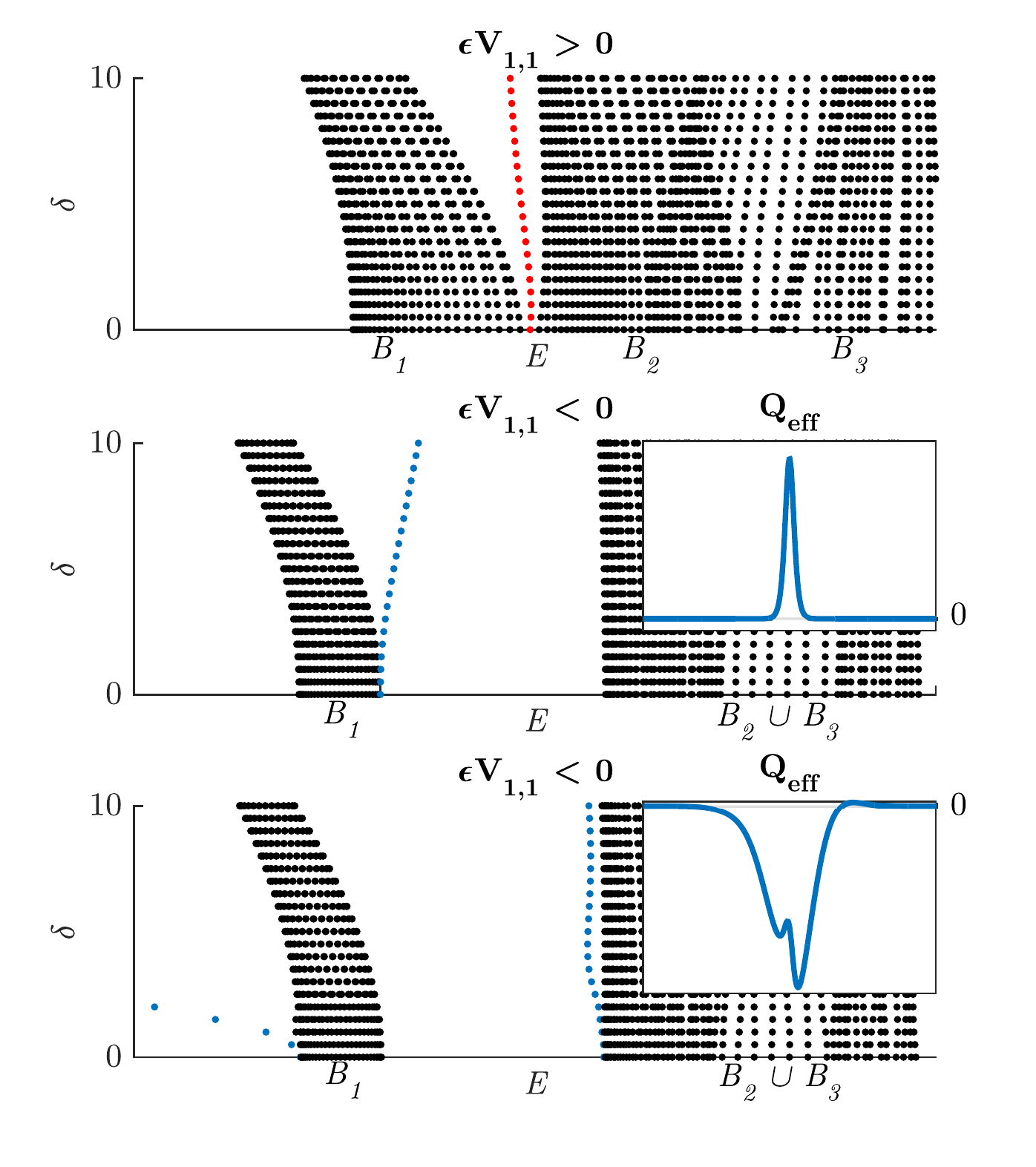}
\caption{\footnotesize
$L^2_{\kpar=\bK\cdot\bv_1}(\Sigma)-$ spectra, where $\bK\cdot\bv_1=\frac{2}{3}\pi$,
of the Hamiltonian $H^{(\eps,\delta)}$ (\eqref{VW-numerics}) for the zigzag edge ($\R\bv_1$). 
 {\bf Top panel:} Case (1) $\eps V_{1,1}>0$. Topologically protected bifurcation of  edge states, described by Theorem \ref{SGC!} (dotted red curve), is seeded by zero-energy mode of a Dirac operator \eqref{multi-dirac-op}. The branch of edge states emanates from intersection of first and second bands ($B_1$ and $B_2$) at $E=E_\star^\eps$ for $\delta=0$; see discussion in Section \ref{zigzag-summary}.
   {\bf Middle panel:}  Case (2) $\eps V_{1,1}<0$ with domain wall function $\kappa$. Spectral \nofold condition does not hold.  
Bifurcation of zigzag edge states from upper endpoint, $E=\widetilde{E}^\eps$, of the first spectral band. This bifurcation is seeded by a bound state of a Schr\"odinger operator \eqref{effective-schroedinger} with effective mass $m_{\rm eff}<0$
 and effective potential $Q_{\rm eff}(\zeta)$ (displayed in the inset) and is  \emph{not} topologically protected; see discussion in Section \ref{unprotected}.
 {\bf Bottom panel:}  Case (2) $\eps V_{1,1}<0$ with domain wall function $\kappa_\natural$. Bifurcation from upper endpoint of $B_1$ is destroyed. Bound states bifurcate from the lower edges of the first two spectral bands.
\label{fig:spectra_vary_delta}
}
\end{figure}

In particular, Dirac points occur at the intersection of the second and third spectral bands of $H^{(\eps,0)}=-\Delta+\eps V(\bx)$ (see Theorem \ref{diracpt-small-thm}),  and the failure of the spectral \nofold condition implies that an $L^2_\kpar-$ spectral gap does not open about $E=E_\star^\eps$  for $\delta\ne0$ and small. However, for $\eps V_{1,1}<0$ there is a spectral gap between the first and second spectral bands of $H^{(\eps,0)}$. For the choice of edge-potential   
 displayed in \eqref{VW-numerics} with $\eps=-10$,  a family of nontrivial edge states bifurcates,  for  $0<|\delta|$ sufficiently small, from the upper edge of the first (lowest) $L^2_{\kparpi}-$ spectral band into the spectral gap (dotted blue curve); see middle panel of Figure \ref{fig:spectra_vary_delta}.  A bifurcation of a similar nature is discussed in \cites{plotnik2013observation}.
 
A formal multiple scale analysis clarifies this latter bifurcation. 
For $\bk\in\brill_h$, let $(\widetilde{E}^\eps(\bk),\widetilde{\Phi}^\eps(\bx;\bk))$ denote the eigenpair associated with a lowest spectral band.
In \cites{FLW-2d_materials:15}, We calculate that the edge state bifurcation is seeded by a discrete eigenvalue effective Schr\"odinger operator: 
\begin{equation}
H^\eps_{\rm eff}= -\frac{1}{2m^\eps_{\rm eff}}\ \frac{\D^2}{\D\zeta^2}\ +\ Q^\eps_{\rm eff}(\zeta;\kappa),\ \ {\rm where}\ \   \frac{1}{m^\eps_{\rm eff}}=\  \sum_{i,j=1,2} [ D^2\widetilde{E}^\eps(\bK)]_{ij} \ \smallktilde_2^i\ \smallktilde_2^j,\label{effective-schroedinger}\end{equation}
and $Q_{\rm eff}(\zeta;\kappa,\widetilde{\Phi^\eps}) = a\ \kappa'(\zeta) + b\ \left(\kappa^2_\infty-\kappa^2(\zeta) \right)$ is a spatially localized effective potential, depending on $\kappa(\zeta)$, and constants $a$ and $b$, with  $b>0$, which depend on $V$, $W$ and $\widetilde{\Phi}^\eps$.
For the above choice of the zigzag edge-potential (middle panel of Figure \ref{fig:spectra_vary_delta}), we have $m^\eps_{\rm eff}<0$ and the effective potential $Q^\eps_{\rm eff}$, displayed in the figure inset, induces a bifurcation into the gap above the first band.

Now, we can construct domain wall functions,  $\kappa_{_\natural}(\zeta)$, for which the corresponding $H^\eps_{\rm eff}$ has no point eigenvalues in a neighborhood of the right (upper) edge of the first spectral band; see bottom panel of Figure \ref{fig:spectra_vary_delta}.
If $\kappa(\zeta)$ is chosen as above, then 
$Q_{\rm eff}(\zeta; (1-\theta)\kappa+\theta\kappa_{_\natural})$, $0\leq\theta\leq1$, provides a smooth homotopy from a Schr\"odinger Hamiltonian for which there is a bifurcation of edge states ($H^{(\eps,\delta)}$ with domain wall $\kappa$) to one for which the branch of edge states does not exist ($H^{(\eps,\delta)}$ with domain wall $\kappa_\natural$). Therefore, this type of bifurcation is not topologically protected; see \cites{FLW-2d_materials:15} for a more detailed discussion.
 This contrast between topologically protected  states and non-protected states is explained 
  and explored numerically, in a one-dimensional setting in \cites{Thorp-etal:15}.

\begin{figure}
\centering
\includegraphics[width=0.8\textwidth]{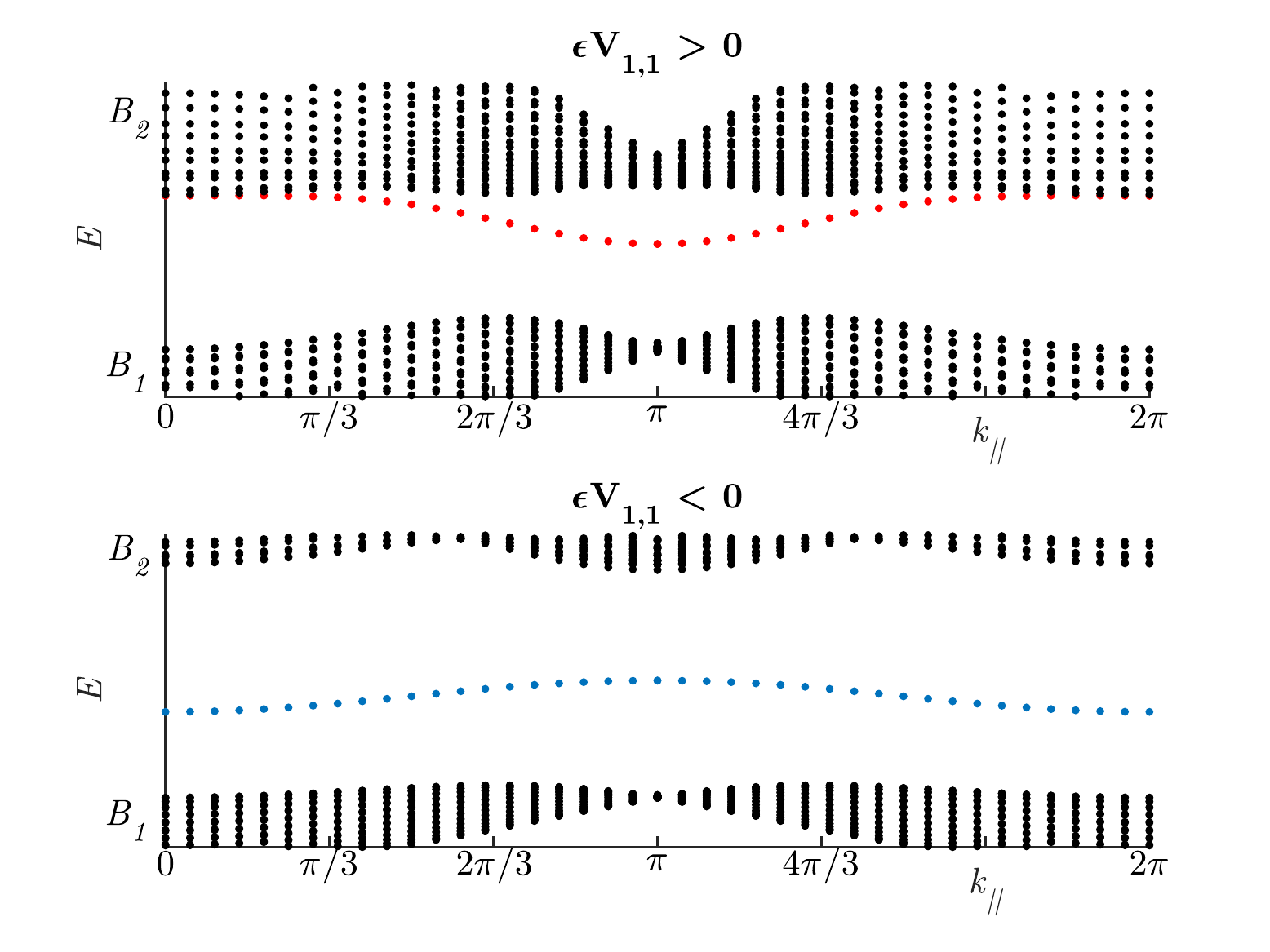}
 \caption{\footnotesize
 {\bf Top panel:} $L^2_{\kpar}(\Sigma)-$ spectrum of protected states of $H^{(\eps,\delta)}$,  for the case $\eps V_{1,1}>0$. {\bf Bottom panel:}  $L^2_{\kpar}(\Sigma)-$ spectrum of non-protected states of $H^{(\eps,\delta)}$ for the case $\eps V_{1,1}<0$.  $V$, $W$ and $\kappa$ are chosen as in \eqref{VW-numerics}.
 For each fixed $\kpar$, edge states shown in the top panel ($\eps V_{1,1}>0$) arise due to a protected bifurcation from a Dirac point displayed in the top panel of Figure \ref{fig:spectra_vary_delta}. Those edge states indicated in the bottom panel ($\eps V_{1,1}<0$) arise via an edge bifurcation of the type shown in the middle and bottom panels of Figure \ref{fig:spectra_vary_delta}. The band edge energies from which this latter bifurcation takes place is well-separated
 from the energy of
 the Dirac point which, when $\eps V_{1,1}<0$,  lies within the overlap of the second and third spectral bands.
 }
 \label{fig:k_parallel3}
\end{figure}

\subsection{Remarks on the spectral \nofold condition}\label{intro-nofold} 

 The spectral \nofold hypothesis of Theorem \ref{thm-edgestate}  requires that the dispersion curves obtained by slicing the band structure (situated in $\R^2_\bk\times\R_E$) with a plane through the Dirac point $(\bK,E_\star)$  containing the direction $\ktilde_2$ (dual direction to the $\vtilde_1-$ edge) do not fold-over
 and fill out energies arbitrarily near $E_\star$.
 This  essentially implies that via a small perturbation which breaks inversion symmetry (as we do with  $H^{(\delta)}=-\Delta+V(\bx)+\delta\kappa(\delta\ktilde_2\cdot\bx)W(\bx)$ for $\delta\ne0$)  we open a $L^2_\kpar(\Sigma)-$ spectral gap about $E_\star$.
  Figure \ref{fig:eps_V11_neg} is illustrative.
 %
 %

In the  first row of plots in Figure \ref{fig:eps_V11_neg}, we consider whether the spectral \nofold condition holds at the Dirac point $(\bK,E_\star^\eps)$
 for the zigzag edge, in the two cases: (1) $\eps V_{1,1}>0$ and (2) $\eps V_{1,1}<0$, as well as for the armchair edge. 
The energy level $E=E_\star^\eps$ is indicated with the dotted line. In the left panel we see that for the zigzag edge, the spectral \nofold condition holds if $\eps V_{1,1}>0$. In this case, there is a topologically protected branch of edge states. In the center panel we see that the spectral \nofold condition fails if $\eps V_{1,1}<0$.
Finally, in the right panel we see that it also fails for the armchair slice. 
  
The   second row of plots  in Figure \ref{fig:eps_V11_neg},  illustrates that  the spectral \nofold condition controls whether a full $L^2_\kpar-$ spectral gap opens when breaking inversion symmetry.
In particular, for $\delta>0$, $H^{(\eps,\delta)}$ is no longer inversion symmetric. For $\eps V_{1,1}>0$, a spectral gap opens about the Dirac point, between the first and second spectral bands (see Theorem \ref{diracpt-small-thm}). For the zigzag edge with $\eps V_{1,1}<0$ there is no spectral gap about the Dirac point. 
(Note, however, that there is a spectral gap between the first and second spectral bands; see the discussion above in Section \ref{unprotected}.)  
 Similarly, for the armchair edge (right panel) there is no spectral gap for $\delta>0$.
 
\begin{figure}
\centering
\includegraphics[width=0.8\textwidth]{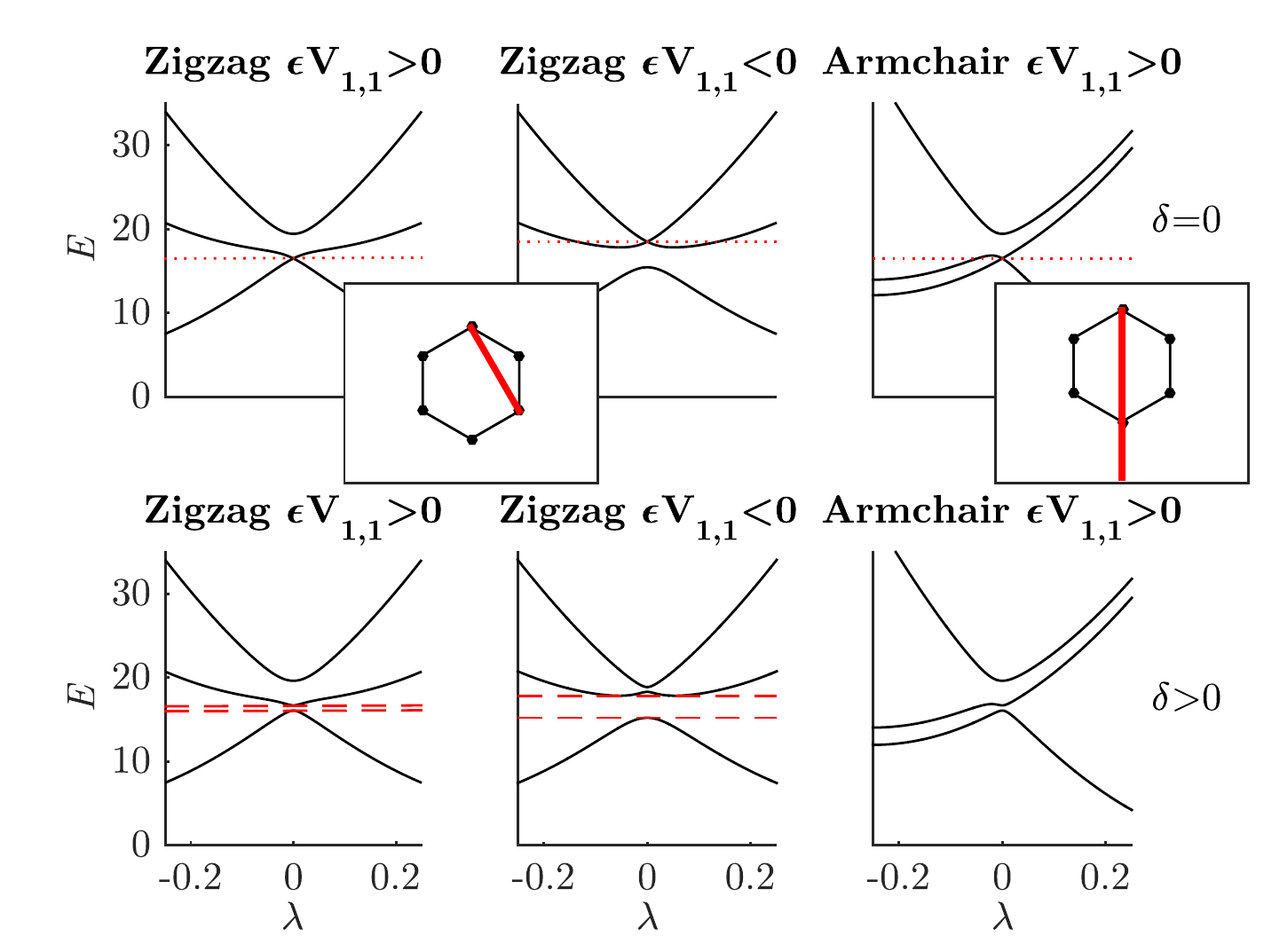}
 \caption{\footnotesize
 Zigzag and armchair slices at the Dirac point $(\bK,E^\eps_\star)$ of the band structure of $-\Delta+\eps V+\delta\kappa_\infty W$ for $\delta=0$ ({\bf first row}) and $\delta>0$ ({\bf second row}). Insets indicate zigzag and armchair quasi-momentum segments (one-dimensional Brillouin zones) parametrized by $\lambda$, for $0\leq\lambda\leq1$. See discussion of Section \ref{meta-stable?} and Theorem \ref{fourier-edge}.
 \label{fig:eps_V11_neg}
 }
\end{figure}

\subsection{Are there meta-stable edge states?}\label{meta-stable?}

Consider the Hamiltonian
 $ H^{(\delta)} = -\Delta_\bx +  V(\bx) + \delta\kappa\left(\delta\ktilde_2\cdot\bx\right) W(\bx)$ (as in  \eqref{schro-domain}), corresponding to an {\it arbitrary} rational edge, $\R\vtilde_1$, {\it i.e.} $\vtilde_1=a_1\bv_1+b_1\bv_2$, $a_1$ and $b_1$ co-prime integers, as introduced in the discussion leading up to  \eqref{schro-domain}; see also Section \ref{ds-slices}. 
Irrespective of whether the spectral \nofold condition holds for the $\vtilde_1-$ edge (see Section \ref{intro-nofold}  and Definition \ref{SGC}),  the multiple scale expansion of Section \ref{formal-multiscale} produces a formal edge state to {\it  any finite order} in the small parameter $\delta$. 

 \nit  {\it But is this formal expansion the expansion of a true edge state?} 
We believe the answer is no, if the spectral \nofold condition fails. 

Indeed,  from Theorem \ref{fourier-edge}, we have that any $\vtilde_1-$  edge state, $\Psi\in L^2_{\kpar=\bK\cdot\vtilde_1}$, is a superposition of Floquet-Bloch modes  of $H^{(0)}=-\Delta+V$ along the quasimomentum segment:
 $\bK+\lambda\ktilde_2,\ |\lambda|\le1/2$.
The formal expansion of Section \ref{formal-multiscale} however is spectrally concentrated on Floquet-Bloch components along this segment, which are near the Dirac point, corresponding to $|\lambda|\ll1$. If the spectral \nofold condition fails, the expansion does not capture the effect of resonant coupling to quasi-momenta along this segment ``far from $\bK$'' (corresponding to $\lambda$ bounded away from $\lambda=0$ in Figure \ref{fig:eps_V11_neg}).
\medskip
  
\nit{\bf Conjecture:}
    {\sl Suppose the spectral \nofold condition fails for the $\vtilde_1-$ edge $\R\vtilde_1$. Then, $H^{(\delta)}$ 
has topologically protected long-lived (meta-stable) edge quasi-modes, $\Psi\in H^2_{\kpar=\bK\cdot\vtilde_1,{\rm loc}}(\Sigma)$, but generically has no topologically protected edge states.}
%
%
%

%

\subsection{Outline}
{\ }

In {\bf Section \ref{honeycomb_basics}} we review spectral theory for two-dimensional periodic Schr\"odinger operators, introduce the triangular lattice, the honeycomb structure and honeycomb lattice potentials.

In {\bf Section \ref{sec:dirac-pts}} we define Dirac points and review the results on the existence of Dirac points for generic honeycomb potentials from \cites{FW:12,FW:14}.

In {\bf Section \ref{ds-slices}} we introduce the notion of an edge or line defect in a bulk (unperturbed) honeycomb structure. Honeycomb structures with edges parallel to a period lattice direction, have a translation invariance.
Thus, an important tool  is the Fourier decomposition of states which are $L^2$ (localized) in the unbounded direction, transverse to the edge, and propagating (plane-wave like) parallel to the edge.

In {\bf Section \ref{zigzag-edges}} we introduce our class of Hamiltonians, consisting of a bulk honeycomb potential, perturbed by a general  line-defect / $\vtilde_1-$ edge potential.

In {\bf Section \ref{formal-multiscale}} we give a  formal multiple scale construction of  edge states to any finite order in the small parameter $\delta$.


In {\bf Section \ref{thm-edge-state}} we formulate general hypotheses which imply the existence of a branch of topologically protected $\vtilde_1-$ edge states, bifurcating from the Dirac point.   The proof uses a  Lyapunov-Schmidt reduction strategy, applied to a system for the Floquet-Bloch amplitudes which is equivalent to the eigenvalue problem. Such a strategy was implemented in a 1D setting in \cites{FLW-MAMS:15}. First, the edge-state eigenvalue problem is formulated in (quasi-) momentum space as an infinite system for the  Floquet-Bloch mode amplitudes. We view this system as consisting of two coupled subsystems; one is for the  quasi-momentum / energy components ``near'' the Dirac point, $(\bK,E_\star)$, and the second governs 
the components which are ``far'' from the Dirac point. We next solve for the far-energy components as a functional of the near-energy components and thereby obtain a reduction to a closed system for the near-energy components. The construction of this map requires that the spectral \nofold condition holds.

In {\bf Section \ref{zz-gap}} we consider the Hamiltonian, introduced in Section \ref{thm-edge-state}, in the weak-potential (low-contrast) regime and prove the existence of topologically protected {\it zigzag} edge states, under the condition $\eps V_{1,1}>0$.

In {\bf  Appendix \ref{V11-section}} we give two families of honeycomb potentials, depending on the lattice scale parameter, $a$, where we can tune between Case (1) $\eps V_{1,1}>0$ and Case (2) $\eps V_{1,1}<0$ by continuously varying the lattice scale parameter.

In a number of places, the proofs of certain assertions are very similar to those of corresponding assertions in \cites{FLW-MAMS:15}.  In such cases, we do not repeat a variation on the proof in \cites{FLW-MAMS:15}, but rather refer to the specific proposition or lemma in \cites{FLW-MAMS:15}.

\subsection{Notation\label{subsec:notation}}

\begin{enumerate}[(1)]
%
\item  $\bv_j,\ j=1,2$ are basis vectors of the triangular lattice in $\R^2$, $\Lambda_h$.
 $\bk_\ell,\ \ell=1,2$ are dual basis vectors of $\Lambda_h^*$, which satisfy $\bk_\ell\cdot\bv_j=2\pi\delta_{\ell j}$. 
 \item For $\bfm=(m_1,m_2)\in\Z^2$, $\bfm\vec\bk=m_1\bk_1+m_2\bk_2$.
 \item $\vtilde_1=a_1\bv_1+a_2\bv_2\in\Lambda_h$,\ $a_1, a_2$ co-prime integers. The $\vtilde_1-$ edge is $\R\vtilde_1$.
 $\vtilde_j,\ j=1,2$, is an alternate basis for $\Lambda_h$ with corresponding dual basis, $\ktilde_\ell, \ell=1,2$, satisfying $\ktilde_\ell\cdot\vtilde_j=2\pi\delta_{\ell j}$. 
 \item  $\ktilde=(\mathfrak{K}^{(1)},\mathfrak{K}^{(2)})$, $\mathfrak{z}\equiv\mathfrak{K}^{(1)} + i \mathfrak{K}^{(2)}$, $|\mathfrak{z}|=|\ktilde|$.
 
%
 \item $\B$ denotes the Brillouin Zone, associated with $\Lambda_h$, shown in the right panel of Figure \ref{fig:lattices}.
 \item $\inner{f,g} = \int\overline{f}g$.
 \item $x\lesssim y$ if and only if there exists $C>0$ such that $x \leq Cy$. $x \approx y$ if and only if $x \lesssim y$ and $y \lesssim x$.
 \item $L^{p,s}(\R)$ is the space of functions $F:\R\rightarrow\R$ such that $(1+\abs{\cdot}^2)^{s/2}F\in L^p(\R)$, endowed with the norm
\[\norm{F}_{L^{p,s}(\R)} \equiv \norm{(1+\abs{\cdot}^2)^{s/2}F}_{L^p(\R)} \approx
  \sum_{j=0}^{s}\norm{\abs{\cdot}^jF}_{L^p(\R)} < \infty,~~~ 1\leq p\leq \infty.\]
 \item For $f,g\in L^2(\R^d)$, the Fourier transform and its inverse are given by 
 {\small
  \begin{equation*}
  \mathcal{F}\{f\}(\xi)\equiv\widehat{f}(\xi)=\frac{1}{(2\pi)^d}\int_{\R^d}e^{-iX\cdot\xi}f(X)dX,~~~
  \mathcal{F}^{-1}\{g\}(X)\equiv\check{g}(X)=\int_{\R^d}e^{ iX\cdot\xi}g(\xi)d\xi.
  \label{FT-def}
  \end{equation*}
  }
The Plancherel relation states:
$\int_{\R^d} f(x)\overline{g(x)} dx = (2\pi)^d\ \int_{\R^d} \widehat{f}(\xi)\overline{\widehat{g}(\xi)} d\xi .$

  
 \item $\sigma_j$, $j=1,2,3$, denote the Pauli matrices, where
  \begin{equation}\label{Pauli-sigma}
  \sigma_1 = \begin{pmatrix}0&1\\1&0\end{pmatrix},~~
  \sigma_2 = \begin{pmatrix}0&-i\\i&0\end{pmatrix},~~\text{and}~~
  \sigma_3 = \begin{pmatrix}1&0\\0&-1\end{pmatrix}.
  \end{equation}
  \end{enumerate}

\subsection{Acknowledgements}

We would like to thank I. Aleiner, A. Millis, J. Liu and  M. Rechtsman for stimulating discussions.

\section{Floquet-Bloch Theory and Honeycomb Lattice Potentials}\label{honeycomb_basics} 

We begin with a review of Floquet-Bloch theory; see, for example, \cites{Eastham:74, RS4, kuchment2012floquet, kuchment2016overview}.

\subsection{Fourier analysis on $L^2(\R/\Lambda)$ and $L^2(\Sigma)$}\label{fourier-analysis}

Let $\{\vtilde_1,\vtilde_2\}$ be a linearly independent set in $\R^2$ and introduce the
\begin{align}
&\text{\bf Lattice: } \Lambda = \Z\vtilde_1\oplus\Z\vtilde_2 = \{m_1\vtilde_1 + m_2\vtilde_2 \ : \ m_1,m_2 \in \Z \} ;   \nn \\
&\text{\bf Fundamental period cell: }  \Omega = \{\theta_1\vtilde_1 + \theta_2\vtilde_2 \ : \ 0\leq\theta_j\leq1,\ j=1,2\} ;   \label{Omega-def} \\
&\text{\bf Dual lattice: }
\Lambda^\ast = \Z\ktilde_1\oplus\Z\ktilde_2 = \{ \bfm\vec\ktilde=m_1\ktilde_1 + m_2\ktilde_2  : m_1,m_2 \in \Z \} , \nn\\ 
&\qquad\qquad\qquad\qquad\qquad\qquad \ktilde_i\cdot\vtilde_j = 2\pi\delta_{ij},\ 1\leq i,j \leq 2;   \nn\\
&\text{\bf Brillouin zone: }  \mathcal{B},\ \textrm{a choice of fundamental dual cell} ; \nn \\
&\text{\bf Cylinder: } \Sigma\equiv  \R^2/\Z\vtilde_1 ;  \nn \\
&\text{\bf Fundamental domain for $\Sigma$: } \Omega_\Sigma\equiv \{\tau_1\vtilde_1 + \tau_2\vtilde_2 : 0\leq\tau_1\leq1, \tau_2\in\R\} .\label{Omega-Sigma-def}
 \end{align}
 
 \nit We denote by $L^2(\Omega)$ and $L^2(\Omega_\Sigma)$  the standard
  $L^2$  spaces on the domains $\Omega$ and $\Omega_\Sigma$, respectively.

 \begin{definition}\label{L2-spaces}[The spaces $L^2(\R^2/\Lambda)$ and $L^2_\bk$]
 \begin{enumerate}
  \item[(a)] $L^2(\R^2/\Lambda)$ denotes the space of $L^2_{loc}$ functions which are  $\Lambda-$ periodic: 
  $ f\in L^2(\R^2/\Lambda)$ if and only if $f(\bx+\vtilde)=f(\bx)$ for all $ \bx\in\R^2,\ \ \vtilde\in\Lambda$ and 
$f\in L^2(\Omega)$.
  \item[(b)] $L^2_\bk$ denotes the space of $L^2_{loc}$ functions which satisfy a pseudo-periodic boundary condition: $ f(\bx+\vtilde)=e^{i\bk\cdot\vtilde}f(\bx)$ for all $\bx\in\R^2,\ \  \vtilde\in\Lambda$ and $e^{-i\bk\cdot\bx}f(\bx)\in L^2(\R^2/\Lambda)$.
  For $f$ and $g$ in $L^2_\bk$, $\overline{f}g$ is in $L^1(\R^2/\Lambda)$ and we define their inner product by
  \begin{equation*}
  \label{L2k_inner_def}
   \inner{f,g}_{L^2_\bk} = \int_\Omega \overline{f(\bx)} g(\bx) d\bx.
  \end{equation*}
 \end{enumerate}
 \end{definition}

\begin{definition}\label{L2Sigma-spaces} [The spaces $L^2(\Sigma)$ and $L^2_{k_\parallel}$]
 \begin{enumerate}
 \item [(a)] $L^2(\Sigma)=L^2(\R^2/\Z\vtilde_1)$ denotes the space of  $L^2_{loc}$ functions, which are periodic in the direction of $\vtilde_1$: $f(\bx+\vtilde_1)=f(\bx), \text{\ \ for\ all\ } \bx\in\R^2$ and such that $f\in L^2(\Omega_\Sigma)$,
  where $\Omega_\Sigma$ is the fundamental domain for $\Sigma$; see \eqref{Omega-Sigma-def}.

  \item [(b)] $L^2_{\kpar}(\Sigma)=L^2_{\kpar}$ denotes the space of $L^2_{loc}$ functions:
  \begin{enumerate}
  \item [(1)] which are $\kpar-$ pseudo-periodic in the direction $\vtilde_1$:
  \begin{equation*}
  f(\bx+\vtilde_1)=e^{i\kpar}f(\bx), \text{\ \ for } \bx\in\R^2, \ \ \text{and}
  \end{equation*}
  \item [(2)] such that $e^{-i (1/2\pi ) k_{\parallel}\ktilde_1\cdot\bx}f(\bx)$, which is defined on $\Sigma$, is in $L^2(\Omega_\Sigma)$.
  \end{enumerate}
  
  For $f$ and $g$ in $L^2_\kpar(\Sigma)$, $\overline{f}g$ is in $L^2(\Sigma)$ and we define their inner product by
   \begin{equation*}
  \label{L2kpar_inner_def}
   \inner{f,g}_{L^2_\kpar} = \int_{\Omega_\Sigma} \overline{f(\bx)} g(\bx) d\bx.
  \end{equation*}
 \end{enumerate}
 %
 %
 The respective Sobolev spaces $H^s(\R^2/\Lambda)$, $H^s_\bk$, $H^s(\Sigma)$ and $H^s_{\kpar}(\Sigma)=H^s_{\kpar}$ are defined in a natural way. 
 
 \nit {\bf Simplified notational convention:} We shall do many calculations requiring us to explicitly write out inner products like   $\inner{f,g}_{L^2(\Sigma)}$ and $\inner{f,g}_{L^2_\kpar(\Sigma)}$. We shall write these as $ \int_\Sigma \overline{f(\bx)}g(\bx)\ d\bx$  rather than as
 $ \int_{\Omega_{_\Sigma}} \overline{f(\bx)}g(\bx)\ d\bx$.
\end{definition}

If $f\in L^2(\R^2/\Lambda)$, then it can be expanded in a Fourier series:
\begin{equation}
 \label{Omega-fourier}
 f(\bx) = \sum_{\bfm\in\Z^2} f_{\bfm} e^{i\bfm\vec\ktilde\cdot\bx}, \quad f_{\bfm} = \frac{1}{|\Omega|}\int_\Omega e^{-i\bfm\vec\ktilde\cdot\by}f(\by)d\by \ , \ \ \bfm\vec\ktilde=m_1\ktilde_1+m_2\ktilde_2 \ ,
\end{equation}
where $|\Omega|$ denotes the area of the fundamental cell, $\Omega$.
In Section \ref{Fourier-edge}, we show that,
if $g\in L^2(\Sigma)$, then it can be expanded in a Fourier series in $\vtilde_1\cdot\bx$ and Fourier transform in $\vtilde_2\cdot\bx$:
\begin{align*}
 g(\bx) &= 2\pi\ \sum_{n\in\Z}\int_\R \widehat{g}_n(2\pi\xi) e^{i\xi\ktilde_2\cdot\bx} d\xi e^{in\ktilde_1\cdot\bx}\ , \\
 2\pi\ \widehat{g}_n(2\pi\xi) &= \frac{1}{\left|\vtilde_1\wedge\vtilde_2\right| } \int_\Sigma e^{-i\xi\ktilde_2\cdot\by} e^{-in\ktilde_1\cdot\by} g(\by) d\by \ .
\end{align*}

\subsection{Floquet-Bloch Theory}\label{flo-bl-theory}

Let $Q(\bx)$ denote a real-valued potential which is periodic with respect to $\Lambda$. We shall assume throughout this paper that $Q\in C^\infty(\R^2/\Lambda)$, although we expect that this condition can be relaxed without much extra work.
Introduce the Schr\"odinger Hamiltonian 
$H \equiv -\Delta + Q(\bx)$. 
For each $\bk\in\R^2$, we study the {\it Floquet-Bloch eigenvalue problem} on $L^2_\bk$:
\begin{align}
 \label{fl-bl-evp}
&H \Phi(\bx;\bk) = E(\bk) \Phi(\bx;\bk), \ \ \bx\in\R^2, \\
&\Phi(\bx+\vtilde)=e^{i\bk\cdot\vtilde}\Phi(\bx;\bk), \ \ \forall \vtilde\in\Lambda \nn .
\end{align}
An $L^2_\bk$ solution of \eqref{fl-bl-evp} is called a {\it Floquet-Bloch} state.

Since the $\bk-$ pseudo-periodic boundary condition in \eqref{fl-bl-evp} is invariant under translations in the dual period lattice, $\Lambda^\ast$, it suffices to restrict our attention to $\bk\in\B$, where  $\B$, the {\it Brillouin Zone},  is a fundamental cell in $\bk-$ space.

An equivalent formulation to \eqref{fl-bl-evp} is obtained by setting $\Phi(\bx;\bk)=e^{i\bk\cdot\bx}p(\bx;\bk)$. Then, 
\begin{equation}
 \label{fl-bl-evp-per}
 H(\bk) p(\bx;\bk) = E(\bk) p(\bx;\bk), \ \bx\in\R^2, \quad p(\bx+\vtilde)=p(\bx;\bk),\ \  \vtilde\in\Lambda,
\end{equation}
where
$ H(\bk) \equiv -(\nabla+i\bk)^2 + Q(\bx)$ 
is a self-adjoint operator on  $L^2(\R^2/\Lambda)$.
 The eigenvalue problem  \eqref{fl-bl-evp-per}, has a discrete set of eigenvalues
$E_1(\bk)\leq E_2(\bk)\leq \cdots \leq E_b(\bk)\leq \cdots$,
with $L^2(\R^2/\Lambda)-$ eigenfunctions $p_b(\bx;\bk),\  b=1,2,3,\ldots$. 
 The maps $\bk\in\mathcal{B}\mapsto E_j(\bk)$ are, in general, Lipschitz continuous functions; see, for example, Appendix A of \cites{FW:14}. For each $\bk\in\B$, the set $\{p_j(\bx;\bk)\}_{j\geq1}$ can be taken to be a  complete orthonormal basis for $L^2(\R^2/\Lambda)$.
%
%
%
%

As $\bk$ varies over $\B$, $E_b(\bk)$ sweeps out a closed real interval. The union over $b\ge1$ of these closed intervals is exactly the $L^2(\R^2)-$ spectrum of $-\Delta+V(\bx)$:
$
 \text{spec} \left(H\right) = \bigcup_{\bk\in\B} \text{spec} \left(H(\bk)\right).
$
Furthermore, the set $\{\Phi_b(\bx;\bk)\}_{b\geq1,\bk\in\B}$ is complete in $L^2(\R^2)$:
\begin{equation*}
 \label{fb-R2-completeness}
 f(\bx) = \sum_{b\geq1} \int_\B \inner{\Phi_b(\cdot;\bk),f(\cdot)}_{L^2(\R^2)} \Phi_b(\bx;\bk) d\bk
 \equiv \sum_{b\geq1} \int_\B \widetilde{f}_b(\bk) \Phi_b(\bx;\bk) d\bk ,
\end{equation*}
where the sum converges in the $L^2$ norm. 
%
%

\subsection{The honeycomb period lattice, $\Lambda_h$, and its dual, $\Lambda_h^*$}\label{sec:honeycomb}

Consider $\Lambda_h=\Z{\bf v}_1 \oplus \Z{\bf v}_2$, the equilateral triangular lattice generated by the basis vectors: ${\bf v}_1=(\frac{\sqrt{3}}{2} , \frac{1}{2})^T$,  ${\bf v}_2=(\frac{\sqrt{3}}{2} , -\frac{1}{2})^T$;  see Figure  \ref{fig:lattices}, left panel.
The dual lattice $\Lambda_h^* =\ \Z {\bf k}_1\oplus \Z{\bf k}_2$ is spanned by the dual basis vectors: $\bk_1=  q( \frac{1}{2} , \frac{\sqrt{3}}{2} )^T$, $\bk_2=  q( \frac{1}{2} , -\frac{\sqrt{3}}{2} )^T$,
 where $ q\equiv \frac{4\pi}{\sqrt{3}}$, 
with the biorthonormality relations ${\bf k}_{i}\cdot {\bf v}_{{j}}=2\pi\delta_{ij}$. Other useful relations are:
 $|\bv_1|=|\bv_2|=1$, $\bv_1\cdot\bv_2=\frac{1}{2}$,  $|\bk_1|=|\bk_2|=q$ and
  $\bk_1\cdot\bk_2=-\frac{1}{2}q^2$.  
The Brillouin zone, $\brill_h$, is a regular hexagon in $\R^2$. Denote by $\bK$ and $\bKp$ its top and bottom vertices (see right panel of Figure \ref{fig:lattices}) given by:\ 
 $\bK\equiv\frac{1}{3}\left(\bk_1-\bk_2\right),\ \ \bKp\equiv-\bK=\frac{1}{3}\left(\bk_2-\bk_1\right)$. 
All six  vertices of $\brill_h$ can be generated by application of the matrix $R$,
 which rotates a vector in $\mathbb{R}^2$ clockwise by $2\pi/3$: 
\begin{equation}
R\ =\ \left(
\begin{array}{cc}
-\frac{1}{2} & \frac{\sqrt{3}}{2}\\
{} & {}\\
-\frac{\sqrt{3}}{2} & -\frac{1}{2}
\end{array}\right)\ .
\label{Rdef}\end{equation}
The vertices of $\brill_h$ fall into two groups, generated by the action of $R$ on $\bK$ and $\bK'$:
$\bK-$ type-points: $ \bK,\ R\bK=\bK+\bk_2,\ R^2\bK=\bK-\bk_1$, and 
$\bKp-$ type-points: $ \bKp,\ R\bKp=\bK'-\bk_2,\ R^2\bKp=\bKp+\bk_1$. 

Functions which are periodic on $\R^2$ with respect to the lattice $\Lambda_h$ may be viewed as functions on the torus, $\R^2/\Lambda_h$.
 As a fundamental period cell, we choose the parallelogram spanned by $\bv_1$ and $\bv_2$, denoted
  $\Omega_h$.
  
\begin{remark}[Symmetry Reduction]\label{symmetry-reduction}
Let $(\Phi(\bx;\bk), E(\bk))$ denote a Floquet-Bloch eigenpair for the eigenvalue problem \eqref{fl-bl-evp} with quasi-momentum $\bk$. Since $V$ is real, 
$(\tilde{\Phi}(\bx;\bk)\equiv\overline{\Phi(\bx;\bk)}, E(\bk))$ is a Floquet-Bloch eigenpair for the eigenvalue problem with quasi-momentum $-\bk$. The above relations among the vertices of $\brill_h$   and the $\Lambda_h^*$- periodicity of: $\bk\mapsto E(\bk)$ and $\bk\mapsto \Phi(\bx;\bk)$ imply that  the local character of the dispersion surfaces in a neighborhood of any vertex of $\brill_h$ is determined by its character about any other vertex of $\brill_h$.
\end{remark}

\subsection{Honeycomb potentials}\label{honeycomb-potentials}

\begin{definition}\label{honeyV}[Honeycomb potentials]
Let $V$ be  real-valued and  $V\in C^\infty(\R^2)$.
$V$ is  a \underline{honeycomb potential} 
 if there exists $\bx_0\in\mathbb{R}^2$ such that $\tilde{V}(\bx)=V(\bx-\bx_0)$ has the following properties:
\begin{enumerate}
\item[(V1)] $\tilde{V}$ is $\Lambda_h-$ periodic, {\it i.e.}  $\tilde{V}(\bx+\bv)=\tilde{V}(\bx)$ for all $\bx\in\mathbb{R}^2$ and $\bv\in\Lambda_h$.  
\item[(V2)] $\tilde{V}$ is even or inversion-symmetric, {\it i.e.} $\tilde{V}(-\bx)=\tilde{V}(\bx)$.
\item[(V3)] $\tilde{V}$  is $\mathcal{R}$- invariant, {\it i.e.}
$ \mathcal{R}[\tilde{V}](\bx)\ \equiv\ \tilde{V}(R^*\bx)\ =\ \tilde{V}(\bx),$
  where, $R^*$ is the counter-clockwise rotation matrix by $2\pi/3$, {\it i.e.} $R^*=R^{-1}$, where $R$ is given by \eqref{Rdef}. 
 \end{enumerate}
N.B. Throughout this paper, we shall omit the tildes on $V$ and choose coordinates with  $\bx_0=0$.
 \end{definition}
 
 
Introduce the mapping $\widetilde{R}:\Z^2\to\Z^2$ which acts on the indices of the Fourier coefficients of $V$:
$ \widetilde{R} (m_1, m_2) = (-m_2, m_1-m_2)$ and therefore 
$ \widetilde{R}^2 (m_1, m_2) = (m_2-m_1,-m_1)$, and  $\widetilde{R}^3(m_1,m_2) = (m_1,m_2)$. 
Any $\bfm\neq0$ lies on an $\widetilde{R}-$ orbit of length exactly three \cites{FW:12}. We say that $\bfm$ and $\bn$ are in the same equivalence class if $\bfm$ and $\bn$ lie on the same $3-$ cycle. Let $\widetilde{S}$ denote a set consisting of exactly one representative from each equivalence class. Honeycomb lattice potentials have the following Fourier series characterization \cites{FW:12}:
 
\begin{proposition}
\label{honey-cosine}
Let $V(\bx)$ denote a honeycomb lattice potential. 
Then,
\begin{align*}
 V(\bx) &= v_{\bf 0} + \sum_{\bfm\in\widetilde{S}} \ v_\bfm \ 
 \left[ \cos(\bfm\vec\bk\cdot\bx) + \cos((\widetilde{R}\bfm)\vec\bk\cdot\bx) + \cos((\widetilde{R}^2\bfm)\vec\bk\cdot\bx) \right] , 
\end{align*}
where $\bfm\vec\bk = m_1\bk_1 + m_2\bk_2$ and the $v_\bfm$ are real.
\end{proposition}

\section{Dirac Points}\label{sec:dirac-pts}

In this section we summarize results of  \cites{FW:12} on {\it Dirac points}. These are conical singularities in the dispersion surfaces of $H_V=-\Delta+V(\bx)$, where $V$ is a honeycomb lattice potential.  

 
 Let $\bK_\star$ denote any vertex of $\brill_h$, and recall that $L^2_{\bK_\star}$ is the space of $\bK_\star-$ pseudo-periodic functions.
 A key property of honeycomb lattice potentials, $V$, is that $H_V$ and $\mathcal{R}$, defined in (V3) of Definition \ref{honeyV}, leave a dense subspace of $L^2_{\bK_\star}$ invariant. Furthermore, restricted to this dense subspace of $L^2_{\bK_\star}$, $H_{V}$ commutes with $\mathcal{R}$: $\left[\mathcal{R},H_{V}\right] = 0$.
  Since $\mathcal{R}$ has eigenvalues $1,\tau$ and $\overline{\tau}$, it is natural to split $L^2_{\bK_\star}$ into the direct sum:
  \begin{equation*}
   L^2_{\bK_\star}\ =\ L^2_{\bK_\star,1}\oplus L^2_{\bK_\star,\tau}\oplus L^2_{\bK_\star,\overline\tau}.\label{L2-directsum}
   \end{equation*}
Here,  $L^2_{\bK_\star,\sigma}$, where $\sigma=1,\tau,\overline{\tau}$ and  $\tau=\exp(2\pi i/3)$, denote the invariant eigenspaces of $\mathcal{R}$:
   \begin{equation*}
   L^2_{\bK_\star,\sigma}\ =\ \Big\{g\in  L^2_{\bK_\star}: \mathcal{R}g=\sigma g\Big\}\ .
\label{L2Ksigma}   \end{equation*}
  
    We next give  a precise definition of a Dirac point. 
\begin{definition}\label{dirac-pt-defn}
Let $V(\bx)$ be a smooth,  real-valued, even (inversion symmetric) and  periodic potential on $\R^2$.
Denote by $\brill_h$, the Brillouin zone. Let $\bK\in\brill_h$.
The energy / quasi-momentum pair $(\bK,E_\star)\in\brill_h\times\R$ is called a {\it Dirac point} if there exists $b_\star\ge1$  such that:
\begin{enumerate}
\item $E_\star$ is a  $L^2_\bK-$ eigenvalue of  $H_V$  of multiplicity two.
\item $\textrm{Nullspace}\Big(H_V-\eig_\star I\Big)\ =\ 
{\rm span}\Big\{ \Phi_1(\bx) , \Phi_2(\bx)\Big\}$, 
where $\Phi_1\in L^2_{\bK,\tau}\ (\mathcal{R}\Phi_1=\tau\Phi_1)$ and 
$\Phi_2(\bx) = \left(\mathcal{C}\circ\mathcal{I}\right)[\Phi_1](\bx) = \overline{\Phi_1(-\bx)}\in L^2_{\bK,\bar\tau}\ (\mathcal{R}\Phi_2=\overline{\tau}\Phi_2)$,
and $\inner{\Phi_a, \Phi_b}_{L^2_\bK(\Omega)} = \delta_{ab}$, $a,b=1,2$.
\item There exist $\lambda_{\sharp}\ne0$, $\zeta_0>0$, and Floquet-Bloch eigenpairs 
\[ \bk\mapsto (\Phi_{b_\star+1}(\bx;\bk),E_{b_\star+1}(\bk))\ \ {\rm and}\ \   \bk\mapsto (\Phi_{b_\star}(\bx;\bk),E_{b_\star}(\bk)),\]
  and Lipschitz functions $e_j(\bk),\ j=b_\star, b_\star+1$, where $e_j(\bK)=0$, defined for $|\bk-\bK|<\zeta_0$ such   that
\begin{align}
E_{b_\star+1}(\bk)-E_\star\ &=\ + |\lambda_\sharp|\ 
\left| \bk-\bK \right|\ 
\left( 1\ +\ e_{b_\star+1}(\bk) \right),\nn\\
E_{b_\star}(\bk)-E_\star\ &=\ - |\lambda_\sharp|\ 
\left| \bk-\bK \right|\ 
\left( 1\ +\ e_{b_\star}(\bk) \right),\label{cones}
\end{align}
where $|e_j(\bk)|\le C|\bk-\bK|,\ j=b_\star, b_\star+1$, for some $C>0$.
\end{enumerate}
\end{definition}


In \cites{FW:12}, the authors prove the following
\begin{proposition}\label{lambda-is=|lambdasharp|}
Suppose conditions $1$ and $2$ of Definition \ref{dirac-pt-defn} hold and let 
$\{c(\bfm)\}_{\bfm\in\mathcal S}$ denote the sequence of $L^2_{\bK,\tau}-$ Fourier-coefficients of $\Phi_1(\bx)$ normalized as in \cites{FW:12}.
Define the sum
\begin{equation}
\lambda_\sharp\ \equiv\   \sum_{\bfm\in\mathcal{S}} c(\bfm)^2\ \left(\begin{array}{c}1\\ i\end{array}\right)\cdot \left(\bK+\bfm\vec\bk\right)\ .
\label{lambda-sharp}
\end{equation}
Here,  $\mathcal{S}\subset\Z^2$ is defined in \cites{FW:12}.
If $\lambda_\sharp\ne0$, then condition 
$3$ of Definition \ref{dirac-pt-defn} holds (see \eqref{cones}).
\end{proposition}

\nit Therefore Dirac points are found by verifying conditions $1$ and $2$ of Definition \ref{dirac-pt-defn} and the additional (non-degeneracy) condition:  $\lambda_\sharp\ne0$.

 Furthermore, Theorem 4.1 of \cites{FW:12} and Theorem 3.2 of \cites{FW:14} imply the following local behavior of Floquet-Bloch modes near the Dirac point:
\footnote{The factor $\frac{\overline{\lambda_\sharp}}{|\lambda_\sharp|}$ in \eqref{Phib_star+1}-\eqref{Phib_star} corrects 
 a typographical error in equation (3.13) of \cites{FW:14}.}

\begin{corollary}\label{Phibstar-bstar+1}

\begin{align}
\Phi_{b_\star+1}(\bx;\bk)\ &=\ \frac{1}{\sqrt2}\ \Big[\ \frac{\overline{\lambda_\sharp}}{|\lambda_\sharp|}\
 \frac{(\bk-\bK)^{(1)}+i(\bk-\bK)^{(2)}}{|\bk-\bK|}\ \Phi_1(\bx)\ +\ \Phi_2(\bx)\ \Big] + \Phi_{b_\star+1}^{(1)}(\bx;\bk), \label{Phib_star+1}\\
\Phi_{b_\star}(\bx;\bk)\ &=\   \frac{1}{\sqrt2}\ \Big[\ \frac{\overline{\lambda_\sharp}}{|\lambda_\sharp|}\
 \frac{(\bk-\bK)^{(1)}+i(\bk-\bK)^{(2)}}{|\bk-\bK|}\ \Phi_1(\bx)\ -\ \Phi_2(\bx)\ \Big] + \Phi_{b_\star}^{(1)}(\bx;\bk) ,
\label{Phib_star}\end{align}
where $ \Phi_{j}^{(1)}(\cdot;\bk) = \mathcal{O}(|\bk-\bK|)$ in $H^2(\Omega_h)$ as $|\bk-\bK|\to0$.
\end{corollary}

In the next section we discuss the result of \cites{FW:12}, that $-\Delta +\eps V$ has Dirac points for generic $\eps$.

\subsection{Dirac points of $-\Delta+\eps V(\bx)$, $\eps$ generic} \label{dpts-generic-eps}

The strategy used in \cites{FW:12} to produce Dirac points is based on a bifurcation theory for the operator $-\Delta + \eps V(\bx)$ acting in $L^2_\bK$, from the $\eps=0$ limit. We describe the setup here, since we shall make detailed use of it. 

Consider $-\Delta$ acting on $L^2_\bK$. We note that $E^0_\star\equiv |\bK|^2$ is an eigenvalue with multiplicity three, since the three vertices of the regular hexagon, $\mathcal{B}_h$: $\bK, R\bK$ and $R^2\bK$ are equidistant from the origin. 
The corresponding three-dimensional eigenspace has an orthonormal basis consisting of the functions: $\Phi_\sigma(\bx) = e^{i\bK\cdot\bx}p_\sigma(\bx) \in L^2_{\bK,\sigma},\ \sigma=1,\tau,\overline{\tau}$, defined by
 \begin{align}
 \Phi_\sigma (\bx) &= e^{i\bK\cdot\bx} p_\sigma(\bx) ,  \qquad ( \sigma=1,\tau,\overline{\tau} )\nn\\
 &= \frac{1}{\sqrt{3|\Omega|}}\ \Big[\ e^{i\bK\cdot\bx} + \overline{\sigma} e^{iR\bK\cdot\bx}+  \sigma e^{iR^2\bK\cdot\bx}\ \Big] \nn \\
 &= \frac{1}{\sqrt{3|\Omega|}}\ e^{i\bK\cdot\bx} \Big[\  1 + \overline{\sigma} e^{i\bk_2\cdot\bx} + \sigma e^{-i\bk_1\cdot\bx}\ \Big]\ .
  \label{p_sigma}
  \end{align}
  We note that 
  \begin{equation}
  \inner{\Phi_\sigma,\Phi_\tsigma}_{L^2_\bK}=
  \inner{p_\sigma, p_\tsigma}_{L^2(\R^2/\Lambda_h)}= \delta_{\sigma,\tsigma} \ .
  \label{orthon}\end{equation}

  In Theorem 5.1 of \cites{FW:12}, the authors proved that for real, small and non-zero $\eps$
  and under the assumption that
 $V$ satisfies the non-degeneracy condition:
\begin{equation}
V_{1,1}\ \equiv\ 
\frac{1}{|\Omega_h|} \int_{\Omega_h} e^{-i(\bk_1+\bk_2)\cdot\by}\ V(\by)\ d\by\ne0,
\label{V11eq0}
\end{equation}
that  the multiplicity three eigenvalue, $E^0_\star=|\bK|^2$, splits into\\
  (A) a multiplicity two eigenvalue, $E^\eps_\star$, with two-dimensional $L^2_{\bK,\tau}\oplus L^2_{\bK,\overline{\tau}}-$ eigenspace structure,  and\\
  (B) a simple eigenvalue, $\widetilde{E}^\eps$, with one-dimensional eigenspace, a subspace of $L^2_{\bK,1}$. 
  
 For all  $\eps$ sufficiently  small, the quasi-momentum pairs $(\bK,E^\eps_\star)$ are Dirac points in the sense of Definition \ref{dirac-pt-defn}. 
 Furthermore,
  a continuation argument is then used to extend this result from the regime of sufficiently small $\eps$ to the regime of arbitrary $\eps$ outside of a possible discrete set; see \cites{FW:12}
   and the refinement concerning the possible exceptional set of $\eps$ values in Appendix D of \cites{FLW-MAMS:15}. 
   We first state the result for arbitrarily large and generic $\eps$, and then the more refined picture for $|\eps|>0$ and sufficiently small.

\begin{theorem}\label{diracpt-thm}[Generic $\eps$]
Let $V(\bx)$ be a honeycomb lattice potential and consider the 
parameter family of Schr\"odinger operators:
\begin{equation*}
H^{(\eps)}\ \equiv\ -\Delta + \eps\ V(\bx),
\label{Heps-def}
\end{equation*}
where $V$ satisfies the non-degeneracy condition \eqref{V11eq0}. 
Then, there exists $\eps_0>0$, such that for all real and nonzero $\eps$, outside of a possible discrete subset of $\R\setminus(-\eps_0,\eps_0)$,  $H^{(\eps)}$ has Dirac points $(\bK,E^\eps_\star)$ in the sense of Definition \ref{dirac-pt-defn}.

Specifically, for all such $\eps$, 
 there exists $b_\star\ge1$ such that $\eig_\star\equiv\eig^\eps_{b_\star}(\bK)=\eig^\eps_{b_\star+1}(\bK)$ is a $\bK-$ pseudo-periodic eigenvalue of multiplicity two
  where
\begin{enumerate}
\item

(a) $\eig^\eps_\star$  is an $L^2_{\bK,\tau}-$ eigenvalue of $H^{(\eps)}$ of multiplicity one, with corresponding eigenfunction, 
$\Phi^\eps_1(\bx)$.\\
(b) $\eig^\eps_\star$ is an $L^2_{\bK,\bar\tau}-$ eigenvalue of $H^{(\eps)}$ of multiplicity one, with corresponding eigenfunction, $\Phi^\eps_2(\bx)=\overline{\Phi^\eps_1(-\bx)}$.\\
(c) $\eig^\eps_\star$ is \underline{not} an $L^2_{\bK,1}-$ eigenvalue of $H^{(\eps)}$.
\item  There exist $\delta_\eps>0,\ C_\eps>0$
  and Floquet-Bloch eigenpairs: $(\Phi_j^\eps(\bx;\bk), \eig_j^\eps(\bk))$  
  and  Lipschitz continuous functions, $e_j(\bk)$,  $j=b_\star, b_\star+1$,
    defined for  $|\bk-\bK|<\delta_\eps$,  such that 
\begin{align}
\eig^\eps_{b_\star+1}(\bk)-\eig^\eps(\bK)\ &=\ +\ |\lambda^\eps_\sharp|\ 
\left| \bk-\bK \right|\ 
\left( 1\ +\ e^\eps_{b_\star+1}(\bk) \right)\ \ {\rm and}\nn\\
\eig^\eps_{b_\star}\bk)-\eig^\eps(\bK)\ &=\ -\ |\lambda^\eps_\sharp|\ 
\left| \bk-\bK \right|\ 
\left( 1\ +\ e^\eps_{b_\star}\bk) \right),\label{conical}
\end{align}
 and where
\begin{equation}
\lambda_\sharp^\eps\ \equiv\   \sum_{\bfm\in\mathcal{S}} c(\bfm,\eig_\star^\eps,\eps)^2\ \left(\begin{array}{c}1\\ i\end{array}\right)\cdot \left(\bK+\vec\bfm\vec\bk\right) \ \ne\ 0
\label{lambda-sharp2}
\end{equation}
 is given in terms of $\{c(\bfm,\eig_\star^\eps,\eps)\}_{\bfm\in\mathcal{S}}$, the $L^2_{\bK,\tau}-$ Fourier coefficients of $\Phi_1^\eps(\bx;\bK)$.
 Furthermore, $|e_j^\eps(\bk)| \le C_\eps |\bk-\bK|$, $j=b_\star, b_\star+1$.  
Thus, in a neighborhood of the point $(\bk,\eig)=(\bK,\eig_\star^\eps)\in \R^3 $, the dispersion surface is closely approximated by a  circular {\it cone}.
\end{enumerate}
\end{theorem}

\subsection{Dirac points of $-\Delta+\eps V(\bx)$, $\eps$ small} \label{dpts-small-eps}

In this section we collect explicit information on Dirac points for the weak potential regime.

\begin{theorem}\label{diracpt-small-thm}[Small $\eps$]
There exists $\eps_0>0$, such that for all $\eps\in I_{\eps_0}\equiv(-\eps_0,\eps_0)\setminus\{0\}$ the following holds:
\begin{enumerate}
\item For  $\eps\in I_{\eps_0}$,  $-\Delta +\eps V(\bx)$ has
\begin{enumerate}[(a)]
\item a multiplicity two $L^2_{\bK}$- eigenvalue $E^\eps_\star$, where $\textrm{ker}(-\Delta + \eps V)\subset L^2_{\bK,\tau}\oplus
L^2_{\bK,\overline{\tau}}$, and
\item a multiplicity one $L^2_{\bK}$- eigenvalue $\widetilde{E}_\star^\eps$, where $\textrm{ker}(-\Delta + \eps V)\subset L^2_{\bK,1}$.
\end{enumerate}
\item The maps $\eps\mapsto E^\eps_\star$ and $\eps\mapsto \widetilde{E}^\eps_\star$ are well defined for all $\eps$ in the deleted neighborhood of zero, $I_{\eps_0}$. They are constructed via perturbation theory of a simple eigenvalue in 
$L^2_{\bK,\tau}$ and in $L^2_{\bK,1}$, respectively. Therefore,  $E^\eps_\star$ and $\widetilde{E}^\eps_\star$ are real-analytic functions of $\eps\in I_{\eps_0}$. Moreover, they have the expansions:
\begin{align}
E^\eps_\star\ &= |\bK|^2 + \eps(V_{0,0}-V_{1,1})+\mathcal{O}(\eps^2), \label{E*expand}\\
\widetilde{E}^\eps_\star &=|\bK|^2 + \eps(V_{0,0}+2V_{1,1})+\mathcal{O}(\eps^2). \label{tildeE*expand}
\end{align}
\item If $\eps V_{1,1}>0$, then conical intersections occur between the $1^{st}$ and $2^{nd}$ dispersion surfaces at the vertices of $\mathcal{B}_h$. Specifically, \eqref{conical} holds with $b_\star=1$.
\item If  $\eps V_{1,1}<0$, 
then conical intersections occur between the $2^{nd}$ and $3^{rd}$ dispersion surfaces at the vertices of $\mathcal{B}_h$. Specifically, \eqref{conical} holds with $b_\star=2$.

For $\eps\in I_{\eps_0}$, 
\begin{equation}
| \lambda_\sharp^\eps| = 4\pi |\Omega_h| + \mathcal{O}(\eps) = 4\pi|\bv_1\wedge\bv_2|+ \mathcal{O}(\eps).
\label{lambda-sharp-expand}
\end{equation}
\end{enumerate}
\end{theorem}

\nit The expansions \eqref{E*expand}, \eqref{tildeE*expand} and \eqref{lambda-sharp-expand} are displayed in equations (6.22), (6.25) and (6.30) of \cites{FW:12}.

\nit The intersections of the first two dispersion surfaces for  $\eps V_{1,1}>0$,  and of the second and third dispersion surfaces for $\eps V_{1,1}<0$, are illustrated in the first two panels of Figure \ref{fig:eps_V11_neg} along a dispersion slice corresponding to the zigzag edge.

\section{Edges and  dual slices}\label{ds-slices}

Edge states are solutions of an eigenvalue equation on $\R^2$, which are spatially localized transverse to a line-defect or ``edge'' and propagating (plane-wave like or pseudo-periodic) parallel to the edge. Recall that $\Lambda_h=\Z\bv_1\oplus\Z\bv_2$ and $\Lambda_h^*=\Z\bk_1\oplus\Z\bk_2$. We consider edges which are lines
 of the form $\R (a_1\bv_1+a_2\bv_2)$, where $(a_1,b_1)=1$, {\it i.e.} $a_1$ and $b_1$ are relatively prime. 
 
 We fix an edge by choosing $\vtilde_1=a_1\bv_1+b_1\bv_2$, where  $(a_1, b_1)=1$. Since 
   $a_1, b_1$ are relatively prime,  there exists a relatively prime pair of integers: $a_2,b_2$  such that $a_1b_2-a_2b_1=1$. 
     Set $\vtilde_2 = a_2 \bv_1 + b_2 \bv_2$. 
     %
%
%
It follows that $\Z\vtilde_1\oplus\Z\vtilde_2=\Z\bv_1\oplus\Z\bv_2=\Lambda_h$.
Since $a_1b_2-a_2b_1=1$, we have dual lattice vectors $\ktilde_1, \ktilde_2\in\Lambda_h^*$, given by
\[ \ktilde_1=b_2\bk_1-a_2\bk_2,\ \  \ktilde_2=-b_1\bk_1+a_1\bk_2, \]
  which satisfy
\begin{equation*} \ktilde_\ell \cdot \vtilde_{\ell'} = 2\pi \delta_{\ell, \ell'},\ \ 1\leq \ell, \ell' \leq 2.\label{ktilde-orthog}
\end{equation*}
Note that 
 $\Z\ktilde_1\oplus\Z\ktilde_2=\Z\bk_1\oplus\Z\bk_2=\Lambda^*_h$.

Fix an edge, $\R\vtilde_1$. In our construction of edge states, an  important role is played by the ``quasi-momentum slice'' of the band structure through the Dirac point and ``dual'' to the given edge.


\begin{definition}\label{dual-slice}
 For the edge $\R\vtilde_1$,  {\it the band structure slice at quasi-momentum $\bK$, dual to the edge $\R\vtilde_1$}, 
   is defined to be the locus given by the union of curves:
  \begin{equation*}
\lambda\mapsto E_b(\bK+\lambda\ktilde_2),\ \ |\lambda|\le 1/2,\ \ b\ge1.
\label{tv-slice}
\end{equation*}
\end{definition}

We give two examples:
\begin{enumerate}
\item Zigzag: $\vtilde_1=\bv_1$, $\vtilde_2=\bv_2$ and $\ktilde_1=\bk_1$ and $\ktilde_2=\bk_2$.
\\  In this case, we shall refer to the {\sl zigzag slice}.
\item Armchair:  $\vtilde_1=\bv_1+\bv_2$, $\vtilde_2=\bv_2$ and $\ktilde_1=\bk_1$ and $\ktilde_2=\bk_2-\bk_1$. \\ In this case, we shall refer to the {\sl armchair slice}.
\end{enumerate}

Figure \ref{fig:eps_V11_neg} (top row) displays three cases, for $-\Delta+\eps V$, where $V$ is a honeycomb lattice potential. Shown are the curves
 $\lambda\mapsto E_b(\bK+\lambda\ktilde_2)$, $b=1,2,3$ for (i) $\eps V_{1,1}>0$ and the zigzag slice (left panel),  (ii) $\eps V_{1,1}<0$ and the zigzag slice (middle panel), and (iii) the armchair slice (right panel). As discussed in the introduction, of these three examples, case (i) is the one for which the spectral \nofold condition of  Definition \ref{SGC} holds.


\subsection{Completeness of Floquet-Bloch modes on $L^2(\Sigma)$}\label{Fourier-edge}

For $\vtilde_1\in\Lambda_h$, introduce the cylinder $\Sigma= \R^2/\Z\vtilde_1$. Consider the family of  states $\Phi_b(\bx;\bK+\lambda\ktilde_2),\ b\ge1$ for $\lambda\in[0,1]$ (or equivalently $|\lambda|\le1/2$) corresponding to quasi-momenta along a line segment within $\B_h$ connecting $\bK$ to $\bK+\ktilde_2$.
Since $\ktilde_2\cdot\vtilde_1=0$, all along this segment we have $\bK\cdot\vtilde_1-$ pseudo-periodicity:
\begin{equation}
\Phi_b(\bx+\vtilde_1;\bK+\lambda\ktilde_2)= e^{i(\bK+\lambda\ktilde_2)\cdot\vtilde_1} \Phi_b(\bx;\bK+\lambda\ktilde_2)= e^{i\bK\cdot\vtilde_1} \Phi_b(\bx;\bK+\lambda\ktilde_2) \ .
\nn\end{equation}
The main result of this subsection is that any $f\in L_\kparv^2(\Sigma)$ is a superposition of these modes.

\begin{theorem}\label{fourier-edge}
Let $f\in L_\kparv^2(\Sigma)= L_\kparv^2(\R^2/\Z\vtilde_1)$. Then,
\begin{enumerate}
\item  $f$ can be represented as a superposition of Floquet-Bloch modes of $-\Delta+V$ with quasimomenta in $\mathcal{B}$ located on the segment
$\Big\{\bk=\bK+\lambda\ktilde_2: |\lambda|\le\frac{1}{2}\Big\}:$
\begin{align}
f(\bx)
&=\sum_{b\ge1}\int_{-\frac12}^{\frac12} \widetilde{f}_b(\lambda) \Phi_b(\bx;\bK+\lambda \ktilde_2) d\lambda  \nn \\ 
&= e^{i\bK\cdot\bx} \sum_{b\ge1}\int_{-\frac12}^{\frac12} e^{i\lambda\ktilde_2\cdot\bx}\widetilde{f}_b(\lambda) p_b(\bx;\bK+\lambda \ktilde_2) d\lambda,\qquad {\rm where} \label{fb-Sigma} \\ 
\widetilde{f}_b(\lambda) \ &=\ \left\langle \Phi_b(\cdot,\bK+\lambda\ktilde_2),f(\cdot)\right\rangle_{L_\kparv^2(\Sigma)} . \nn
\end{align}
Here, the sum representing $e^{-i\bK\cdot\bx}f(\bx)$, in \eqref{fb-Sigma} converges in the $L^2(\Sigma)$ norm. 
\item In the special case where $V\equiv0$:
 \begin{align*}
 f(\bx) = \sum_{{\bfm}\in\Z^2}e^{ i(\bK+\bfm\vec\ktilde)\cdot\bx } \int_{-\frac12}^{\frac12}   
 \widehat{f}_\bfm(\lambda)e^{i\lambda\ktilde_2\cdot\bx} d\lambda\ .
 \end{align*}
 \end{enumerate}
\end{theorem}

\begin{proof}[Proof of Theorem \ref{fourier-edge}]
We introduce the parameterizations of the fundamental period cell $\Omega$ of $V(\bx)$:
\begin{align}
\label{omega-parameterization}
\bx\in\Omega:\quad &\bx = \tau_1\vtilde_1+\tau_2\vtilde_2,\quad 0\le\tau_1, \tau_2\le1 , \quad
\bk_i\cdot\bx = 2\pi\tau_i ,  \nn \\ 
& dx_1\ dx_2\ = \left|\vtilde_1\wedge\vtilde_2\right|\ d\tau_1\ d\tau_2\ =\ |\Omega|\  d\tau_1\ d\tau_2;
\end{align}
and of the cylinder $\Sigma=\R^2/\Z\vtilde_1$:
\begin{align}
\label{sigma-parameterization}
\bx\in\Sigma:\quad &\bx = \tau_1\vtilde_1+\tau_2\vtilde_2,\ \ 0\le\tau_1\le1,\ s\in\R , \quad
\ktilde_1\cdot\bx = 2\pi\tau_1 , \ \ktilde_2\cdot\bx = 2\pi\tau_2 , \nn \\
&dx_1\ dx_2\ = \left|\vtilde_1\wedge\vtilde_2\right|\ d\tau_1\ d\tau_2\ =\ |\Omega|\  d\tau_1\ d\tau_2.
\end{align}

Let $f\in L^2_\kparv(\Sigma)$ be such that $g(\bx)= e^{-i\bK\cdot\bx}f(\bx)$  is defined and smooth on $\Sigma$,  and rapidly decreasing.  It suffices to prove the result for such $f$, and then  pass to all  $L^2_\kparv(\Sigma)$ by standard arguments.
The function $g(\bx)$ has the Fourier representation
\begin{align}
 \label{g-fourier}
 g(\bx) &= 2\pi\ \sum_{n\in\Z}\int_\R \widehat{g}_n(2\pi\xi) e^{i\xi\ktilde_2\cdot\bx} d\xi e^{in\ktilde_1\cdot\bx}, \\
 2\pi\ \widehat{g}_n(2\pi\xi) &= \frac{1}{\left|\vtilde_1\wedge\vtilde_2\right| } \int_\Sigma e^{-i\xi\ktilde_2\cdot\by} e^{-in\ktilde_1\cdot\by} g(\by) d\by \ . \nn
\end{align}
The relation \eqref{g-fourier} is obtained by noting that $G(\tau_1,\tau_2)=g(\tau_1\vtilde_1+\tau_2\vtilde_2)$ is $1-$ periodic in $\tau_1$ and in $L^2(\R; d\tau_2)$, and applying the standard Fourier representations.

Introduce the Gelfand-Bloch transform
\begin{equation}
 \label{g-FB-tilde}
 \widetilde{g}(\bx;\lambda) = 2\pi \sum_{(m_1,m_2)\in\Z^2}\widehat{g}_{m_1}\left(2\pi(m_2+\lambda)\right)e^{i(m_1\ktilde_1+m_2\ktilde_2)\cdot\bx},\quad |\lambda|\leq 1/2 \ .
\end{equation}
Note that $\bx\mapsto \widetilde{g}(\bx;\lambda)$ is $\Lambda_h-$ periodic and $\lambda\mapsto \widetilde{g}(\bx;\lambda)$ is $1-$ periodic.
Using \eqref{g-FB-tilde} and \eqref{g-fourier}, it is straightforward to check that
\begin{equation}
 \label{g-FB}
 g(\bx) = \int_{-\frac12}^{\frac12} e^{i\lambda\ktilde_2\cdot\bx}\ \widetilde{g}(\bx;\lambda)\ d\lambda \ .
\end{equation}

\begin{remark}\label{pb-Omega}
For any fixed  $|\lambda|\le1/2$, the mapping $\bx\mapsto \widetilde{g}(\bx;\lambda)$ is $\Lambda_h-$ periodic. 
We wish to expand $\bx\mapsto \widetilde{g}(\bx;\lambda)$ in terms of a basis for $L^2(\Omega)$,
 where $\Omega$ denotes our choice of period cell (parallelogram) for $\R^2/\Lambda$ with $\Lambda=\Z\vtilde_1\oplus\Z\vtilde_2$;
  see \eqref{Omega-def}. Now the eigenvalue problem $H(\bk)p^\Omega=E^\Omega p^\Omega$  
on $\Omega$  with periodic boundary conditions has a discrete sequence of eigenvalues, $E_j^\Omega(\bk),\ j\ge1$ with corresponding eigenfunctions $p^\Omega_j(\bx;\bk),\ j\ge 1$, which can be taken to be a complete orthonormal sequence. Recall $p^{\Omega_h}_b(\bx;\bk),\ b\ge1$, with corresponding eigenvalues, $E_b(\bk)$,  the complete set of eigenfunctions of $H(\bk)$  with periodic boundary conditions on  $\Omega_h$, the elementary period parallelogram spanned by $\{\bv_1,\bv_2\}$; see Section \ref{flo-bl-theory}. By periodicity, $p^{\Omega_h}_b(\bx;\bk),\ b\ge1,$ (initially defined on $\Omega_h$) and $p_j^\Omega(\bx;\bk),\ j\ge1$, (initially defined on $\Omega$)
can be extended to all $\R^2$ as periodic functions. We continue to denote these extensions  by:   $p_j^\Omega(\bx;\bk)$ and $p^{\Omega_h}_b(\bx;\bk)$, respectively. Since $\Lambda=\Z\vtilde_1\oplus\Z\vtilde_2=\Z\bv_1\oplus\Z\bv_2=\Lambda_h$,  both sequences of eigenfunctions are $\Lambda_h-$ periodic. Thus, in a natural way, we can take $p^\Omega_b(\bx;\bk)=A_b\ p^{\Omega_h}_b(\bx;\bk),\ b\ge1$, where $A_b$ is a normalization constant. Abusing notation, we henceforth drop the explicit dependence on $\Omega$, and simply write $p_b(\bx;\bk)$ for $p_b^{\Omega}(\bx;\bk)$. 
\end{remark}
 
In view of Remark \ref{pb-Omega} we expand $\widetilde{g}(\bx;\lambda)$ in terms of the states $\{p_b(\cdot;\bK+\lambda\ktilde_2)\},\ b\ge1$:
\begin{equation}
 \label{g-FB-p-expansion}
 \widetilde{g}(\bx;\lambda) = \sum_{b\geq1}\inner{p_b(\cdot;\bK+\lambda\ktilde_2),\widetilde{g}(\cdot,\lambda)}_{L^2(\Omega)}\ p_b(\bx;\bK+\lambda\ktilde_2) .
\end{equation}
Recall that $f(\bx)= e^{i\bK\cdot\bx}g(\bx)$. We claim (and prove below) that 
\begin{equation}
 \label{g-FB-claim}
 \inner{p_b(\cdot;\bK+\lambda\ktilde_2),\widetilde{g}(\cdot,\lambda)}_{L^2(\Omega)}
 = \inner{\Phi_b(\cdot;\bK+\lambda\ktilde_2),f(\cdot)}_{L_\kparv^2(\Sigma)} \ .
\end{equation}
The assertions of Theorem \ref{fourier-edge} then follow from  \eqref{g-FB}, 
 \eqref{g-FB-p-expansion} and the claim \eqref{g-FB-claim}:
\begin{align*}
 f(\bx) &= e^{i\bK\cdot\bx} g(\bx)
 = e^{i\bK\cdot\bx} \int_{-\frac12}^{\frac12} e^{i\lambda\ktilde_2\cdot\bx} \widetilde{g}(\bx;\lambda) d\lambda \\
  & = e^{i\bK\cdot\bx} \sum_{b\geq1} \int_{-\frac12}^{\frac12} e^{i\lambda\ktilde_2\cdot\bx} \inner{p_b(\cdot;\bK+\lambda\ktilde_2),\widetilde{g}(\cdot,\lambda)}_{L^2(\Omega)} p_b(\bx;\bK+\lambda\ktilde_2) d\lambda \\
  &= \sum_{b\geq1} \int_{-\frac12}^{\frac12} \inner{\Phi_b(\cdot;\bK+\lambda\ktilde_2),f(\cdot)}_{L_\kparv^2(\Sigma)} \Phi_b(\bx;\bK+\lambda\ktilde_2) d\lambda \ ,
\end{align*}
where, in the final line we have used that $\Phi_b(\by;\lambda)=p_b(\by;\lambda)e^{i(\bK+\lambda\ktilde_2)\cdot\bx}$. Therefore, it remains to prove claim \eqref{g-FB-claim}.  We shall employ the one-dimensional Poisson summation formula:
\begin{equation}
\label{poisson-sum-formula}
 2\pi\sum_{n\in\Z} \widehat{f}\left(2\pi(n+\lambda)\right) e^{2\pi iny} = \sum_{n\in\Z}f(y+n)e^{-2\pi i \lambda(y+n)} \ .
\end{equation}
In the following calculation we use the abbreviated notation:
\[p_b(\by;\lambda)\equiv p_b(\by;\bK+\lambda\ktilde_2)\quad \text{and} \quad \ \Phi_b(\by;\lambda)\equiv \Phi_b(\by;\bK+\lambda\ktilde_2).\]
Substituting \eqref{g-fourier}-\eqref{g-FB-tilde} into the left hand side of \eqref{g-FB-claim} and applying  \eqref{poisson-sum-formula} gives
{\footnotesize
\begin{align*}
 &\inner{p_b(\cdot;\lambda),\widetilde{g}(\cdot,\lambda)}_{L^2(\Omega)}
 = \int_{\Omega} \overline{p_b(\by;\lambda)}  \widetilde{g}(\by,\lambda)  d\by \nn \\
 &\ = |\Omega| \int_0^1 \int_0^1 \overline{p_b(\tau_1\vtilde_1+\tau_2\vtilde_2;\lambda)}  2\pi \sum_{\bfm\in\Z^2} \widehat{g}_{m_1} \left(2\pi(m_2+\lambda)\right) e^{2\pi i(m_1\tau_1+m_2\tau_2)}  d\tau_1 d\tau_2 \ \ ( \eqref{omega-parameterization}) \nn \\
 &\ = |\Omega| \int_0^1 \int_0^1 \overline{p_b(\tau_1\vtilde_1+\tau_2\vtilde_2;\lambda)}  \sum_{m_1\in\Z} \left [2\pi \sum_{m_2\in\Z} \widehat{g}_{m_1}\left(2\pi(m_2+\lambda)\right) e^{2\pi im_2\tau_2} \right] e^{2\pi im_1\tau_1}  d\tau_1 d\tau_2 \nn \\
 &\ = |\Omega| \int_0^1 \int_0^1 \overline{p_b(\tau_1\vtilde_1+\tau_2\vtilde_2;\lambda)}  \sum_{m_1\in\Z} \left [\sum_{m_2\in\Z} g_{m_1}(\tau_2+m_2) e^{-2\pi i \lambda(\tau_2+m_2)} \right] e^{2\pi im_1\tau_1}  d\tau_1 d\tau_2 \ \  (\eqref{poisson-sum-formula}) \nn \\
 &\ = |\Omega| \int_0^1 \int_\R \overline{p_b(\tau_1\vtilde_1+s\vtilde_2;\lambda)} e^{-2\pi i \lambda s}  \sum_{m_1\in\Z} g_{m_1}(s) e^{2\pi im_1\tau_1}  d\tau_1 ds \ (\tau_2+m_2=s) \nn \\
 &\ = |\Omega| \int_0^1 \int_\R \overline{p_b(\tau_1\vtilde_1+s\vtilde_2;\lambda)} e^{-2\pi i \lambda s}  G(\tau_1,s)  d\tau_1 ds \nn 
 \ \ ( G(\tau_1,\tau_2)=g(\tau_1\vtilde_1+\tau_2\vtilde_2) \text{ and } \eqref{g-fourier}) \\
 &\ = \int_\Sigma \overline{p_b(\by;\lambda)}  e^{-i \lambda \ktilde_2\cdot\bx}  g(\by)   d\by \ \ (\text{by } \eqref{sigma-parameterization}) \nn  .
\end{align*}
}
Finally, recalling that $g= e^{-i\bK\cdot\bx}f$ and $\overline{p_b(\by;\lambda)}e^{-i(\bK+\lambda\ktilde_2)\cdot\bx}=\overline{\Phi_b(\by;\lambda)}$, we obtain that
\begin{align*}
 \inner{p_b(\cdot;\lambda),\widetilde{g}(\cdot,\lambda)}_{L^2(\Omega)}
 &= \int_\Sigma \overline{p_b(\by;\lambda)} \ e^{-i \lambda \ktilde_2\cdot\bx} e^{-i\bK\cdot\bx} \ f(\by)  \ d\by \nn \\
 &= \int_{\Sigma} \overline{\Phi_b(\by;\lambda)} f(\by) \ d\by \nn
 =\inner{\Phi_b(\cdot;\lambda),f(\cdot)}_{L_\kparv^2(\Sigma)}. \nn
\end{align*}

This completes the proof of claim \eqref{g-FB-claim} and part 1 for the case where $f$ is smooth and rapidly decreasing. 
 Passing to arbitrary $f\in L^2_\kpar$ is standard. 
In the case where $V\equiv0$, the Schr\"odinger operator reduces to the Laplacian $-\Delta$. In this case the Floquet-Bloch coefficients of $f\in L^2_\kparv$ are simply its Fourier coefficients: $\widetilde{f}(\lambda)=\widehat{f}(\lambda)$. Part 2 therefore follows from part 1, completing the proof of Theorem \ref{fourier-edge}.
\end{proof}


Sobolev regularity can be measured in terms of the Floquet-Bloch coefficients. Indeed, as in
 Lemma 2.1 in \cites{FLW-MAMS:15}, by  the $2D-$ Weyl law  $E_b(\bK+\lambda\ktilde_2)\approx b$ for all $\lambda\in[-1/2,1/2],\ b\gg1$, we have

\begin{corollary}\label{fourier-edge-norms}
$L_\kparv^2(\Sigma)$ and $H_\kparv^s(\Sigma),\ s\in\N$, norms can be expressed in terms of the Floquet-Bloch coefficients $\widetilde{f}_b(\lambda)$, $b\ge1$. For $f\in L_\kparv^2(\Sigma)= L_\kparv^2(\R^2/\Z\vtilde_1)$:
\begin{align*}
\|f\|_{L_\kparv^2(\Sigma)}^2\ &\sim\ \sum_{b\ge1}\int_{-\frac12}^{\frac12}\ |\widetilde{f}_b(\lambda) |^2 d\lambda\ ,\\ 
\|f\|_{H_\kparv^{^s}(\Sigma)}^2\ &\sim\ \sum_{b\ge1} (1+b)^s\ \int_{-\frac12}^{\frac12} |\widetilde{f}_b(\lambda) |^2 d\lambda\ . 
\end{align*}
\end{corollary}

\subsection{Expansion of  $\bk\mapsto E_b(\bk)$ along a quasi-momentum slice}\label{bloch-near-dirac}

Let $(\bK,E_\star)$ denote a Dirac point as in Definition  \ref{dirac-pt-defn}. In a neighborhood of 
the Dirac point, the eigenvalues $E_{b_\star}(\bk)$ and $E_{b_\star+1}(\bk)$ are Lipschitz continuous functions and the corresponding normalized eigenmodes, $\Phi_{b_\star}(\bx;\bk)$
 and $\Phi_{b_\star+1}(\bx;\bk)$ are discontinuous functions of $\bk$; see \cites{FW:14}. Note however that what is relevant to our construction of $\vtilde_1-$ edge states are Floquet-Bloch modes along the quasi-momentum line $\bK+\lambda\ktilde_2$, $|\lambda|\leq1/2$. The following proposition gives a smooth parametrization  of these modes along this quasi-momentum line.

\begin{proposition}\label{directional-bloch}
Let $(\bK,E_\star)$ denote a Dirac point in the sense of Definition \ref{dirac-pt-defn}. Let $\{ \Phi_1(\bx), \Phi_2(\bx) \}$ denote the basis of the 
$L^2_{\bK}=L^2_{\bK,\tau}\oplus L^2_{\bK,\overline\tau}-$ nullspace of $H_V - E_\star I$ in Definition \ref{dirac-pt-defn}. Introduce the $\Lambda_h-$ periodic functions
\begin{equation}
 P_1(\bx)=e^{-i\bK\cdot \bx}\Phi_1(\bx),\ \ P_2(\bx)=e^{-i\bK\cdot \bx}\Phi_2(\bx).
 \label{PhipmKlam}
\end{equation}
For each $|\lambda|\le1/2$, there exist $L^2_{\bK+\lambda\ktilde_2}-$ eigenpairs $(\Phi_\pm(\bx;\lambda),E_\pm(\lambda))$,  real analytic in $\lambda$,  such that 
$\left\langle \Phi_a(\cdot;\lambda),\Phi_b(\cdot;\lambda)\right\rangle=\delta_{ab}$ and 
\[ \textrm{span}\ \{\Phi_{-}(\bx;\lambda), \Phi_{+}(\bx;\lambda)\}
 =\ \textrm{span}\ \{\Phi_{b_\star}(\bx;\bK+\lambda\ktilde_2), \Phi_{b_\star+1}(\bx;\bK+\lambda\ktilde_2)\}\ .
 \]
  
Introduce $\Lambda_h-$ periodic functions $p_\pm(\bx;\lambda)$ by 
 \begin{equation}
 \Phi_\pm(\bx;\lambda)\ =\ e^{i(\bK+\lambda\ktilde_2)\cdot\bx}\ p_\pm(\bx;\lambda),\ \ \left\langle p_a(\cdot;\lambda),p_b(\cdot;\lambda)\right\rangle=\delta_{ab} ,\ \ 
 a,b\in\{+,-\}.
 \label{p_pm-def}
 \end{equation}
 There is a constant $\zeta_0>0$ such that for $|\lambda|<\zeta_0$ the following holds:
 \begin{enumerate}
 \item The mapping $\lambda\mapsto E_\pm(\lambda)$ is real analytic in $\lambda$ with expansion 
 \begin{equation}
 E_\pm(\lambda) = E_\star \pm |\lambda_\sharp|\ |\ktilde_2|\ \lambda + E_{2,\pm}(\lambda)\lambda^2 \ ,
 \label{EpmKlam}
 \end{equation}
 where $\lambda_\sharp\in\C$ is given by \eqref{lambda-sharp}, $|E_{2,\pm}(\lambda)|\leq C$ with $C$ a positive constant independent of $\lambda$. 
  \item Let
 $ \mathfrak{z}_2=\smallktilde_2^{(1)}+i\smallktilde_2^{(2)},\ |\mathfrak{z}_2|=|\ktilde_2|$. The $\Lambda_h-$ periodic functions, $p_\pm(\bx;\lambda)$, can be chosen to depend real analytically on $\lambda$ and so that
  \footnote{The factor of $\frac{\overline{\lambda_\sharp}}{|\lambda_\sharp|}$ in \eqref{ppmKlam} corrects 
 a typographical error in equation (3.13) of \cites{FW:14}.}
 \begin{align}
p_\pm(\bx;\lambda) &= 
P_\pm(\bx)\ +\ \varphi_{\pm}(\bx,\lambda) \ \in L^2_\bK(\R^2/\Lambda_h),
 \label{ppmKlam}
 \end{align} 
 where $p_\pm(\bx;0)=P_\pm(\bx)$ is given by
 \begin{equation*}
  P_\pm(\bx)\ \equiv\ \frac{1}{\sqrt{2}}\ \Big[\ \frac{\overline{\lambda_\sharp}}{|\lambda_\sharp|} \frac{\mathfrak{z}_2}{|\mathfrak{z}_2|}\ P_1(\bx)
 \pm P_2(\bx)\ \Big]\ ,
 \label{Ppm-def}\end{equation*}
 and  $\Phi_\pm(\bx;0)=\Phi_\pm(\bx)$ is given by
 \begin{equation}
  \Phi_\pm(\bx)\ \equiv\ \frac{1}{\sqrt{2}}\ \Big[\ \frac{\overline{\lambda_\sharp}}{|\lambda_\sharp|} \frac{\mathfrak{z}_2}{|\mathfrak{z}_2|}\ \Phi_1(\bx)
 \pm \Phi_2(\bx)\ \Big]\ .
 \label{Phi_pm-def}\end{equation}
Finally, $\lambda \mapsto\varphi_{\pm}(\bx;\lambda)$ are real analytic satisfying the bound 
 $|\D_\bx^\aleph\varphi_\pm(\bx;\lambda)|\leq C' \lambda$ for all $\bx\in\Lambda_h$, where $\aleph=(\aleph_1,\aleph_2),\ |\aleph|\le2$.
\end{enumerate}
\end{proposition}

\nit{N.B.} We wish to point out that the subscripts $\pm$ have a different meaning here than in \cites{FW:12,FW:14}. In \cites{FW:12,FW:14},
$E_\pm(\bk)$ denote ordered eigenvalues, $E_-(\bk)\le E_+(\bk)$ (Lipschitz continuous) with corresponding eigenstates $\Phi_\pm(\bx;\bk)$ (discontinuous at $\bk=\bK$); 
 see Definition \ref{dirac-pt-defn} and Corollary \ref{Phibstar-bstar+1}. 
In Proposition \ref{directional-bloch} and throughout this paper $E_\pm(\lambda)$ and $\Phi_\pm(\bx;\lambda)$ refer to smooth parametrizations in $\lambda$ of Floquet-Bloch eigenvalues and eigenfunctions of the spectral bands, which intersect at energy $E_\star$. 

\begin{proof}[Proof of Proposition \ref{directional-bloch}]
We present a proof along the lines of Theorem 3.2 in \cites{FW:14}; see also \cites{Friedrichs:65,kato1995perturbation}.
 The $\bk-$ pseudo-periodic Floquet-Bloch modes can be expressed in the form 
$\Phi(\bx;\bk)=e^{i\bk\cdot\bx}p(\bx;\bk)$,
 where  $p(\bx;\bk)$ is $\Lambda_h-$ periodic. For $\bk=\bK+\lambda\ktilde_2$, 
 consider the family of eigenvalue problems, parametrized by $|\lambda|\le1/2$:
\begin{align}
& H_V(\bK+\lambda\ktilde_2)\ p(\bx;\lambda)\ =\ E(\lambda)\ p(\bx;\lambda)\ ,\label{Hk-evp}\\
&p(\bx+\bv;\lambda)=p(\bx;\lambda),\ \ \textrm{for all}\ \bv\in \Lambda_h\ ,
\label{psi-per}\end{align}
where  $H_V(\bk)\ \equiv\  -\left(\nabla_\bx + i\bk\right)^2\ +\ V(\bx)$.
Degenerate perturbation theory of the double eigenvalue $E_\star$ of $H_V(\bK)$, yields eigenvalues: $E_\pm(\lambda) =\ E_\star\ +\ E_\pm^{(1)}(\lambda)$, where
\begin{align}
\ E_\pm^{(1)}(\lambda)\ \equiv\  \pm\ |\lambda_\sharp|\ |\ktilde_2|\ \lambda +\ \mathcal{O}(\lambda^2) ;\ \ \text{see\ \cites{FW:12}.}
 \label{mu-exp}
  \end{align}
 Denote by $Q_\perp$ the  projection onto the orthogonal complement of ${\rm span}\{P_1,P_2\}$.
 Then,
 \[ R_\bK(E_\star)\ \equiv\ \left( H_V(\bK)\ - E_\star\ I \right)^{-1}:\ Q_\perp L^2(\R^2/\Lambda_h)\to Q_\perp L^2(\R^2/\Lambda_h)\]
 is bounded.
 Furthermore, via Lyapunov-Schmidt reduction analysis of the periodic eigenvalue problem 
\eqref{Hk-evp}-\eqref{psi-per} we obtain, corresponding to the eigenvalues \eqref{mu-exp},  the $\Lambda_h-$ periodic  eigenstates:
\begin{align*}
 p_\pm(\bx;\lambda)\ 
   &=\  
\left( I + 
R_\bK(E_\star) Q_\perp 
\left(2i\lambda\ \ktilde_2\cdot\left(\nabla+i\bK\right)\right)  \right) \times \\
&\quad \left( \alpha(\lambda)\ P_1(\bx)\ +\ \beta(\lambda)\ P_2(\bx) \right) \ 
+\  \mathcal{O}_{H^2(\R^2/\Lambda_h)}\left(\lambda(|\alpha|^2+|\beta|^2)^{\frac{1}{2}}\right).
\end{align*}
Here,  the pair  $\alpha(\lambda), \beta(\lambda)$ satisfies the homogeneous system:
\begin{align*}
 \mathcal{M}(E^{(1)},\lambda)\  \left(\begin{array}{c} \alpha \\ { }\\ \beta\end{array}\right)\ &=\ 0\ , \text{ \ \ where } \\
 \mathcal{M}(E^{(1)},\lambda) &\equiv\  
\left(\begin{array}{cc} 
E^{(1)} + \mathcal{O}\left( \lambda^2\right)&  
-\overline{\lambda_\sharp}\ \lambda\ \mathfrak{z}_2\ + \mathcal{O}\left(\lambda^2\right)  \\
&\nn\\
-\lambda_\sharp\ \lambda\ \overline{\mathfrak{z}_2} +
 \mathcal{O}\left(\lambda^2\right) 
 &
 E^{(1)}  + \mathcal{O}\left(\lambda^2\right)
 \end{array}\right) ;
 \end{align*}
 see \cites{FW:12}.
For $E^{(1)}=E^{(1)}_j(\lambda),\ j=\pm$, normalized solutions, $p_j(\bk;\lambda),\ j=\pm$, are obtained by choosing:
{\footnotesize{
 \begin{align*}
 \left(\begin{array}{c} \alpha_+(\lambda) \\  \\ \beta_+(\lambda)\end{array}\right) &= 
 \left(\begin{array}{c} \frac{1}{\sqrt{2}}\ \frac{\overline{\lambda_\sharp}}{|\lambda_\sharp|}\ 
 \frac{ \mathfrak{z}_2}{| \mathfrak{z}_2|} + 
 \mathcal{O}\left(\lambda\right)\\ \\
 + \frac{1}{\sqrt{2}}\ +\ \mathcal{O}\left(\lambda\right)
 \end{array}\right) ,\quad 
 \left(\begin{array}{c} \alpha_-(\lambda) \\  \\ \beta_-(\lambda)\end{array}\right)\ &=\ 
 \left(\begin{array}{c} \frac{1}{\sqrt{2}}\ \frac{\overline{\lambda_\sharp}}{|\lambda_\sharp|}\ 
 \frac{ \mathfrak{z}_2}{| \mathfrak{z}_2|} + 
 \mathcal{O}\left(\lambda\right)\\ \\
 - \frac{1}{\sqrt{2}} + \mathcal{O}\left(\lambda\right)
 \end{array}\right) .
 \end{align*}
 }}

Finally, we note that $\mathcal{M}(E^{(1)},\lambda)$ is analytic in the parameter $\lambda$. Therefore the eigenvalues $E^{(1)}_\pm(\lambda)$ and eigenvectors $( \alpha_\pm(\lambda), \beta_\pm(\lambda) )^T$ are analytic functions of $\lambda$; see, for example,  \cites{Friedrichs:65,kato1995perturbation}. It follows that
 $E_\pm(\lambda)$ and $p_\pm(\bx;\lambda)$ are bounded, real analytic functions of $\lambda\in\R$. This completes the proof of Proposition \ref{directional-bloch}. 
\end{proof}

\section{Model of a honeycomb structure with an edge}\label{zigzag-edges}

Let  $V(\bx)$ denote a honeycomb potential in the sense of Definition \ref{honeyV}. In this section we introduce a model of an edge in a honeycomb structure. A one-dimensional variant of this model was introduced and studied in 
\cites{FLW-PNAS:14,FLW-MAMS:15,Thorp-etal:15}.

Let $W\in C^\infty(\R^2)$ be real-valued  and satisfy the following properties:
\begin{enumerate}
\item[(W1)] ${W}$ is $\Lambda_h-$ periodic, {\it i.e.}  ${W}(\bx+\bv)={W}(\bx)$ for all $\bx\in\R^2$ and $\bv\in\Lambda_h$.  
\item[(W2)]  ${W}$ is odd, {\it i.e.} ${W}(-\bx)=-{W}(\bx)$.
\item[(W3)] $\vartheta_\sharp\equiv \left\langle \Phi_1,W\Phi_1\right\rangle_{L^2(\Omega_h)}\ne0$, with $\Phi_1$ as in Definition \ref{dirac-pt-defn}.
\end{enumerate} 
The non-degeneracy condition (W3) arises in the multiple scale perturbation theory of Section \ref{formal-multiscale}.

\nit Our model of a honeycomb structure with an edge is a smooth and slow interpolation between the Schr\"odinger Hamiltonians
 $H^{(\delta)}_{-\infty} = -\Delta_\bx + V(\bx) - \delta\kappa_\infty W(\bx)$
and
 $H^{(\delta)}_{+\infty} = -\Delta_\bx +  V(\bx) + \delta\kappa_\infty W(\bx)$, 
 which is transverse to a lattice direction, say $\vtilde_1$.  Here, $\kappa_\infty$ is a positive constant.
 This interpolation is effected by a {\it domain wall function}.

\begin{definition} \label{domain-wall-defn}
We call $\kappa(\zeta)\in C^{\infty}(\R)$ a  domain wall function if $\kappa(\zeta)$ tends to $\pm\kappa_\infty$ as $\zeta\to\pm\infty$, and 
$
\Upsilon_1(\zeta)=\kappa^2(\zeta)-\kappa_\infty^2$ and 
$ \Upsilon_2(\zeta)=\kappa'(\zeta)
$ 
satisfy:
{\small
\begin{equation}
\int_\R (1+|\zeta|)^a |\Upsilon_\ell(\zeta)| d\zeta<\infty\ \ \textrm{for some $a>5/2$ \ and }\ \int_\R |\D_\zeta\Upsilon_\ell(\zeta)|
d\zeta<\infty, \ \ \ell=1,2.
\label{kappa-hypotheses}
\end{equation}
}
Without loss of generality, we assume $\kappa_\infty>0$.
\end{definition}

\begin{remark}
 \label{kappa-description}
The technical hypotheses in \eqref{kappa-hypotheses} are required for the boundedness of {\it wave operators}  used in the proof of Proposition \ref{solve4beta}. See also Section 6.6 of \cites{FLW-MAMS:15} and, in particular, the application of Theorem 6.15. 
\end{remark}

  Our model of a honeycomb structure with an edge is the domain-wall modulated Hamiltonian:
\begin{equation*}
 \label{perturbed-ham}
 H^{(\delta)} = -\Delta_\bx +  V(\bx) + \delta\kappa\left(\delta\ktilde_2\cdot\bx\right) W(\bx) \ ,
\end{equation*}
where $\kappa(\zeta)$ is a domain wall function. 
 Suppose $\kappa(\zeta)$ has a single zero at $\zeta=0$. The ``edge'' is then given by 
 $\R\vtilde_1=\{\bx:\ktilde_2\cdot\bx=0\}$. 

We shall seek solutions of the eigenvalue problem
\begin{align}
&H^{(\delta)} \Psi = E\Psi, \label{schro-evp}\\
&\Psi(\bx+\vtilde_1)=e^{i\bK\cdot\vtilde_1}\Psi(\bx)\qquad \textrm{(propagation parallel to the edge,  $\R\vtilde_1$)},\label{pseudo-per}\\
&\Psi(\bx) \to 0\ \ {\rm as}\ \ |\bx\cdot\ktilde_2|\to\infty \qquad  \textrm{(localization tranverse to the edge,  $\R\vtilde_1$)}\label{localized} .
\end{align}

\nit In the next section we present  a formal asymptotic expansion of $\vtilde_1-$ edge states
 and in Section \ref{thm-edge-state} we formulate a rigorous theory.

\section{Multiple scales and effective Dirac equations}\label{formal-multiscale}

We re-express the eigenvalue problem \eqref{schro-evp}-\eqref{localized} in terms of an unknown function  $\Psi=\Psi(\bx,\zeta)$, depending on  fast ($\bx$) and slow ($\zeta=\delta\ktilde_2\cdot\bx$) 
spatial scales:
\begin{align}
&\Big[\ -\left(\nabla_\bx+\delta\ktilde_2\ \partial_\zeta\right)^2+V(\bx)\Big]\Psi
\ +\ \delta\kappa(\zeta)W(\bx)\Psi = E\Psi, \label{multi-schroedinger} \\
&\Psi(\bx+\vtilde_1,\zeta)=e^{i\bK\cdot\vtilde_1}\Psi(\bx,\zeta), \ \ {\rm and} \ \
\Psi(\bx,\zeta) \to 0\ \ {\rm as}\ \ \zeta\to\pm\infty. \label{multi-schroedinger-bc}
\end{align}
We seek a solution to \eqref{multi-schroedinger}-\eqref{multi-schroedinger-bc} in the form:
\begin{align}
E^{\delta}&=E^{(0)}+\delta E^{(1)}+\delta^2E^{(2)}+\ldots, \label{formal-E}\\
\Psi^{\delta}&= \psi^{(0)}(\bx,\zeta)+\delta\psi^{(1)}(\bx,\zeta)+\delta^2\psi^{(2)}(\bx,\zeta)+\ldots\ \label{formal-psi}.
\end{align}
The conditions \eqref{pseudo-per}, \eqref{localized} are encoded by requiring, for $i\ge0$, that 
\begin{align*} 
&\psi^{(i)}(\bx+\vtilde,\cdot)=e^{i\bK\cdot\vtilde}\psi^{(i)}(\bx,\cdot) \ \ \ \forall\ \vtilde\in\Lambda_h ,  \\
&\zeta\to\psi^{(i)}(\bx,\zeta)\in L^2(\R_\zeta).\end{align*}

Substituting the expansions \eqref{formal-E}-\eqref{formal-psi} in \eqref{multi-schroedinger} yields
\begin{align*}
 &\left[-\left(\Delta_{\bx}+2\delta\ \ktilde_2\cdot \nabla_{\bx}\ \partial_\zeta+\delta^2\ |\ktilde_2|^2\ \partial_\zeta^2+\ldots\right)
 +\left(V(\bx)+\delta\kappa(\zeta)W(\bx)\right) \right. \\
 &\quad \left.-\left(E^{(0)}+\delta E^{(1)}+\delta^2E^{(2)}+\ldots\right)\right]
\left(\psi^{(0)}+\delta\psi^{(1)}+\delta^2\psi^{(2)}+\ldots\right)=0.
\end{align*}
Equating terms of equal order in $\delta^j,\ j\ge0$, yields a hierarchy of equations governing $\psi^{(j)}(\bx,\zeta)$.

At order $\delta^0$ we have that $(E^{(0)},\psi^{(0)})$ satisfy
\begin{equation}
 \label{perturbed_schro_delta0}
 \begin{split}
 &\left(-\Delta_{\bx} +V(\bx)-E^{(0)}\right)\psi^{(0)} = 0, \\
 &\psi^{(0)}(\bx+\vtilde,\cdot)=e^{i\bK\cdot\vtilde}\psi^{(0)}(\bx,\cdot) \ \ \ \forall\ \vtilde\in\Lambda_h .
 \end{split}
\end{equation}
Equation \eqref{perturbed_schro_delta0} may be solved in terms of the orthonormal basis of the $L^2_\bK(\Omega)-$ nullspace of $H_V-E_{\star}$ in Definition \ref{dirac-pt-defn}, namely $\{\Phi_1,\Phi_2\}$.
Expansion \eqref{ppmKlam} in Proposition \ref{directional-bloch} suggests that a particularly natural orthonormal basis of the $L^2_{\bK}(\Omega)-$ nullspace of $H_V-E_{\star}$ is given by $\{\Phi_+,\Phi_-\}$, where
\begin{equation}
 \label{correct-basis}
\Phi_\pm(\bx) \equiv \frac{1}{\sqrt{2}}\left[\ \frac{\overline{\lambda_\sharp}}{|\lambda_\sharp|} \frac{ \mathfrak{z}_2}{| \mathfrak{z}_2|}\ \Phi_1(\bx)
 \pm \Phi_2(\bx)\ \right].
 \end{equation}
 Here $\lambda_\sharp$ is given in \eqref{lambda-sharp}, $ \mathfrak{z}_2=\smallktilde_2^{(1)}+i\smallktilde_2^{(2)}$ and $| \mathfrak{z}_2|=|\ktilde_2|$.
We therefore solve \eqref{perturbed_schro_delta0} with
\begin{equation}
E^{(0)}=E_\star,\qquad \psi^{(0)}(\bx,\zeta)=\alpha_+(\zeta)\Phi_+(\bx)
+\alpha_-(\zeta)\Phi_-(\bx).
\label{psi0-soln}
\end{equation}

Proceeding to order $\delta^1$ we find that $(E^{(1)},\psi^{(1)})$ satisfies
\begin{equation}
 \label{psi1-eqn}
 \begin{split}
 &\left(-\Delta_{\bx}+V(\bx)-E_{\star}\right)\psi^{(1)}(\bx,\zeta)=G^{(1)}(\bx,\zeta;\psi^{(0)}) + 
 E^{(1)} \psi^{(0)}, \\
 &\psi^{(1)}(\bx+\vtilde,\cdot)=e^{i\bK\cdot\vtilde}\psi^{(1)}(\bx,\cdot) \ \ \ \forall\ \vtilde\in\Lambda_h ,
 \end{split}
\end{equation}
 where
 \begin{align*}
 &G^{(1)}(\bx,\zeta;\psi^{(0)}) = G^{(1)}(\bx,\zeta;\alpha_+,\alpha_-)\ \nn \\ 
 &\quad \equiv\ 
 2\partial_\zeta\alpha_+\ \ktilde_2\cdot\nabla_\bx\Phi_++ 2\partial_\zeta\alpha_-\ \ktilde_2\cdot\nabla_\bx\Phi_- 
 -\kappa(\zeta)W(\bx) \left(\alpha_+\Phi_++ \alpha_-\Phi_- \right) .
\end{align*}
Viewed as an equation in $\bx$,  \eqref{psi1-eqn} is solvable if and only if its right hand side is   $L^2_\bK(\Omega;d\bx)-$ orthogonal to the nullspace of $H_V-E_{\star}$. This is expressible in terms of the two orthogonality conditions:
\begin{align}
-E^{(1)}\alpha_j&=2\left\langle\Phi_j,\ \ktilde_2\cdot\nabla_\bx\Phi_+\right\rangle \partial_\zeta\alpha_+\  +
 2\left\langle\Phi_j,\ktilde_2\cdot\nabla_\bx\Phi_-\right\rangle\ \partial_\zeta\alpha_-\nn\\
  &\qquad -\kappa(\zeta)\ \left[ \left\langle\Phi_j,W\Phi_+\right\rangle\alpha_++\left\langle\Phi_j,W\Phi_-\right\rangle\alpha_-\right],
  \qquad j=\pm .
 \label{alpha-j} \end{align}
We evaluate the inner products in \eqref{alpha-j}  using the following two propositions.
 
 \begin{proposition}\label{inner-prods-sharp}
 \begin{align}
 \left\langle\Phi_+, \ktilde_2\cdot\nabla_\bx\Phi_-\right\rangle_{L^2_\bK(\Omega)}
 &= 0 \label{ip12} \ , \\
 \left\langle\Phi_-, \ktilde_2\cdot\nabla_\bx\Phi_+\right\rangle_{L^2_\bK(\Omega)}
 &= 0 \label{ip21} \ , \\
 2\left\langle\Phi_+, \ktilde_2\cdot\nabla_\bx\Phi_+\right\rangle_{L^2_\bK(\Omega)}&=\ 
 + i|\lambda_\sharp|\ |\ktilde_2|\ ,
  \label{ip-aa+}\\
 2 \left\langle\Phi_-, \ktilde_2\cdot\nabla_\bx\Phi_-\right\rangle_{L^2_\bK(\Omega)}&=\ - i|\lambda_\sharp|\  |\ktilde_2|\ ,
 \label{ip-aa-}\end{align}
  The constant, $\lambda_\sharp\in\C$, is generically non-zero; see Theorem \ref{diracpt-thm}. 
 \end{proposition}

\begin{proof}
 Let $\mathfrak{z}_2=\smallktilde_2^{(1)}+i\smallktilde_2^{(2)}$.
 By (7.28)-(7.29) of \cites{FW:14} (see also \cites{FW:12}) we have:
 \begin{align}
 \left\langle\Phi_1, \ktilde_2\cdot\nabla_\bx\Phi_2\right\rangle_{L^2_\bK(\Omega)}
 &= \frac{i}{2}\ \overline{\lambda}_\sharp\ \mathfrak{z}_2\ ,\ \label{ip12_1} \\
 \left\langle\Phi_2, \ktilde_2\cdot\nabla_\bx\Phi_1\right\rangle_{L^2_\bK(\Omega)}
 &= \frac{i}{2}\  \lambda_\sharp\ \overline{\mathfrak{z}_2}\ \label{ip21_2} \ , \\
 \left\langle\Phi_b,\  \ktilde_2\cdot\nabla_\bx\Phi_b\right\rangle_{L^2_\bK(\Omega)}&=0,\ \ b=1,2.
 \label{ip-aa_1}\end{align}
Relations \eqref{ip12}-\eqref{ip-aa-} follow from the expressions for $\Phi_\pm$ in \eqref{correct-basis} and relations \eqref{ip12_1}-\eqref{ip-aa_1}.
  \end{proof}


 \begin{proposition}\label{inner-prods-W}
 Assume that $W(\bx)$ is real-valued, odd and $\Lambda_h-$ periodic.
 Let $\vartheta_\sharp\equiv \left\langle \Phi_1,W\Phi_1\right\rangle_{L^2_{\bK}(\Omega)}$. Then, 
 $  \vartheta_\sharp\in\R$ and 
 \begin{align}
\vartheta_\sharp\ =\  \left\langle \Phi_+, W\Phi_-\right\rangle_{L^2_\bK(\Omega)}\ &=\  \left\langle \Phi_-, W\Phi_+\right\rangle_{L^2_\bK(\Omega)} \label{ip-1W1} , \\
  \left\langle \Phi_+, W\Phi_+\right\rangle_{L^2_\bK(\Omega)}\ &=\  \left\langle \Phi_-, W\Phi_-\right\rangle_{L^2_\bK(\Omega)}=0\label{ip-1W2}  \ . 
  \end{align}
  Note that since $\Phi_+(\bx)=e^{i\bK_\star\cdot\bx}P_+(\bx)$ and $\Phi_-(\bx)=e^{i\bK_\star\cdot\bx}P_-(\bx)$,
  relations \eqref{ip-1W1} and \eqref{ip-1W2} hold with $\Phi_+$ and $\Phi_-$ replaced, respectively, by $P_+(\bx)$ 
   and $P_-(\bx)$.
 \end{proposition}
 
\begin{proof}
Equations \eqref{ip-1W1}-\eqref{ip-1W2} follow from the relations
 \begin{align}
\vartheta_\sharp \equiv \left\langle \Phi_1, W\Phi_1\right\rangle_{L^2_\bK(\Omega)}\ &=\  -\left\langle \Phi_2, W\Phi_2\right\rangle_{L^2_\bK(\Omega)} \label{ip-1W1_1} , \\
  \left\langle \Phi_1, W\Phi_2\right\rangle_{L^2_\bK(\Omega)}\ &=\  \left\langle \Phi_2, W\Phi_1\right\rangle_{L^2_\bK(\Omega)}=0\label{ip-1W2_1}  \ .
  \end{align}
 To prove \eqref{ip-1W1_1} and \eqref{ip-1W2_1}, we begin by recalling that $\Phi_2(\bx)=\overline{\Phi_1(-\bx)}$, $W$ is real-valued and $W(-\bx)=-W(\bx)$. 
  Since $W$ is real-valued, it is clear that $\vartheta_\sharp\in\R$. 
  Furthermore,
  \begin{align*}
   \left\langle\Phi_2,W\Phi_2\right\rangle_{L^2_\bK(\Omega)}&=
   \int_\Omega\overline{\Phi_2(\bx)}W(\bx)\Phi_2(\bx)d\bx
   =  \int_\Omega \Phi_1(-\bx)W(\bx)\overline{\Phi_1(-\bx)}d\bx\\
   &=  \int_\Omega \Phi_1(\bx)W(-\bx)\overline{\Phi_1(\bx)}d\bx = -\vartheta_\sharp \ .
   \end{align*}
This proves \eqref{ip-1W1}. To prove  \eqref{ip-1W2}, observe that
 \begin{align*}
 \left\langle\Phi_2,W\Phi_1\right\rangle_{L^2_\bK(\Omega)}&=\int_\Omega\overline{\Phi_2(\bx)}W(\bx)\Phi_1(\bx)d\bx
 =\int_\Omega\Phi_1(-\bx)W(\bx)\Phi_1(\bx)d\bx\\
 &=\int_\Omega\Phi_1(\bx)W(-\bx)\Phi_1(-\bx)d\bx
 = -\int_\Omega\Phi_1(\bx)W(\bx)\Phi_1(-\bx)d\bx\\
 &= -\overline{\left\langle\Phi_1,W\Phi_2\right\rangle_{L^2_\bK(\Omega)}} =- \left\langle\Phi_2,W\Phi_1\right\rangle_{L^2_\bK(\Omega)} \ .
 \end{align*}
 This completes the proof of Proposition \ref{inner-prods-W}.
 \end{proof}
 
Propositions \ref{inner-prods-sharp} and \ref{inner-prods-W} imply that the orthogonality conditions \eqref{alpha-j} reduce to the following eigenvalue problem  for $\alpha(\zeta)=(\alpha_+(\zeta), \alpha_-(\zeta))^T$:
 \begin{align}
 \left( \mathcal{D} - E^{(1)} \right) \alpha = 0 , \quad \alpha\in L^2(\R) \ . \label{m-dirac-eqn}
 \end{align}
Here, $\mathcal{D}$  denotes the 1D Dirac operator: 
\begin{equation}
 \label{multi-dirac-op}
 \mathcal{D} = -i|\lambda_\sharp| |\ktilde_2| \sigma_3 \partial_\zeta + \vartheta_\sharp \kappa(\zeta) \sigma_1,
 \quad \text{and} \quad \lamsharp \times \thetasharp \neq 0 \ .
\end{equation}

In Section \ref{dirac-solns} we prove that the eigenvalue problem \eqref{m-dirac-eqn} has an
  exponentially localized eigenfunction $\alpha_{\star}(\zeta)$ with corresponding (mid-gap) zero-energy eigenvalue $E^{(1)}=0$.
 Moreover, this eigenvalue has multiplicity one. We impose the normalization: $\|\alpha_\star\|_{L^2(\R)}=1$. 

Fix $(E^{(1)},\alpha)=(0,\alpha_\star)$. Then $\alpha_\star\in L^2(\R)$, $\psi^{(0)}(\bx, \zeta)$ is completely determined (up to normalization) and the solvability conditions \eqref{alpha-j} are satisfied. Therefore, the right hand side of \eqref{psi1-eqn} lies in the range of $H_V-E_\star: H^2_\bK\to L^2_\bK$, and we may invert $(H_V-E_{\star})$ obtaining
\begin{align}
 \psi^{(1)}(\bx,\zeta) &= \left(R(E_{\star})G^{(1)}\right)(\bx,\zeta) + \psi^{(1)}_h(\bx,\zeta)
 \equiv \psi^{(1)}_p(\bx,\zeta) + \psi^{(1)}_h(\bx,\zeta) ,\label{psi1p-def}
\end{align} 
where 
\begin{equation*}
R(E_\star) = \left(H_V-E_\star\right)^{-1}:P_\perp L^2_\bK\to P_\perp H^2_\bK\ 
\end{equation*}
and $P_\perp$ is the $L^2_\bK(\Omega)-$ projection on to the orthogonal complement of the kernel of $H_V-E_\star$, equal to ${\rm span}\{\Phi_+,\Phi_-\}$.
Here, $\psi^{(1)}_p$ denotes  a particular solution, and
\begin{equation*}
 \psi^{(1)}_h(\bx,\zeta) = \alpha^{(1)}_+(\zeta)\Phi_+(\bx)+\alpha^{(1)}_-(\zeta)\Phi_-(\bx)
\end{equation*} is a homogeneous solution.

Note that by exploiting the degrees of freedom coming from the $L^2_\bK-$ kernel of $H_V-E_\star$, we can continue the formal expansion to any order in $\delta$. Indeed, at $\mathcal{O}(\delta^\ell)$ for $\ell\geq2$, we have
\begin{align}
 \label{psi2-eqn}
& \left(-\Delta_\bx +V(\bx)-E_\star\right)\psi^{(\ell)}(\bx,\zeta)\\
&\quad = \left(\ 2 (\ktilde_2\cdot \nabla_{\bx})\  \partial_\zeta
-\kappa(\zeta)W(\bx)\right)\psi_h^{(\ell-1)}(\bx,\zeta) +E^{(\ell)}\psi^{(0)}(\bx,\zeta) \nn\\
&\qquad + G^{(\ell)}\left(\bx,\zeta;\psi^{(0)},\ldots,\psi^{(\ell-2)},\psi_p^{(\ell-1)},E^{(1)},\ldots,E^{(\ell-1)}\right) , \nn\\
&\psi^{(\ell)}(\bx+\vtilde,\cdot)=e^{i\bK\cdot\vtilde}\psi^{(\ell)}(\bx,\cdot) \ \ \ \forall\ \vtilde\in\Lambda_h , \nn
\end{align}
where, for the case $\ell=2$, 
\begin{equation}
 G^{(2)}(\bx,\zeta;\psi^{(0)},\psi_p^{(1)}) \\
 =  \left(\ 2 (\ktilde_2\cdot \nabla_{\bx})\ \partial_\zeta
-\kappa(\zeta)W(\bx)\ \right)\psi_p^{(1)}(\bx,\zeta) + |\ktilde_2|^2\ \partial^2_\zeta \psi^{(0)}(\bx,\zeta)\ . \label{G2def}
\end{equation}

As before, \eqref{psi2-eqn} has a solution if and only if the right hand side is $L^2_\bK(\Omega;d\bx)$-orthogonal to the functions $\Phi_j(\bx)$, $j=\pm$.
This solvability condition reduces to the inhomogeneous system:
\begin{equation}
\mathcal{D} \alpha^{(\ell-1)}(\zeta) =
\mathcal{G}^{(\ell)}\left(\zeta\right)+E^{(\ell)}\alpha_{\star}(\zeta), \quad \alpha^{(\ell-1)}\in L^2(\R) , \ \ {\rm where} 
\label{solvability_cond}
\end{equation}
\begin{equation}
\mathcal{G}^{(\ell)}(\zeta) = 
 \left( \begin{array}{c}
 \left\langle \Phi_+(\cdot),G^{(\ell)}(\cdot,\zeta;\psi^{(0)},\ldots,\psi^{(\ell-2)},\psi_p^{(\ell-1)},E^{(1)},\ldots,E^{(\ell-1)})) \right\rangle_{L^2_\bK(\Omega)} \\ 
  \left\langle \Phi_-(\cdot),G^{(\ell)}(\cdot,\zeta;\psi^{(0)},\ldots,\psi^{(\ell-2)},\psi_p^{(\ell-1)},E^{(1)},\ldots,E^{(\ell-1)})) \right\rangle_{L^2_\bK(\Omega)}
 \end{array} \right) . \label{ipG}
\end{equation} 
Solvability of the non-homogeneous Dirac system \eqref{solvability_cond}  in $L^2(\R)$, is ensured by imposing $L^2(\R)-$ orthogonality
of the right hand side of \eqref{solvability_cond} to $\alpha_\star(\zeta)$. This yields:
\begin{equation}
 \label{solvability_cond_E2}
E^{(\ell)} = -\inner{\alpha_{\star},\mathcal{G}^{(\ell)}}_{L^2(\R)}.
\end{equation}

Thus we obtain, at $\mathcal{O}(\delta^\ell)$, that $\psi^{(\ell)}=\psi_p^{(\ell)}+\psi_h^{(\ell)}$, where $\psi_p^{(\ell)}$ is a particular solution
of \eqref{psi2-eqn} and 
$\psi^{(\ell)}_h(\bx,\zeta) = \alpha^{(\ell)}_+(\zeta)\Phi_+(\bx)+\alpha^{(\ell)}_-(\zeta)\Phi_-(\bx)$
is a homogeneous solution. 

\nit {\bf Summary:} Given a zero-energy $L^2(\R)-$ eigenstate of the Dirac operator, $\mathcal{D}$ (see Section \ref{dirac-solns}), 
we can, to any polynomial order in $\delta$, construct a formal solution of the eigenvalue problem $H^{(\delta)}\psi=E\psi,\ \psi\in L^2_{\kpar=\bK\cdot\vtilde_1}$.

\subsection{Zero-energy eigenstate of the Dirac operator, $\mathcal{D}$\label{dirac-solns}}

\begin{proposition}\label{zero-mode}
Let $\kappa(\zeta)$ be a domain wall function (Definition \ref{domain-wall-defn}) and assume, without loss of generality, that $\vartheta_\sharp>0$. Then, 
\begin{enumerate}
\item The Dirac operator, $\mathcal{D}$, has a zero-energy eigenvalue, $E^{(1)}=0$,  with   exponentially localized solution given by:
\begin{align}
\alpha_\star(\zeta)\ &= \begin{pmatrix}\alpha_{\star,+}(\zeta) \\ \alpha_{\star,-}(\zeta) \end{pmatrix} = 
\ \gamma \begin{pmatrix} 1 \\ -i \end{pmatrix} 
e^{-\frac{\thetasharp}{|\lambda_\sharp| | \mathfrak{z}_2|}\ \int_0^\zeta \kappa(s) ds}\label{Fstar} \ .
\end{align}
Here,  $\gamma\in\C$ is any constant for which $\|\alpha_\star\|_{L^2}=1$. 
\item The solution \eqref{Fstar}, $\alpha_\star$, generates a leading order approximate (two-scale) edge state:
\begin{align}
&\Psi^{(0)}(\bx,\delta\ktilde_2\cdot\bx)\nn\\
&\qquad =  \alpha_{\star,+}(\delta\ktilde_2\cdot\bx)\Phi_+(\bx) + 
\alpha_{\star,-}(\delta\ktilde_2\cdot\bx)\Phi_-(\bx) \label{Psi0-a}\\
&\qquad = e^{i\bK\cdot\bx}\ \gamma \ (-1-i) \left[\ i \frac{\overline{\lambda_\sharp}}{|\lambda_\sharp|} \frac{ \mathfrak{z}_2}{| \mathfrak{z}_2|} P_1(\bx) -
 P_2(\bx)\ \right]\ 
 e^{-\frac{\vartheta_\sharp}{|\lambda_\sharp| | \mathfrak{z}_2|}\int_0^{\delta\ktilde_2\cdot\bx}\kappa(s)ds} \ .
\label{Psi0-d}\end{align}
$\Psi^{(0)}(\bx,\delta\bx)$ is propagating in the $\vtilde_1$ direction with parallel quasimomentum $\kparv$, and is exponentially decaying, $\ktilde_2\cdot\bx\to\pm\infty$,  in the transverse direction.
\end{enumerate}
\end{proposition}

\begin{proof}[Proof of Proposition \ref{zero-mode}]
The system \eqref{m-dirac-eqn} with energy $E^{(1)}=0$ may be written as:
\begin{align*}
\D_\zeta \alpha\ &=\ \frac{-i\thetasharp}{|\lambda_\sharp|| \mathfrak{z}_2|}\ \kappa(\zeta)\ \begin{pmatrix} 0 & 1\\ -1 & 0 \end{pmatrix}\ \alpha,\quad \
 \alpha = \begin{pmatrix} \alpha_+\\ \alpha_-\end{pmatrix},
\end{align*}
and has solutions:
\begin{align*}
\beta_1(\zeta)&=
\begin{pmatrix} 1\\ i\end{pmatrix} \
e^{\frac{\thetasharp}{|\lambda_\sharp| | \mathfrak{z}_2|}\ \int_0^\zeta \kappa(s) ds} \ , \quad \text{and} \quad
\beta_2(\zeta)=
 \begin{pmatrix} 1\\ -i\end{pmatrix} \
e^{\frac{-\thetasharp}{|\lambda_\sharp| | \mathfrak{z}_2|}\ \int_0^\zeta \kappa(s) ds} \ . 
\end{align*}
Since $\thetasharp>0$ and $\kappa(\zeta)\to\pm \kappa_\infty$ as $\zeta\to\pm\infty$, with $\kappa_\infty>0$, the solution $\beta_2(\zeta)$ 
decays as $\zeta\to\pm\infty$. Thus we set $\alpha_\star(\zeta)=\gamma \beta_2(\zeta)$, with constant $\gamma\in\C$ chosen so that $\norm{\alpha_\star}_{L^2(\R)}=1$.  This yields the expression for $\Psi^{(0)}(\bx,\delta\bx)$ in \eqref{Psi0-a}-\eqref{Psi0-d}, and completes the proof of Proposition \ref{zero-mode}.
\end{proof}

\begin{remark}[Topological Stability] 
The  zero-energy eigenpair, \eqref{Fstar},  is ``topologically stable'' or
``topologically protected''  in the sense that 
it (and hence the bifurcation of edge states, which it seeds) persists for any localized 
perturbation of  $\kappa(\zeta)$. Such perturbations may be large but  do not change  the  asymptotic behavior of $\kappa(\zeta)$ as 
 $\zeta\to\pm\infty$.
\end{remark}

\section{Existence of edge states localized along an edge}\label{thm-edge-state}
 
In this section we prove the existence of edge states for the eigenvalue problem:
\begin{align}
 \label{EVP_2}
 &H^{(\delta)} \Psi\ =\ E\ \Psi \ ,\ \ \Psi\in H_{\kparv}^2(\Sigma) \ , \quad \text{where} \\
 &H^{(\delta)} \equiv-\Delta+V(\bx)+\delta\kappa(\delta \ktilde_2\cdot\bx)W(\bx) \ .  \nn
\end{align}

We make the following assumptions:
\begin{enumerate}
\item[(A1)] $V$ is a  honeycomb potential in the sense of Definition \ref{honeyV} and $-\Delta+V$ has a Dirac point at $(\bK,E_\star)$; see Definition \ref{dirac-pt-defn} and
 the conclusions of  Theorem \ref{diracpt-thm}. In particular, 
the degenerate subspace of $H^{(0)}-E_\star$ has orthonormal basis of Floquet-Bloch modes $\{\Phi_1(\bx)\ ,\ \Phi_2(\bx)\}$ and 
 \begin{equation*}
\lambda_\sharp\ \equiv\   \sum_{\bfm\in\mathcal{S}} c(\bfm)^2\ \left(\begin{array}{c}1\\ i\end{array}\right)\cdot \left(\bK+\bfm\vec\bk\right) \ \ne\ 0 ; \ \textrm{see \eqref{lambda-sharp2}.}
\end{equation*} 
\item[(A2)] $W$ is real-valued and  $\Lambda_h-$ periodic,  odd  and non-degenerate; {\it i.e.} (W1), (W2) and (W3) of Section \ref{zigzag-edges} hold. In particular,
\begin{equation*}
\thetasharp\equiv\inner{\Phi_1,W\Phi_1}_{L_\bK^2}\ =\ \inner{\Phi_+,W\Phi_-}_{L_\bK^2}  \ \ne\ 0\ .
\end{equation*}
\item[(A3)] $\kappa(\delta \ktilde_2\cdot\bx)$ is a domain wall function in the sense of Definition \ref{domain-wall-defn}.
\end{enumerate}

\nit The following {\it spectral no-fold condition} plays a central role.

\begin{definition}\label{SGC}[Spectral \nofold condition]
Let $H_V=-\Delta+V(\bx)$, where $V$ is a honeycomb potential in the sense of Definition \ref{honeyV}. Further, let $(\bK,E_\star)$ be a Dirac point for $H_V$ in the sense of  Definition \ref{dirac-pt-defn}, in which we use the convention of labeling the dispersion maps by:
 $\bk\mapsto E_b(\bk)$, where 
 $b\in\{b_\star,b_{\star}+1\}\cup\{b\ge1: b\ne b_\star,\ b_{\star}+1\}$
 $\equiv \{-,+\}\cup\{b\ge1: b\ne -, +\}$.
  
    To the $\vtilde_1-$ edge, $\R\vtilde_1$, we associate the  ``$\ktilde_2 -$ slice at quasi-momentum $\bK$'', given by  the union over all $b\in \{-,+\}\cup\{b\ge1: b\ne -, +\}$ of the curves 
$ \{(\bK+\lambda\ktilde_2\ ,\ E_b(\bK+\lambda\ktilde_2) : \ |\lambda|\le\frac12\}$.

We say the band structure of $H_V$ satisfies the spectral \nofold condition for  the $\vtilde_1-$ edge or, equivalently at the Dirac point and along the $\ktilde_2 -$ slice, with constants $c_1, c_2, \mathfrak{a}_0$, and $\exponent\in (0,1)$ if the following holds:

 There is a  ``modulus'', $\omega(\mathfrak{a})$, which is continuous,  non-negative and increasing on  $0\le\mathfrak{a}<\mathfrak{a}_0$,  satisfying $\omega(0)=0$ and  \[ \omega(\mathfrak{a^\exponent})/\mathfrak{a}\to\infty\ \ {\rm as}\ \ \mathfrak{a}\to0,\] 
 %
 such that for all  $0\le\mathfrak{a}< \mathfrak{a}_0$:
\begin{align}
\mathfrak{a}^\exponent \le |\lambda|\le\frac12\quad &\implies\quad \Big|\ E_\pm(\bK+\lambda\ktilde_2) - E_\star\ \Big|\ \ge\ c_1\ \omega(\mathfrak{a}^\exponent) \ ,
\label{no-fold-over} \\
b\ne\pm, \ |\lambda|\le1/2 \quad &\implies \quad \Big| E_b(\bK+\lambda\ktilde_2)-E_\star \Big|\ \ge\  c_2\ (1+|b|) \ .
\label{no-fold-over-b}
\end{align}
\end{definition}
Our final assumption is 

\begin{enumerate}
\item[(A4)]
$-\Delta+V$ satisfies the spectral \nofold condition  at quasimomentum $\bK$ along the $\ktilde_2 -$ slice; see Definition \ref{SGC}.
\end{enumerate}

\begin{remark}\label{control-in-1d}
\begin{enumerate}
\item  Conditions \eqref{no-fold-over}-\eqref{no-fold-over-b} ensure that, restricted to the quasi-momentum slice $\lambda\mapsto \bK+\lambda\ktilde_2\in\mathcal{B}_h$, the dispersion curves which touch at the Dirac point $(\bK,E_\star)$
 do not ``fold over'' and attain energies within $c_1\cdot\omega(\mathfrak{a}^\exponent)$ of $E_\star$ for quasimomenta bounded away from $\bK$.
\item Dispersion curves of periodic Schr\"odinger operators on $\R^1$ (Hill's operators,\ $H=-\D_x^2+Q(x)$, where $Q(x+1)=Q(x)$) with ``Dirac points'' (see \cites{FLW-PNAS:14,FLW-MAMS:15}) always satisfy the natural 1D analogue of the spectral \nofold condition with $\omega(\mathfrak{a})=\mathfrak{a}$.  Dirac points  occur at quasi-momentum $k=\pm\pi$ and ODE arguments ensure that dispersion curves are monotone functions of $k$ away from $k=0,\pm\pi$. 
\item In Section \ref{zz-gap}
we prove that  $H_{\eps V}=-\Delta+\eps V$, where $V$ is a honeycomb potential, satisfies the \nofold condition along the zigzag slice ($\vtilde_1=\bv_1$) with modulus $\omega(\mathfrak{a})=\mathfrak{a}^2$, under the assumption that  $\eps V_{1,1}>0$ and $\eps$ is sufficiently small. 
 \end{enumerate}
\end{remark}

We now state a key result of this paper, giving sufficient conditions for the existence
of $\vtilde_1-$ edge states of $H^{(\delta)}$, for $\vtilde_1\in\Lambda_h$.

\begin{theorem}\label{thm-edgestate}
 Consider the $\vtilde_1-$ edge state eigenvalue problem,
\eqref{EVP_2},  
where $V(\bx)$, $W(\bx)$ and $\kappa(\zeta)$ satisfy assumptions (A1)-(A4).
  Then, there exist positive constants $\delta_0, c_0$ and a branch of solutions of \eqref{EVP_2},
   \[|\delta|\in(0,\delta_0)\longmapsto (E^\delta,\Psi^\delta)\in (E_\star-c_0\ \delta_0\ ,\ E_\star+c_0\ \delta_0)\times H^2_{\kparv}(\Sigma) , \] 
   such that the following holds:
\begin{enumerate}
\item $\Psi^\delta$ is well-approximated by a slow modulation of a linear combination of 
degenerate Floquet-Bloch modes $\Phi_+$ and $\Phi_-$ (\eqref{correct-basis}),
which is decaying transverse to the edge, $\Z\vtilde_1$:
\begin{align}
&\left\|\ \Psi^\delta(\cdot)\ -\
\left[\alpha_{\star,+}(\delta\ktilde_2 \cdot)\Phi_+(\cdot)+\alpha_{\star,-}(\delta\ktilde_2 \cdot )\Phi_-(\cdot)\right]\
\right\|_{H^2_{\kparv}}\
\lesssim\ \delta^{\frac12}\ , \label{TM_Psi-error}\\
& E^\delta = E_\star\ +\ E^{(2)}\delta^2\ +\  o(\delta^2), \label{TM_E-error}
\end{align}
where $E^{(2)}$ is obtained directly from \eqref{solvability_cond_E2}, \eqref{ipG} and \eqref{G2def}.  The implied constant in \eqref{TM_Psi-error} depends on $V$, $W$ and $\kappa$, but is independent of $\delta$. 
\item The amplitude vector,  $\alpha_\star(\zeta)=\left(\alpha_{\star,+}(\zeta),\alpha_{\star,-}(\zeta)\right)$, is an  $L^2(\R_\zeta)- $
normalized, topologically protected zero-energy eigenstate of the Dirac system \eqref{multi-dirac-op}: $\mathcal{D}\alpha_\star=0$ (see  Proposition \ref{zero-mode}).
\end{enumerate}
\end{theorem}

Perturbation theory for  $\kpar$ near $\bK\cdot\vtilde_1$  can be used to show the persistence of edge states
for parallel quasi-momenta near $\bK\cdot\vtilde_1$.

\begin{corollary}
 \label{vary_k_parallel}
 Fix $V$, $W$, $\kappa$ and $\delta$ as in Theorem \ref{thm-edgestate}. Then there exists $\eta_0\ll\delta$ such that for all $\kpar$ satisfying $|\kpar - \bK\cdot\vtilde_1| <\eta_0$, there exists an $H^2_\kpar(\Sigma)-$ eigenfunction with eigenvalue $\mu(\delta,\kpar)=\ E^\delta\ + \  \mu^{\delta}\ (\kpar-\bK\cdot\vtilde_1)\ +\mathcal{O}(|\kpar-\bK\cdot\vtilde_1|^2)$, where $E^{\delta}$ is given in \eqref{TM_E-error}, and $\mu^\delta$ is a constant, which is independent of $\kpar$.
\end{corollary}
%

\nit Zigzag edge states for  $\kpar$ in a neighborhood of $\kpar=\bK\cdot\vtilde_1=2\pi/3\ (\vtilde_1=\bv_1)$ and, by symmetry, in a neighborhood of  $\kpar=4\pi/3$ are indicated in Figure \ref{fig:k_parallel3}.

\subsection{Corrector equation}\label{corrector-equation}

We seek a solution of the eigenvalue problem  \eqref{EVP_2}, $\Psi^\delta$, in the form 
\begin{align}
\Psi^\delta &\equiv \psi^{(0)}(\bx,\delta\ktilde_2\cdot\bx)+\delta\psi^{(1)}(\bx,\delta\ktilde_2\cdot\bx)+\delta\eta^\delta(\bx) , \label{eta-def}\\
E^\delta &\equiv E_\star + \delta^2\mu^\delta \label{mu-def} .
\end{align}
Here, 
$\psi^{(0)}$ and $\psi_p^{(1)}$ are given by their respective multiple scale expressions \eqref{psi0-soln} and \eqref{psi1p-def}:
\begin{align*}
 \psi^{(0)}(\bx,\delta\ktilde_2\cdot\bx) &= \alpha_{\star,+}(\delta\ktilde_2\cdot\bx)\Phi_+(\bx) + \alpha_{\star,-}(\delta\ktilde_2\cdot\bx)\Phi_-(\bx)  , \nn \\
 \psi_p^{(1)}(\bx,\delta\bx) &= \left(R(E_\star)G^{(1)}\right)(\bx, \delta\ktilde_2\cdot\bx)  ,
\end{align*}
and $(\mu^\delta,\eta^\delta(\bx))$ is the corrector, to be constructed. We may assume throughout that $\delta\ge0$.

\begin{remark}
 \label{regularity}
 We shall make frequent use of the regularity of  $\Phi_+(\bx)$, $\Phi_-(\bx)$
and $\alpha_{\star}(\zeta)\equiv(\alpha_{\star,+}(\zeta),\alpha_{\star,-}(\zeta))^T$. 
 In particular,  $V \in C^{\infty}(\R^2/\Lambda)$ and elliptic
regularity theory imply that  $e^{-i\bK\cdot\bx}\Phi_\pm$ is $ C^{\infty}(\R^2/\Lambda)$,  and
 by Proposition  \ref{zero-mode},   $\alpha_\star(\zeta)$ and its derivatives with respect to $\zeta$ are all exponentially decaying as $|\zeta|\to\infty$.
\end{remark}

The following proposition lists useful bounds on $\psi^{(0)}$ and $\psi_p^{(1)}$.

\begin{proposition}[$H_\kparv^s(\Sigma_\bx)$ bounds on $\psi^{(0)}(\bx,\delta\ktilde_2\cdot\bx)$ and $\psi_p^{(1)}(\bx,\delta\ktilde_2\cdot\bx)$] 
 \label{lemma:psi_bounds} For all $s=1,2,\dots$, there exists $\delta_0>0$, such that if $0<|\delta|<\delta_0$,
then the leading order expansion terms $\psi^{(0)}(\bx,\delta\ktilde_2\cdot\bx)$ and $\psi^{(1)}_p(\bx,\delta\ktilde_2\cdot\bx)$ displayed in \eqref{psi0-soln} and
\eqref{psi1p-def} satisfy the bounds:
\begin{align*}
 \norm{\psi^{(0)}(\bx,\delta\ktilde_2\cdot\bx)}_{H_\kparv^s} + \norm{\D_\zeta^2\psi^{(0)}(\bx,\zeta)\Big|_{\zeta=\delta\ktilde_2\cdot\bx}}_{L^2_\kparv}\ &\approx |\delta|^{-1/2},\\ 
 \norm{\psi^{(1)}_p(\bx,\delta\ktilde_2\cdot\bx)}_{H_\kparv^s} &\lesssim |\delta|^{-1/2},  \\	
 \norm{\D_\zeta^2\psi^{(1)}_p(\bx,\zeta)\Big|_{\zeta=\delta\ktilde_2\cdot\bx}}_{L_\kparv^2}\ +\ 
 \norm{\D_\bx\D_\zeta\psi^{(1)}_p(\bx,\zeta)\Big|_{\zeta=\delta\ktilde_2\cdot\bx}}_{L_\kparv^2} &\lesssim |\delta|^{-1/2}.
\end{align*}
It follows that 
$\|\psi^{(0)}\|_{H^2_\kparv} \approx \delta^{-1/2} \ \gg\ 
 \|\delta \psi^{(1)}_p(\cdot,\delta\cdot) \|_{H^2_\kparv} =\mathcal{O}( \delta^{1/2})$.
\end{proposition}
%

The proof of Proposition \ref{lemma:psi_bounds} follows the approach taken in the proof of 
 Lemma  6.1  in Appendix G of \cites{FLW-MAMS:15}. We omit the details but  make two key technical remarks, that facilitate this proof.
 
\nit {\it Bound on $\|\Phi_b\|_{H^s}$:} Recall $\Phi_b(\bx;\bk)= e^{i\bk\cdot\bx}p_b(\bx;\bk)$, where
  $\|\Phi_b\|_{L^2(\Omega)}=\|p_b\|_{L^2(\Omega)}=1$. Now $p_b(\bx;\bk)$ satisfies 
$-\Delta p_b= 2i\bk\cdot\nabla p_b - |\bk|^2 p_b - V p_b+E_b(\bk) p_b$, where $V$ is bounded and smooth.
%
%
By 2D Weyl asymptotics  $|E_b(\bk)|\approx(1+|b|), \ b\gg1$ and therefore we have
 $\| \Delta p_b\|_{H^{s-1}}\le C_s (1+b)\ \|p_b\|_{H^s}$. Hence,  by elliptic theory
  $\| p_b\|_{H^{s+1}}\le C_s (1+b)\ \|p_b\|_{H^s}$ and  induction on $s\ge0$ yields
   $\| p_b\|_{H^{s}}\le C_s (1+b)^s\ $.
   
\nit  {\it Rapid decay of  $\inner{\Phi_b(\cdot;\bK),f(\cdot)}_{L^2(\Omega)}$:}
Using  $H^{(0)}\Phi_b(\bx;\bK) = E_b(\bK) \Phi_b(\bx;\bK)$, for sufficiently smooth $f$ we have:
$ \inner{\Phi_b(\cdot;\bK),f(\cdot)}_{L^2(\Omega)} =  (E_b(\bK))^{-M} \inner{\Phi_b(\cdot;\bK),[H^{(0)}]^M f(\cdot)}_{L^2(\Omega)}.$  Hence, for any $M\ge0$, $ | \inner{\Phi_b(\cdot;\bK),f(\cdot)}_{L^2(\Omega)} |\ \le\ C_M (1+b)^{-M}$.
\smallskip

   It remains to construct and bound the corrector $(\mu,\eta(\bx))$. Substitution of the expansion \eqref{eta-def} into the eigenvalue problem \eqref{EVP_2}, yields an equation for  $\eta(\bx) \in H^2_\kparv(\Sigma) $, which depends on $\mu$ and the small parameter $\delta$:
{\small
 \begin{align}
& \left(-\Delta_\bx + V(\bx) - E_\star\right)\eta(\bx) + \delta\kappa(\delta\ktilde_2\cdot\bx) W(\bx)\eta(\bx) - \delta^2\mu\ \eta(\bx) \nn \\
&\quad = \delta \Big(2\ktilde_2\cdot \nabla_\bx\ \D_\zeta-\kappa(\delta\ktilde_2\cdot\bx)W(\bx)\Big)\psi_p^{(1)}(\bx,\zeta)\Big|_{\zeta=\delta\ktilde_2\cdot\bx}\  +\ \delta\mu\psi^{(0)}(\bx,\delta\ktilde_2\cdot\bx) \nn\\
&\quad\qquad 
+ \delta |\ktilde_2|^2 \D_\zeta^2\psi^{(0)}(\bx,\zeta)\Big|_{\zeta=\delta\ktilde_2\cdot\bx}
+\delta^2\mu\psi_p^{(1)}(\bx,\delta\ktilde_2\cdot\bx) + \delta^2|\ktilde_2|^2\ \D_\zeta^2\ \psi_p^{(1)}(\bx,\zeta)\Big|_{\zeta=\delta\ktilde_2\cdot\bx}  .
\label{corrector-eqn1} 
\end{align}
}
To prove Theorem \ref{thm-edgestate}, we shall prove that \eqref{corrector-eqn1} has a solution $(\mu(\delta),\eta^\delta)$, with $\eta^\delta \in H^2_\kparv$ satisfying the bound
\begin{equation*}
 \|\delta \eta^\delta \|_{H^2_\kparv} \leq C\delta^{1/2} \ .
 \end{equation*}

\subsection{Decomposition of corrector, $\eta$, into near and far energy components}\label{rough-strategy}

Introduce the abbreviated notation, for  $|\lambda|\le1/2$:
%
%
\begin{align}\label{Eb_of_lambda}
E_b(\lambda)\ =\ 
\begin{cases}
E_b(\bK+\lambda\ktilde_2) & b_\star\notin\{b_\star, b_\star+1\} , \\
E_-(\lambda) & b=b_\star , \\
E_+(\lambda) & b=b_\star+1 ,
\end{cases}
\end{align}
and
\begin{align}\label{Phib_of_lambda}
\Phi_b(\bx;\lambda)\ =\ 
\begin{cases}
\Phi_b(\bx;\bK+\lambda\ktilde_2) & b_\star\notin\{b_\star, b_\star+1\} , \\
\Phi_-(\bx;\lambda) & b=b_\star , \\
\Phi_+(\bx;\lambda) & b=b_\star+1 .
\end{cases}
\end{align}

Define $\widetilde{f}_b(\lambda)=\inner{\Phi_b(\cdot, \lambda), f(\cdot)}_{L^2_{\kparv}}$. 
By Theorem \ref{fourier-edge}, any $\eta\in H^2_\kparv(\Sigma) $ has the representation
\begin{equation}
\label{eta-expansion}
 \eta(\bx) = \sum_{b\ge1}\ \int_{|\lambda|\le1/2}\ \Phi_b(\bx;\lambda)\ \widetilde{\eta}_b(\lambda)\ d\lambda \ . 
\end{equation}
Our strategy is to next derive a system of equations governing $\{\widetilde{\eta}_b(\lambda)\}_{b\geq1}$, which is formally equivalent to system \eqref{corrector-eqn1}. We then prove this system has a solution, which is  used to construct $\eta(\bx)$.  

Take the inner product of \eqref{corrector-eqn1} with $\Phi_b(\bx;\lambda)$, for $ b\ge1$, to obtain
\begin{align}
b\ge1:\quad &\left(\ E_b(\lambda)\  -\ E_\star\ \right) \widetilde{\eta}_b(\lambda) \nn \\ 
&\qquad\qquad + 
\delta \left\langle \Phi_b(\cdot;\lambda) ,
\kappa(\delta\ktilde_2\cdot)  W(\cdot) \eta(\cdot)  \right\rangle_{L^2_{\kparv}}\label{eta-b-system}\\
 & \qquad = \delta\widetilde{F}_b[\mu,\delta](\lambda)
 + \delta^2\ \mu\ \widetilde{\eta}_b(\lambda) \ ,\ |\lambda|\le1/2. \nn
 \end{align}
Here, $\widetilde{F}_b[\mu,\delta](\lambda),\ b\ge1$, is given by:
\begin{align}
& \widetilde{F}_b[\mu,\delta](\lambda)\equiv
\widetilde{F}^{1,\delta}_b(\lambda)\ + \mu \widetilde{F}^{2,\delta}_b(\lambda)\ +\ \delta\mu \widetilde{F}^{3,\delta}_b(\lambda)\ +\ \widetilde{F}^{4,\delta}_b(\lambda)\ +\ \delta \widetilde{F}^{5,\delta}_b(\lambda), 
\label{Fb-def}
\end{align}
where
{\small
\begin{align}
 \widetilde{F}^{1,\delta}_b(\lambda) & \equiv
\inner{\Phi_b(\bx,\lambda), (2\ktilde_2\cdot \nabla_\bx\ \D_\zeta -\kappa(\delta\ktilde_2\cdot\bx)W(\bx))
\psi^{(1)}_p(\bx,\zeta)\Big|_{\zeta=\delta\ktilde_2 \bx}}_{L_\kparv^2} ,\nn\\
\widetilde{F}^{2,\delta}_b(\lambda)& \equiv \inner{\Phi_b(\bx,\lambda),\psi^{(0)}(\bx,\delta\ktilde_2\cdot\bx)}_{L_\kparv^2} , \nn\\
\widetilde{F}^{3,\delta}_b(\lambda) & \equiv \inner{\Phi_b(\bx,\lambda),\psi^{(1)}_p(\bx,\delta\ktilde_2\cdot\bx)}_{L_\kparv^2} ,  \label{Fdef}\\
\widetilde{F}^{4,\delta}_b(\lambda) & \equiv \inner{\Phi_b(\bx,\lambda), \left. |\ktilde_2|^2\ \D_\zeta^2  \psi^{(0)}(\bx,\zeta)\right|_{\zeta=\delta \ktilde_2\cdot\bx}}_{L_\kparv^2} , \nn\\
\widetilde{F}^{5,\delta}_b(\lambda) & \equiv \inner{\Phi_b(\bx,\lambda), \left. |\ktilde_2|^2\ \D_\zeta^2 \psi^{(1)}_p(\bx,\zeta)\right|_{\zeta=\delta\ktilde_2\cdot\bx}}_{L_\kparv^2} . \nn
\end{align}
}

Recall the spectral \nofold condition ensuring that  $\delta/\omega(\delta^\exponent) \to 0$ as $\delta\to0$, where $\exponent>0$.   We next decompose $\eta(\bx)$ into its components with energies ``near'' and ``far'' from  the Dirac point:
\begin{align*}
\eta(\bx)\ &=\  \eta_{\rm near}(\bx)\ +\ \eta_{\rm far}(\bx), \ \ {\rm where}		
\end{align*}
\begin{align}
 \eta_{\rm near}(\bx)\   &\equiv\ 
 \sum_{b=\pm}\ \int_{|\lambda|\le1/2}\ \Phi_b(\bx;\lambda)\ \widetilde{\eta}_{b,{\rm near}}(\lambda)\ d\lambda \label{eta-near} , \\
 \eta_{\rm far}(\bx)\ &\equiv\ \sum_{b\geq1}\ \int_{|\lambda|\le1/2}\ \Phi_b(\bx;\lambda)\  \widetilde{\eta}_{b,{\rm far}}(\lambda)\ d\lambda,\qquad  {\rm and} \label{eta-far}
\end{align}
\begin{align*}
 \widetilde{\eta}_{\pm,{\rm near}}(\lambda)\ &\equiv\ \chi\left(|\lambda|\le\delta^\exponent\right)\ 
   \widetilde{\eta}_{\pm}(\lambda) , \\		
   \widetilde{\eta}_{b,{\rm far}}(\lambda)\ &\equiv\ \chi\left((\delta_{_{b,+}}+\delta_{_{b,-}})\delta^\exponent\le  |\lambda|\le \frac12 \right)\ 
   \widetilde{\eta}_b(\lambda),\ b\geq1 ; 
   \end{align*} 
   $\delta_{b,+}$ and $\delta_{b,-}$ are Kronecker delta symbols. 

%
   %
%
We rewrite system \eqref{eta-b-system} as two coupled subsystems:\ 
 a pair of equations, which governs the {\bf near energy components}: 
\begin{align}
&\left(E_{+}(\lambda)-E_{\star}\right)\widetilde{\eta}_{+,\rm near}(\lambda) \nn \\
&\qquad+ \delta\chi\Big(\abs{\lambda}\leq\delta^{\exponent}\Big)\inner{\Phi_{+}(\cdot,\lambda),\kappa(\delta\ktilde_2
\cdot)W(\cdot)\left[\eta_{\rm near}(\cdot)+\eta_{\rm far}(\cdot)\right]}_{L_\kparv^2(\Sigma)}
\label{near_cpt_1} \\
&\quad =\delta\chi\Big(\abs{\lambda}\leq\delta^{\exponent}\Big)
\widetilde{F}_{+}[\mu,\delta](\lambda) + \delta^2\mu\ \widetilde{\eta}_{+,{\rm near}}(\lambda), \nn \\
&\left(E_{-}(\lambda)-E_{\star}\right)\widetilde{\eta}_{-,\rm near}(\lambda) \nn \\
&\qquad+\delta\chi\left(\abs{\lambda}\leq\delta^{\exponent}\right)\inner{\Phi_{-}(\cdot,\lambda),\kappa(\delta\ktilde_2
\cdot)W(\cdot)\left[\eta_{\rm near}(\cdot)+  \eta_{\rm far}(\cdot)\right]}_{L_\kparv^2(\Sigma)}
\label{near_cpt_2} \\
&\quad =\delta\chi(\abs{\lambda}\leq\delta^{\exponent})
\widetilde{F}_{-}[\mu,\delta](\lambda) + \delta^2\mu\ \widetilde{\eta}_{-,{\rm near}}(\lambda), \nn
\end{align} 
coupled to an infinite system governing the {\bf far energy components}:
\begin{align}
&\left(E_b(\lambda)-E_{\star}\right)\widetilde{\eta}_{b,\rm far}(\lambda) 
+\delta\chi\Big(1/2\ge\abs{\lambda}\geq(\delta_{b,-}+\delta_{b,+}\Big)\delta^{\exponent}) \times \nn \\
&\qquad 
\left\langle\Phi_{b}(\cdot,\lambda),\kappa(\delta\ktilde_2
\cdot)W(\cdot)\left[\eta_{\rm near}(\cdot)+\eta_{\rm far}
(\cdot)\right] \right\rangle_{L^2_{\kparv}(\Sigma)} \label{far_cpts} \\
&\quad =\delta\chi\Big(1/2\ge\abs{\lambda}\geq(\delta_{b,-}+\delta_{b,+})\delta^{\exponent}\Big)
\widetilde{F}_{b}[\mu,\delta](\lambda) + \delta^2\mu\ \widetilde{\eta}_{b,\rm far}(\lambda),\ \ b\geq1. \nn
\end{align} 

We now systematically manipulate  \eqref{near_cpt_1}-\eqref{far_cpts} into  the form of a  band-limited Dirac system; see Proposition \ref{near_freq_compact}. This latter equation is then solved in Proposition \ref{solve4beta}. Since all steps  are reversible, this yields a solution $(\mu^\delta , \{\widetilde{\eta}^\delta_b(\lambda)\}_{b\geq1})$ of  \eqref{near_cpt_1}-\eqref{far_cpts}. Finally, $\eta^\delta\in H^2_\kparv(\Sigma)$, the  solution of corrector equation \eqref{corrector-eqn1},  is reconstructed from the amplitudes $\{\widetilde{\eta}^\delta_b(\lambda)\}_{b\geq1}$ using  \eqref{eta-expansion}.

\subsection{Construction of $\eta_{\rm far}=\eta_{\text{far}}[\eta_{\text{near}},\mu,\delta]$ and derivation of a closed system for $\eta_{\rm near}$}

We solve \eqref{far_cpts} for $\eta_{\rm far}$ as a functional of $\eta_{\rm near}$, and the parameters $\mu$ and $\delta$.
We then study the {\it closed} equation for $\eta_{\rm near}$ obtained by substitution of $\eta_{\rm far}$ into \eqref{near_cpt_1} and \eqref{near_cpt_2}.

\nit {\it It is in the construction of this map that we use assumption (A4), the spectral \nofold condition along the $\ktilde_2 -$ slice; Definition \ref{SGC}.
   We apply it in the form: There exists a modulus, $\omega(\mathfrak{a})$, and positive constants $\exponent$, $c_1$ and $c_2$, depending on $V$,  such that for all $\delta\ne0$ and sufficiently small:}
\begin{align}
\delta^\exponent \le |\lambda|\le\frac12\quad &\implies\quad \Big|\ E_\pm(\lambda) - E_\star\ \Big|\ \ge\ c_1\ \omega(\delta^{\exponent}),
\label{no-fold-over-A} \\
b\ne\pm:\ \  \ |\lambda|\le1/2 \quad &\implies \quad \Big| E_b(\lambda)-E_\star \Big|\ \ge\  c_2\ (1+|b|) \ .
\label{no-fold-over-b-A}
\end{align}

\nit The far energy system \eqref{far_cpts} may be written as
a fixed point system for $\widetilde\eta_{\rm far}=\{\widetilde{\eta}_{b,\rm far}(\lambda)\}_{b\ge1}$:
\begin{equation}
\widetilde{\mathcal{E}}_b[\widetilde{\eta}_{\rm far};\eta_{\rm near},\mu,\delta]\ =\
\widetilde{\eta}_{b,\rm far}\ ,\qquad  b\ge1, 
\label{fixed-pt1}
\end{equation}
where the mapping $\widetilde{\mathcal{E}}_b$ is given by
\begin{align*}
\widetilde{\mathcal{E}}_b[\phi;\psi,\mu,\delta](\lambda)&\equiv 
\delta^2\mu\ \frac{\widetilde{\phi}_{b,\textrm{far}}(\lambda)}{{E_b(\lambda)-E_{\star}}}
+\frac{\delta\ \chi\Big(1/2\ge\abs{\lambda}\geq(\delta_{b,-}+\delta_{b,+})\delta^{\exponent}\Big)}{E_b(\lambda)-E_{\star}} \times \\
&  \quad \left(-\inner{\Phi_{b}(\cdot,\lambda),\kappa(\delta \ktilde_2 \cdot)W(\cdot)\left[\psi(\cdot)+ \phi
(\cdot)\right]}_{L_\kparv^2} + \widetilde{F}_{b}[\mu,\delta](\lambda) \right) ,
\end{align*}
and
\begin{align*}
\phi(\bx)&= \sum_{b\ge1} \int_{|\lambda| \leq 1/2} \chi\Big(\abs{\lambda}\geq(\delta_{b,-}+\delta_{b,+})\delta^{\exponent}\Big)\
\widetilde{\phi}_b(\lambda) \Phi_b(\bx;\lambda)\ d\lambda \\
& =\
\sum_{b\ge1} \int_{|\lambda| \leq 1/2} \widetilde{\phi}_{b,{\rm far}}(\lambda) \Phi_b(\bx;\lambda)\ d\lambda\ .
\end{align*}
Equivalently,  
\begin{equation}
\mathcal{E}[\eta_{\rm far};\eta_{\rm near},\mu,\delta]\ =\ \eta_{\rm far}\ .
\label{fixed-pt-notilde}
\end{equation}
For fixed $\mu$, $\delta$ and band-limited $\eta_{\rm near}$:
\begin{equation}
\widetilde{\eta}_{\pm,\rm near}(\lambda)\ =\ \chi\left(\abs{\lambda}\leq \delta^\exponent\right)
\widetilde{\eta}_{\pm,\rm near}(\lambda)
 , \label{near-def}\end{equation}
we seek a solution $\{\widetilde{\eta}_{b,\rm far}(\lambda)\}_{b\ge1}$, supported at energies
bounded away from $E_\star$:
\begin{equation}
\widetilde{\eta}_{b,\rm far}(\lambda)\\ =\
\chi\left(\abs{\lambda}\ge(\delta_{b,-}+\delta_{b,+})\delta^\exponent\right) \widetilde{\eta}_{b,\rm far}(\lambda)
 , ~~~b\ge1 .\label{far-def}\end{equation}

 \nit Introduce the Banach spaces of functions  limited to ``far'' and ``near'' energy regimes:
 \begin{align*}
  L^2_{{\rm near}, \delta^\exponent}(\Sigma) &\equiv\
  \left\{ f\in L_\kparv^2(\Sigma) : \widetilde{f}_b(\lambda)\ \textrm{satisfies \eqref{near-def}}\right\} , \\	
  L^2_{{\rm far}, \delta^\exponent}(\Sigma) &\equiv\
  \left\{ f\in L_\kparv^2(\Sigma) : \widetilde{f}_b(\lambda)\ \textrm{satisfies \eqref{far-def}}\right\} , 
\end{align*}
Near- and far- energy Sobolev spaces $H^s_{\rm far}(\Sigma) $ and 
$H^s_{\rm near}(\Sigma) $ are analogously defined.
The corresponding open balls of radius $\rho$ are given by: 
 \begin{align*}
  B_{{\rm near},\delta^\exponent}(\rho) &\equiv\
  \left\{ f\in  L^2_{{\rm near}, \delta^\exponent} : \|f\|_{L_\kparv^2}<\rho \right\} ,  \\	
 B_{{\rm far},\delta^\exponent}(\rho) &\equiv\
  \left\{ f\in  L^2_{{\rm far}, \delta^\exponent} : \|f\|_{L_\kparv^2}<\rho \right\} . 
\end{align*}

\nit Using (A4) that $H^{(0)}=-\Delta+V$ satisfies the \nofold condition for the $\vtilde_1 -$ edge, we deduce:
 
\begin{proposition}\label{fixed-pt} 
\begin{enumerate}
\item For any fixed $M>0, R>0$, there exists a positive number, $\delta_0\le1$,  such that for all $0<\delta<\delta_0$, 
equation \eqref{fixed-pt-notilde}, or equivalently, the system \eqref{fixed-pt1}, has a unique solution
\begin{align*}
&(\eta_{\text{near}},\mu,\delta)\in B_{{\rm near},\delta^\exponent}(R)\times\{|\mu|<M\}\times\{0<\delta<\delta_0\}
\nn\\
&\qquad\qquad \mapsto\ \eta_{\rm far}(\cdot;\eta_{\rm near},\mu,\delta)=
\mathcal{T}^{-1}\widetilde{\eta}_{\text{far}}\in B_{{\rm far},\delta^\exponent}(\rho_\delta) ,\quad 
\rho_\delta=\mathcal{O}\left(\frac{\delta^{\frac12}}{\omega(\delta^\exponent)}\right).
\end{align*}
\item The mapping $(\eta_{\text{near}},\mu,\delta)\mapsto \eta_{\rm far}(\cdot;\eta_{\rm near},\mu,\delta)\in H_\kparv^2$ is 
Lipschitz in $(\eta_{\rm near},\mu)$ with:
\begin{align}
&\left\|\ \eta_{\text{far}}[\psi_1,\mu_1,\delta] -  \eta_{\text{far}}[\psi_2,\mu_2,\delta]\ \right\|_{H_\kparv^2(\Sigma)} \nn \\
&\qquad \leq \ C' \ \frac{\delta}{\omega(\delta^\exponent)} \ \Big(\norm{\psi_1-\psi_2}_{H_\kparv^2} + \abs{\mu_1-\mu_2} \Big), \nn \\
&\norm{\ \eta_{\text{far}}[\eta_{\text{near}};\mu,\delta]\ }_{H_\kparv^2} 
 \le\ C''\left(\ 
\frac{\delta}{\omega(\delta^\exponent)}\norm{\eta_{\text{near}}}_{H_\kparv^2}+\frac{\delta^{\frac12}}{\omega(\delta^\exponent)} \right)\ . \label{eta-far-bound}
\end{align}
The constants $C'$ and $C''$ depend only on $M, R$ and $\exponent$.
\item 
The mapping $(\eta_{\text{near}},\mu,\delta)\mapsto\eta_{\text{far}}[\eta_{\text{near}},\mu,\delta]$ 
satisfies:
\begin{equation}
\label{eta_far_affine}
\eta_{\text{far}}[\eta_{\text{near}},\mu,\delta](\bx) = [A\eta_{\text{near}}](\bx;\mu,\delta) + \mu
B(\bx;\delta) + C(\bx;\delta).
\end{equation} 
For $\eta_{\text{near}}\in B_{{\rm near}}(R)$ we have:
\begin{align*}
&\left\| [A\eta_{\text{near}}](\cdot,\mu_1,\delta) - [A\eta_{\text{near}}](\cdot,\mu_2,\delta)\right\|_{H_\kparv^2}
\le \ C'_{M,R}\ \frac{\delta}{\omega(\delta^\exponent)}\ |\mu_1-\mu_2|, \\		
 &\norm{[A\eta_{\text{near}}](\cdot;\mu,\delta)}_{H_\kparv^2} 
 \leq \frac{\delta}{\omega(\delta^\exponent)}
\norm{\eta_{\text{near}}}_{H_\kparv^2},   \\		
 &\norm{B(\cdot;\delta)}_{H_\kparv^2} \leq\frac{\delta^\frac12}{\omega(\delta^\exponent)}, \ \ \text{and} \ \ 
 \norm{C(\cdot;\delta)}_{H_\kparv^2} \leq \frac{\delta^\frac12}{\omega(\delta^\exponent)}.  
\end{align*}
\item We may extend $\eta_{\rm far}[\cdot;\eta_{\rm near},\mu,\delta]$ to be defined on the half-open interval
$\delta\in[0,\delta_0)$
by defining
$\eta_{\rm far}[\eta_{\rm near},\mu,\delta=0]=0$. Then,  by \eqref{eta-far-bound}
$\eta_{\rm far}[\eta_{\rm near},\mu,\delta]$ is continuous at $\delta=0$.
\end{enumerate}
\end{proposition}

\begin{remark}\label{frak-e-remark}[Remarks on the proof of Proposition \ref{fixed-pt}]
The proof follows that of Corollary 6.4 in \cites{FLW-MAMS:15}, with changes that we now discuss. 
\begin{enumerate}
\item[(a)] The fixed point equation \eqref{fixed-pt1} for $\psi_{\rm far}$ is of the form: 
\begin{align}
\eta_{\rm far}\ &=\ \mathcal{Q}_\delta \eta_{\rm far}\ + \ \frac{\delta\ \chi(\abs{\lambda}\geq(\delta_{b,-}+\delta_{b,+})\delta^{\exponent})}{E_b(\lambda)-E_{\star}} \  \times \nn \\
&\qquad\qquad \left( - \inner{\Phi_{b}(\cdot,\lambda),\kappa(\delta \ktilde_2 \cdot)W(\cdot)\ \eta_{\rm near}(\cdot)}_{L_\kparv^2}
+\widetilde{F}_{b}[\mu,\delta](\lambda) \right) ,\label{fixed-pt2}
\end{align}
where $\mathcal{Q}_\delta$ is bounded and linear on   $H_\kparv^2$ and defined by:
 \begin{align}
\widetilde{\left[\ \mathcal{Q}_\delta \phi\ \right]}_b(\lambda)\ &\equiv\ -\delta\
 \frac{ \chi(\abs{\lambda}\geq(\delta_{b,b_{\star}}+\delta_{b,b_{\star}+1})\delta^{\exponent})}{E_b(\lambda)-E_{\star}}
\inner{\Phi_{b}(\cdot,\lambda),\kappa(\delta \ktilde_2 \cdot)W(\cdot)\ \phi
(\cdot)}_{L_\kparv^2}\nn\\
&\qquad + \delta^2\mu\ \frac{\widetilde{\phi}_b(\lambda)}{{E_b(\lambda)-E_{\star}}}\ .
\label{tQ-def} \end{align}
 To construct the mapping $(\eta_{\rm near},\mu,\delta)\mapsto\eta_{\rm far}[\eta_{\rm near},\mu,\delta]$ and obtain the conclusions of Proposition \ref{fixed-pt}  it is convenient to solve \eqref{fixed-pt2} via the contraction mapping principle. 
Thus we need to bound the operator norm of $\mathcal{Q}_\delta$ and we find from \eqref{tQ-def}
 that $\mathcal{Q}_\delta$  maps $L^2_{{\rm near},\delta^\exponent}$ to $H^2_{{\rm near},\delta^\exponent}$
  with norm bounded by $ {\rm constant}\times\mathfrak{e}(\delta)$, where 
 \begin{align}
\mathfrak{e}(\delta)\equiv \sup_{b=\pm}\ \ \sup_{|\delta|^\exponent\le|\lambda|\le\frac12}\ 
 \frac{|\delta|}{|E_b(\lambda)-E_{\star}|}\ +\ \sup_{ b\ge1,\ b\ne\pm}\ (\ 1+|b|\ ) \sup_{0\le|\lambda|\le\frac12}
 \frac{|\delta|}{|E_b(\lambda)-E_{\star}|} .
\label{frak-e-def} \end{align}
 The spectral \nofold condition hypothesis \eqref{no-fold-over-A}-\eqref{no-fold-over-b-A}
   implies that 
 \begin{equation}
 \mathfrak{e}(\delta) \lesssim\ \frac{|\delta|}{\omega(\mathfrak{\delta^{\exponent}})\ c_1(V)}\ +\ \frac{|\delta|}{ c_2(V)} \ , 
 \label{frak-e-bound}
\end{equation}
which tends to zero as $\delta$ tends to zero. 
 Hence, the contraction mapping principle can be applied on the ball
  $B_{{\rm far},\delta^\exponent}(\rho_\delta),\ \rho_\delta=\mathcal{O}({\delta^\frac12}/{\omega(\delta^\exponent)})$. 
 %
 %
\item[(b)] We note that although $\Sigma$ is a two-dimensional region, since $\Sigma$ is  unbounded in only one direction, estimates on $H^2_\kpar(\Sigma)$ have the same scaling behavior in the parameter $\delta$ as in the 1D study \cites{FLW-MAMS:15}.
\end{enumerate}
\end{remark}

\subsection{Analysis of the closed system for $\eta_{\rm near}$\label{subsec:near_freq}}

Substitution of  $\eta_{\text{far}}[\eta_{\text{near}},\mu,\delta]$ into the system \eqref{near_cpt_1}-\eqref{near_cpt_2} yields a closed system for $(\eta_{\text{near}},\mu)$, which depends on the parameter $\delta\in[0,\delta_0)$. 
In this section we show, by careful rescaling and expansion of terms, that the equation for $\eta_{\rm near}$ may be rewritten as a Dirac-type system. We then solve  this system in Section \ref{analysis-blDirac}.
 Recall the abbreviated notation: $E_b(\lambda)$ and $\Phi_b(\bx;\lambda)$,  introduced in \eqref{Eb_of_lambda}-\eqref{Phib_of_lambda}.

%
Since both the spectral support of $\eta_{\text{near}}$ (parametrized by $\bK+\lambda\ktilde_2$, with $|\lambda|\le\delta^\exponent$), and size of the domain wall perturbation, $\mathcal{O}(\delta)$,
tend to zero as $\delta\to0$, it is natural to scale in such a way as to obtain an order one limit. 
 We begin by introducing $\xi$,  a scaling of the quasi-momentum parameter, $\lambda$, 
 and  $\widehat{\eta}_{\pm,\rm near}\left(\xi\right)$, an expression for
  $\widetilde{\eta}_{\pm,\rm near}(\lambda)$ as a standard Fourier transform on $\R$:
\begin{equation}
\label{amplitude-rescaled}
\widehat{\eta}_{\pm,\rm near}\left(\xi\right)\ \equiv\  \widetilde{\eta}_{\pm,\rm near}(\lambda),\quad {\rm where}\quad  \xi\equiv \frac{\lambda}{\delta}.
\end{equation}
\nit By Proposition \ref{directional-bloch}:
$E_{\pm}(\lambda)-E_{\star} = 
\pm\abs{\lambda_\sharp}\ \abs{\ktilde_2}\ \delta\xi + E_{2,\pm}(\delta\xi)\ (\delta\xi)^2$,
where  $\abs{E_{2,\pm}(\delta\xi)} \lesssim 1$, for all $\xi$; see \eqref{EpmKlam}. 
Substitution of this expansion and the rescaling \eqref{amplitude-rescaled} into  \eqref{near_cpt_1}-\eqref{near_cpt_2}, and then canceling a factor of $\delta$ yields:
{\small
\begin{align}
&+|\lamsharp|\abs{ \ktilde_2}\ \xi\ \widehat{\eta}_{+,\rm near}(\xi)
+\chi(\abs{\xi}\leq\delta^{\exponent-1})\inner{\Phi_{+}(\cdot,\delta\xi),\kappa(\delta \ktilde_2
\cdot)W(\cdot)\eta_{\rm near}(\cdot)}_{L_\kparv^2} \nonumber\\
&\qquad=\chi(\abs{\xi}\leq\delta^{\exponent-1}) \widetilde{F}_{+}[\mu,\delta](\lambda) + \delta\mu\ \widehat{\eta}_{+,{\rm
near}}(\xi) - \delta E_{2,+}(\delta\xi)\xi^2\widehat{\eta}_{+,\rm near}(\xi) \label{near5} \\
&\qquad\qquad - \chi(\abs{\xi}\leq\delta^{\exponent-1})\inner{\Phi_{+}(\cdot,\delta\xi),\kappa(\delta \ktilde_2
\cdot)W(\cdot)\eta_{\rm far}[\eta_{\rm near},\mu,\delta](\cdot)}_{L_\kparv^2} ,\nn\\
&\nn\\
&-|\lamsharp|\abs{\ktilde_2}\ \xi\ \widehat{\eta}_{-,\rm near}(\xi)
+\chi(\abs{\xi}\leq\delta^{\exponent-1})\inner{\Phi_{-}(\cdot,\delta\xi),\kappa(\delta \ktilde_2
\cdot)W(\cdot)\eta_{\rm near}(\cdot)}_{L_\kparv^2} \nonumber\\
&\qquad=\chi(\abs{\xi}\leq\delta^{\exponent-1}) \widetilde{F}_{-}[\mu,\delta](\lambda) + \delta\mu\ \widehat{\eta}_{-,{\rm
near}}(\xi) -\delta E_{2,-}(\delta\xi)\xi^2\widehat{\eta}_{-,\rm near }(\xi) \label{near6} \\
&\qquad\qquad - \chi(\abs{\xi}\leq\delta^{\exponent-1})\inner{\Phi_{-}(\cdot,\delta\xi),\kappa(\delta \ktilde_2
\cdot)W(\cdot)\eta_{\rm far}[\eta_{\rm near},\mu,\delta](\cdot)}_{L_\kparv^2} \nn . 
\end{align}}
We next extract the dominant behavior, for $\delta$ small, of the inner products involving $\eta_{\rm near}$ by first  expanding $\eta_{\rm near}$ in terms of its spectral components near energy $E_\star=E_\pm(\lambda=0)$ plus a correction. To this end we apply Proposition \ref{directional-bloch}  to expand $p_\pm(\bx,\lambda)$ for $\lambda=\delta\xi$ small:
\begin{equation*}
p_{\pm}(\bx,\lambda) = P_{\pm}(\bx)\ +\ \varphi_\pm(\bx,\delta\xi),\ \ 
P_{\pm}(\bx) \equiv  \frac{1}{\sqrt{2}} \Big[\frac{\overline{\lambda_\sharp}}{|\lambda_\sharp|} \frac{\mathfrak{z}_2}{|\mathfrak{z}_2|} P_1(\bx) \pm P_2(\bx) \Big]\ \ {\rm where}
\end{equation*}
\begin{equation}
 \label{deltapbdd}
  \abs{\varphi_\pm(\bx,\delta\xi)} \leq \underset{\bx\in\Sigma, ~ \abs{\omega}\leq\delta^{\exponent}}{\sup}
\abs{\varphi_\pm(\bx,\omega)} \leq \delta^{\exponent},\ \ \abs{\xi}\leq\delta^{\exponent-1}\ .
\end{equation}
Thus, using  \eqref{eta-near} and that  $\Phi_\pm(\bx;\lambda) = e^{i(\bK+\lambda \ktilde_2)\cdot\bx}\ p_\pm(\bx;\lambda)$ (see \eqref{p_pm-def}), we obtain

{\small
\begin{align}
 \eta_{\text{near}}(\bx) &=  \int_{\abs{\lambda}\leq\delta^{\exponent}}
\Phi_{+}(\bx,\lambda)\widetilde{\eta}_{+,\text{near}}(\lambda) d\lambda
 +  \int_{\abs{\lambda}\leq\delta^{\exponent}}
 \Phi_{-}(\bx,\lambda)\widetilde{\eta}_{-,\text{near}}(\lambda)d\lambda \nn\\
 &=  \int_{\abs{\lambda}\leq\delta^{\exponent}}e^{i\bK\cdot\bx}e^{i\lambda\ktilde_2\cdot\bx}
 p_{+}(\bx,\lambda)\widehat{\eta}_{+,\text{near}}\left(\frac{\lambda}{\delta}\right)d\lambda \nn \\
&\quad+ \int_{\abs{\lambda}\leq\delta^{\exponent}}e^{i\bK\cdot\bx}e^{i\lambda\ktilde_2\cdot\bx}
p_{-}(\bx,\lambda)\widehat{\eta}_{-,\text{near}}\left(\frac{\lambda}{\delta}\right)d\lambda \nn\\
 &= \delta e^{i\bK\bx} P_+(\bx)  \int_{\abs{\xi}\leq\delta^{\exponent-1}}
 e^{i\delta\xi\ktilde_2\cdot\bx}\widehat{\eta}_{+,\text{near}}(\xi)d\xi 
 + \delta e^{i\bK\cdot\bx} \rho_{+}(\bx,\delta\ktilde_2\cdot \bx)\nonumber\\
&\quad+ \delta e^{i\bK\bx} P_-(\bx) \int_{\abs{\xi}\leq\delta^{\exponent-1}}
 e^{i\delta\xi\ktilde_2\cdot\bx}\widehat{\eta}_{-,\text{near}}(\xi)d\xi 
 + \delta e^{i\bK\cdot\bx} \rho_{-}(\bx,\delta \ktilde_2\cdot\bx)\nonumber\\
&=\delta e^{i\bK\cdot\bx}\left[ P_+(\bx)\ \eta_{+,\text{near}}(\delta \ktilde_2\cdot \bx) 
+ P_-(\bx)\ \eta_{-,\text{near}}(\delta \ktilde_2\cdot \bx) + \sum_{b=\pm}\rho_{b}(\bx,\delta\ktilde_2\cdot \bx)\right] ,\label{nearinxi}\\
 \label{rhodefn}
&{\rm where}\qquad\  \rho_{\pm}(\bx,\zeta) = \int_{\abs{\xi}\leq\delta^{\exponent-1}}e^{i\xi \zeta} \varphi_\pm(\bx,\delta\xi)
\widehat{\eta}_{\pm,\text{near}}(\xi)d\xi.
\end{align}
}

\nit We now expand the inner product in \eqref{near5}; the corresponding term in  \eqref{near6} is treated similarly. Substituting
\eqref{nearinxi} into the inner product in \eqref{near5} yields (using $\Phi_\pm(\bx;\bk)=e^{i\bk\cdot\bx}p_\pm(\bx;\bk)$) 
\begin{align}
 &\inner{\Phi_{+}(\cdot,\delta\xi),\kappa(\delta\cdot)W(\cdot)
\eta_{\rm near}(\cdot)}_{L_\kparv^2} \nn \\
&\equiv\  \inner{\Phi_{+}(\cdot,\bK+\delta\xi\ktilde_2),\kappa(\delta\cdot)W(\cdot)
\eta_{\rm near}(\cdot)}_{L_\kparv^2}\label{full_inner} \\
&=\  \delta \inner{e^{i\delta\xi \ktilde_2\cdot}p_+(\cdot,\delta\xi), 
P_+(\cdot)\ 
W(\cdot)\ \kappa(\delta\ktilde_2\cdot)\ {\eta}_{+,\rm near}(\delta \ktilde_2\cdot)}_{L^2(\Sigma)} \label{inner1}\\
&\qquad+ \delta \inner{e^{i\delta\xi \ktilde_2\cdot}p_{+}(\cdot,\delta\xi), 
P_-(\cdot)\ 
W(\cdot)\ \kappa(\delta\ktilde_2\cdot)\ {\eta}_{-,\rm near}(\delta \ktilde_2\cdot)}_{L^2(\Sigma)} \label{inner2}\\
&\qquad+ \delta \sum_{b=\pm}\ \inner{e^{i\delta\xi \ktilde_2\cdot}\ p_{+}(\cdot,\delta\xi),
W(\cdot)\ \kappa(\delta\ktilde_2\cdot)\ \rho_{b}(\cdot,\delta \ktilde_2\cdot)}_{L^2(\Sigma)} \ .\label{inner3}
\end{align} 

The inner product terms in \eqref{inner1}-\eqref{inner3} are each of the form:
\begin{equation}
\label{G_inner_def}
\mathcal{G}(\delta; \xi)
\equiv  \delta \int_\Sigma e^{-i\xi\delta\ktilde_2\cdot\bx} g(\bx,\delta\xi)\Gamma(\bx,\delta\ktilde_2\cdot\bx) d\bx,\ \ {\rm where}
\end{equation}
\begin{enumerate}
 \item [(IP1)] $g(\bx ,y)$ is a smooth function of $(\bx,y)\in \R^2/\Lambda_h\times\R$  and
 \item [(IP2)] $\bx\mapsto\Gamma(\bx,\zeta)$ is $\Lambda_h-$ periodic and $H^2(\Omega)$  with values in $L^2(\R_\zeta)$, {\it i.e.}
\begin{align}
\label{Gamma-conditions1}
&\Gamma(\bx+\bv,\zeta)\ =\ \Gamma(\bx,\zeta), \quad \text{for all } \bv\in\Lambda_h,  \\
&\sum_{j=0}^2\ \sum_{{\bf |c}|=j}\int_\Omega\ \left\|\D_\bx^{\bf c}\Gamma(\bx,\zeta)\right\|_{L^2(\R_\zeta)}^2\ d\bx\ <\ \infty. 
\label{Gamma-conditions2}
\end{align}
We denote this Hilbert space of functions by $\mathbb{H}^2$ with norm-squared, $\|\cdot\|_{\mathbb{H}^2}^2$, given in \eqref{Gamma-conditions2}.
It is easy to check that conditions \eqref{Gamma-conditions1}-\eqref{Gamma-conditions2} are satisfied for the cases $\Gamma=\Gamma(\zeta)=\kappa(\zeta)\eta_{\pm,{\rm near}}(\zeta)$ and $\Gamma=\Gamma(\bx,\zeta)=\kappa(\zeta)\rho_{\pm}(\bx,\zeta)$, where $\rho_\pm$  is defined in \eqref{rhodefn}.
\end{enumerate}

\nit To expand expressions of the form $G(\delta,\xi)$, we use:
\begin{lemma}
 \label{poisson_exp}
 Let $g(\bx ,y)$ and $\Gamma(\bx,\zeta)$ satisfy conditions (IP1) and (IP2), respectively.
 Denote by $\widehat{\Gamma}(\bx,\omega)$ the Fourier transform of $\Gamma(\bx,\zeta)$ with respect to the $\zeta-$ variable,
 given by
\begin{equation}
\widehat{\Gamma}(\bx,\omega)\ \equiv\ \lim_{N\uparrow\infty}\ \frac{1}{2\pi}\int_{|\zeta|\le N}e^{- i\omega \zeta}\Gamma(\bx,\zeta) d\zeta ,
\label{Gammahat-def}
\end{equation}
where the limit is taken in $L^2(\Omega\times\R_\omega;d\bx d\omega)$. Then,
\begin{equation}
 \label{poisson_app}
 \mathcal{G}(\delta; \xi) = \ 
 \sum_{n\in\mathbb{Z}} \int_\Omega e^{in \ktilde_2\cdot\bx} \widehat{\Gamma}\left(\bx,\frac{n}{\delta}+\xi\right) g(\bx,\delta\xi) d\bx,
\end{equation}
with equality holding in $L^2_{\rm loc}([-\xi_{\rm max},\xi_{max}];d\xi)$, for any fixed $\xi_{\rm max}>0$.
\end{lemma}

\nit We adapt the proof in  \cites{FLW-MAMS:15} (Lemma 6.5) for the 1D setting. We  require the following variant of the Poisson summation formula  in $L^2_{\rm loc}$.
\begin{theorem}\label{psum-L2}
Let $\Gamma(\bx,\zeta)$ satisfy (IP2). Denote by $\widehat{\Gamma}(\bx,\omega)$ the Fourier transform of $\Gamma(\bx,\zeta)$ with respect to the variable, $\zeta$; see \eqref{Gammahat-def}.
Fix an arbitrary $y_{\rm max}>0$, and introduce the parameterization of the cylinder $\Sigma$: $\bx=\tau_1\vtilde_1+\tau_2\vtilde_2, \ 0\le\tau_1\le1,\ \tau_2\in\R$. Then,
\begin{align*}
&\sum_{n\in\Z} e^{- iy(\tau_2+n)}\Gamma(\tau_1\vtilde_1+\tau_2\vtilde_2,\tau_2+n) 
= 2\pi\ \sum_{n\in\Z} e^{2\pi i n \tau_2}\widehat{\Gamma}\left(\tau_1\vtilde_1+\tau_2\vtilde_2,2\pi n+y \right)
\end{align*}
in $L^2\left([0,1]^2\times[-y_{\rm max},y_{\rm max}];d\tau_1d\tau_2\cdot dy\right)$.
\end{theorem}
 The 1D analogue of Theorem \ref{psum-L2} was proved in Appendix A of  \cites{FLW-MAMS:15}.
  Since the proof is very similar, we omit it.
 We also require

\begin{lemma}\label{interchange}
Let $F(\bx,y)$ and $F_N(\bx,y),\ N=1,2,\dots$, belong to $L^2(\Sigma\times[-y_{\rm max},y_{max}];d\bx dy)$. Assume that
\begin{equation*}
\left\| F_N-F \right\|_{L^2(\Sigma\times[-y_{\rm max},y_{max}];d\bx dy)}\ \to\ 0,\ \ {\rm as}\ \ N\to\infty\ .
\end{equation*}
Let $G\in L^2([-y_{\rm max}, y_{max}];d y)$. Then, in the $L^2(\Sigma;d\bx)$ sense, we have:
\begin{align*}
\lim_{N\to\infty}\int^{y_{\rm max}}_{-y_{\rm max}}F_N(\bx,y)G(y)\ dy\ &=\ 
\int^{y_{\rm max}}_{-y_{\rm max}}\lim_{N\to\infty}F_N(\bx,y)G(y)\ dy\ \\
&=\ \int^{y_{\rm max}}_{-y_{\rm max}} F(\bx,y)G(y)\ dy . 
\end{align*}
\end{lemma}

\begin{proof}[Proof of Lemma \ref{interchange}]
Square the difference, apply Cauchy-Schwarz and then integrate $d\bx$ over $\Sigma$.
\end{proof}
%

\begin{proof}[Proof of Lemma \ref{poisson_exp}]
Recall the parameterization of the cylinder, $\Sigma$:
%
\begin{align*}
\bx\in\Sigma:\qquad &\bx = \tau_1\vtilde_1+\tau_2\vtilde_2,\ \ 0\le\tau_1\le1,\ \tau_2\in\R \ , \quad
\ktilde_1\cdot\bx = 2\pi\tau_1 , \ \ktilde_2\cdot\bx = 2\pi\tau_2  , \nn \\
&dx_1\ dx_2\ = \left|\vtilde_1\wedge\vtilde_2\right|\ d\tau_1\ d\tau_2\ \equiv\ |\Omega|\ d\tau_1\ d\tau_2 .
\end{align*}
Using that $g(\bx, \delta\xi) = g(\tau_1\vtilde_1+\tau_2\vtilde_2, \delta \xi)$ and $\Gamma(\bx,\delta\ktilde_2\cdot\bx) = \Gamma(\tau_1\vtilde_1+\tau_2\vtilde_2, 2\pi\delta \tau_2)$ are both appropriately $1$-periodic, we expand $\mathcal{G}(\delta;\xi)$ defined in \eqref{G_inner_def}. By  Lemma \ref{interchange}:
{\footnotesize{
\begin{align}
\mathcal{G}(\delta;\xi)\ &=  
\delta\ |\Omega|\ \int_0^1d\tau_1\ \int_{-\infty}^{\infty} \
 e^{-2\pi i\delta\xi\tau_2}g(\tau_1\vtilde_1+\tau_2\vtilde_2,\delta\xi)\ \Gamma\left(\tau_1\vtilde_1+\tau_2\vtilde_2, 2\pi\delta\tau_2\right)\ d\tau_2\nn\\
&= \delta\ |\Omega|\ \int_0^1d\tau_1\ \lim_{N\to\infty} \sum_{n=-N}^{N}\int_n^{n+1}\ e^{-2\pi i\delta\xi \tau_2}g(\tau_1\vtilde_1+\tau_2\vtilde_2,\delta\xi)\ \Gamma\left(\tau_1\vtilde_1+\tau_2\vtilde_2, 2\pi\delta \tau_2\right)\ d\tau_2\nn\\
&= \delta\ |\Omega|\  \int_0^1d\tau_1\ \lim_{N\to\infty} \sum_{n=-N}^{N} \int_0^1\ e^{-2\pi i\delta\xi (\tau_2+n)}g(\tau_1\vtilde_1+\tau_2\vtilde_2,\delta\xi)\ \Gamma\left(\tau_1\vtilde_1+\tau_2\vtilde_2, 2\pi\delta(\tau_2+n)\right)\ d\tau_2\nn\\
&= \delta\ |\Omega|\ \int_0^1d\tau_1\ \int_0^1\  g(\tau_1\vtilde_1+\tau_2\vtilde_2,\delta\xi)\  
\Big[\ \sum_{n\in\Z}\ e^{-2\pi i\xi\delta(\tau_2+n)}
\ \ \Gamma(\tau_1\vtilde_1+\tau_2\vtilde_2, 2\pi\delta(\tau_2+n))\ \Big]\ d\tau_2 \ . \nn 
\end{align}
}}
\nit  By Theorem \ref{psum-L2} with 
$\Gamma=\Gamma(\tau_1\vtilde_1+\tau\vtilde_2, 2\pi\delta \tau_2)$ and $y=2\pi\delta \xi$, we have
\begin{align*}
 \sum_{n\in\Z}\ e^{-2\pi i\delta\xi(\tau_2+n)}\  \Gamma(\tau_1\vtilde_1+\tau\vtilde_2
 , 2\pi\delta \tau_2) 
&=\ \frac{1}{\delta} \sum_{n\in\Z} e^{2\pi i n \tau_2} \widehat{\Gamma}\left(\tau_1\vtilde_1+\tau \vtilde_2, \frac{n}{\delta}+\xi\right) ,
\end{align*}
with equality holding in $L^2\left( [0,1]^2\times [-\xi_{\rm max}, \xi_{\rm max}]; d\tau_1 d\tau_2\cdot d\xi\right)$.  Again using Lemma \ref{interchange} we may interchange the sum and integral to  obtain\begin{align}
\mathcal{G}(\delta;\xi) & =   \frac{\delta}{\delta}
|\Omega| \int_0^1d\tau_1 \int_0^1  g(\tau_1\vtilde_1+\tau_2\vtilde_2,\delta\xi) 
\sum_{n\in\Z} e^{2\pi i n \tau_2} \widehat{\Gamma}\left(\tau_1\vtilde_1+\tau_2\vtilde_2, \frac{n}{\delta}+\xi\right) d\tau_2 \nn\\
 &= \sum_{n\in\Z} \int_\Omega e^{i n \ktilde_2\cdot\bx} \widehat{\Gamma}\left(\bx, \frac{n}{\delta}+\xi\right) g(\bx,\delta\xi) d\bx \ .
\end{align}
 This completes the proof of Lemma \ref{poisson_exp}.
 \end{proof}

We next apply Lemma \ref{poisson_exp} to each of the inner products \eqref{inner1}-\eqref{inner3}.

\nit \underline{\it Expansion of inner product \eqref{inner1}:}

  \nit  Let 
$g(\bx,\delta\xi)=\overline{p_+(\bx,\delta\xi)}P_+(\bx)W(\bx)$
 and $\Gamma(\bx,\zeta)=\kappa(\zeta) \eta_{+,\rm near}(\zeta)$.\ By Lemma \ref{poisson_exp}, 
\begin{align*}
&\delta\inner{e^{i\xi\delta\ktilde_2 \cdot}p_+(\cdot,\delta\xi), P_+(\cdot)
W(\cdot)\kappa(\delta\ktilde_2\cdot) {\eta}_{+,\rm near}(\delta\ktilde_2\cdot)}_{L^2(\Sigma)} \\
 &=
\sum_{n\in\mathbb{Z}}\int_\Omega e^{in\ktilde_2\cdot\bx}\mathcal{F}_{\zeta}[\kappa\eta_{+,\rm near}]\left(\frac{n}{\delta}+\xi \right) 
\overline{p_{+}(\bx,\delta\xi)}P_+(\bx) W(\bx)d\bx.
\end{align*}

Since $p_\pm(\bx,\lambda=\delta\xi)=P_\pm(\bx)+\varphi_\pm(\bx,\delta\xi)$, where $\varphi_\pm(\bx,\delta\xi)$ satisfies the 
bound  \eqref{deltapbdd}, we have
\begin{align*}
&\delta\inner{e^{i\xi\delta\ktilde_2 \cdot}p_+(\cdot,\delta\xi), P_+(\cdot)
W(\cdot)\kappa(\delta\ktilde_2\cdot) {\eta}_{+,\rm near}(\delta\ktilde_2\cdot)}_{L^2(\Sigma)} \\
&\equiv I_+^1(\xi;\eta_{+,\rm near}) + I_+^2(\xi;\eta_{+,\rm near}),
\end{align*}
where
\begin{align}
I_+^1(\xi;\eta_{+,\rm near}) &=
\sum_{n\in\mathbb{Z}} \mathcal{F}_{\zeta}[\kappa\eta_{+,\rm near}]\left(\frac{n}{\delta}+\xi\right)
\int_\Omega e^{in\ktilde_2\cdot\bx} \abs{P_{+}(\bx)}^2 W(\bx)d\bx \label{I_1}  , \\
I_+^2(\xi;\eta_{+,\rm near}) &= \sum_{n\in\mathbb{Z}}
\mathcal{F}_{\zeta}[\kappa\eta_{+,\rm near}]\left(\frac{n}{\delta}+\xi\right)  \int_\Omega
e^{in\ktilde_2\cdot\bx} \overline{\varphi_+(\bx,\delta\xi)}P_{+}(\bx) W(\bx)d\bx  . \nn 
\end{align}
From Proposition \ref{inner-prods-W}  and Assumption (W3) we have 
\begin{equation}
 \label{zeroth-terms}
 \int_\Omega \abs{P_{+}(\bx)}^2 W(\bx)d\bx = 0 \quad \text{and} \quad
 \int_\Omega \overline{P_{+}(\bx)}P_{-}(\bx) W(\bx)d\bx = \thetasharp \neq 0  . 
\end{equation}
Therefore,  the $n=0$ term in the summation of $I_+^1(\xi;\eta_{+,\rm near})$ in \eqref{I_1} is zero and we may write:
\begin{equation*}
I_+^1(\xi;\eta_{+,\rm near}) =\
\sum_{\abs{n}\geq1} \mathcal{F}_{\zeta}[\kappa\eta_{+,\rm near}]\left(\frac{n}{\delta}+\xi\right) 
\int_\Omega e^{in\ktilde_2\cdot\bx} \abs{P_{+}(\bx)}^2 W(\bx)d\bx . 
\end{equation*} 

\nit \underline{\it Expansion of the inner product \eqref{inner2}:}

\nit Similarly, with $g(\bx,\delta\xi)=\overline{p_+(\bx,\delta\xi)}
P_-(\bx)W(\bx)$ and  $\Gamma(\bx,\zeta) = \kappa(\zeta) \eta_{-,\rm near}(\zeta)$,  we have
\begin{align*}
&\delta\inner{e^{i\xi\delta\ktilde_2 \cdot}p_+(\cdot,\delta\xi), P_-(\cdot)
W(\cdot)\kappa(\delta\ktilde_2\cdot) {\eta}_{-,\rm near}(\delta\ktilde_2\cdot)}_{L^2(\Sigma)} \\
& \equiv \tilde{I_+^3}(\xi;\eta_{-,\rm near}) + I_+^4(\xi;\eta_{-,\rm near}),
\end{align*}
where (noting, by \eqref{zeroth-terms}, that the $n=0$ contribution is nonzero)
\begin{align*}
\tilde{I_+^3}(\xi;\eta_{-,\rm near}) &= \thetasharp \widehat{\kappa\eta}_{+,\rm near}(\xi) \ +\ {I_+^3}(\xi;\eta_{-,\rm near}),\ {\rm where}\nn\\
{I_+^3}(\xi;\eta_{-,\rm near}) &\equiv 
\sum_{\abs{n}\geq1} \mathcal{F}_{\zeta}[\kappa\eta_{-,\rm near}]\left(\frac{n}{\delta}+\xi\right)
\int_\Omega e^{in\ktilde_2\cdot\bx} \overline{P_{+}(\bx)} P_{-}(\bx) W(\bx)d\bx   ,\ \ {\rm and} \\
I_+^4(\xi;\eta_{-,\rm near}) &= \sum_{n\in\mathbb{Z}}
\mathcal{F}_{\zeta}[\kappa\eta_{-,\rm near}]\left(\frac{n}{\delta}+\xi\right)  \int_\Omega
e^{in\ktilde_2\cdot\bx} \overline{\varphi_+(\bx,\delta\xi)}P_{-}(\bx) W(\bx)d\bx  . 
\end{align*} 
%

\nit \underline{\it Expansion of  inner products \eqref{inner3}:}  

\nit Consider the $b=+$ term in \eqref{inner3}. 
Let $g(\bx,\delta\xi)=\overline{p_+(\bx,\delta\xi)}W(\bx)$ and $\Gamma(\bx,\zeta) = \kappa(\zeta) \rho_+(\bx,\zeta)$.
By Lemma \ref{poisson_exp} and the expansion of  $p_+(\bx,\delta\xi)$ about $P_+(\bx)$ in \eqref{deltapbdd} we have:
\begin{align*}
&\delta\inner{e^{i\xi\delta\ktilde_2 \cdot}p_+(\cdot,\delta\xi), 
W(\cdot)\kappa(\delta\ktilde_2\cdot) \rho_+(\cdot,\delta\ktilde_2\cdot)}_{L^2(\Sigma)} \equiv
I_+^5(\xi;\eta_{+,\rm near}) + I_+^6(\xi;\eta_{+,\rm near}),
\end{align*}
where
\begin{align*}
I_+^5(\xi;\eta_{+,\rm near}) &=
\sum_{n\in\Z} \int_\Omega \mathcal{F}_{\zeta}[\kappa\rho_+]\left(\bx,\frac{n}{\delta}+\xi\right)
 e^{in\ktilde_2\cdot\bx} \overline{P_{+}(\bx)} W(\bx)d\bx   , \\
I_+^6(\xi;\eta_{+,\rm near}) &= \sum_{n\in\mathbb{Z}} \int_\Omega
\mathcal{F}_{\zeta}[\kappa\rho_+]\left(\bx,\frac{n}{\delta}+\xi\right)  
e^{in\ktilde_2\cdot\bx} \overline{\varphi_+(\bx,\delta\xi)} W(\bx)d\bx  . 
\end{align*} 
\nit  For the  $b=-$ term in \eqref{inner3} we have
\begin{align*}
&\delta\inner{e^{i\xi\delta\ktilde_2 \cdot}p_+(\cdot,\delta\xi), 
W(\cdot)\kappa(\delta\ktilde_2\cdot) \rho_-(\cdot,\delta\ktilde_2\cdot)}_{L^2(\Sigma)} \equiv
I_+^7(\xi;\eta_{-,\rm near}) + I_+^8(\xi;\eta_{-,\rm near}),
\end{align*}
where
\begin{align*}
I_+^7(\xi;\eta_{+,\rm near}) &=
\sum_{n\in\Z} \int_\Omega \mathcal{F}_{\zeta}[\kappa\rho_-]\left(\bx,\frac{n}{\delta}+\xi\right)
e^{in\ktilde_2\cdot\bx} \overline{P_{+}(\bx)} W(\bx)d\bx   , \\
I_+^8(\xi;\eta_{+,\rm near}) &= \sum_{n\in\mathbb{Z}} \int_\Omega
\mathcal{F}_{\zeta}[\kappa\rho_-]\left(\bx,\frac{n}{\delta}+\xi\right)
e^{in\ktilde_2\cdot\bx} \overline{\varphi_+(\bx,\delta\xi)} W(\bx)d\bx  . 
\end{align*} 
Assembling the above expansions, we find that the full inner product, \eqref{full_inner}, may be expressed as:
\begin{equation}
 \inner{\Phi_{+}(\cdot,\delta\xi),\kappa(\delta\ktilde_2\cdot)W(\cdot)
\eta_{\rm near}(\cdot)}_{L^2(\Sigma)} = \thetasharp\widehat{\kappa\eta}_{{\rm near},-}(\xi) +
\sum_{j=1}^8I_{+}^j(\xi;\eta_{\rm near}) ,\label{full-ipplus}
\end{equation} 
A similar calculation yields:
\begin{equation}
 \inner{\Phi_{-}(\cdot,\delta\xi),\kappa(\delta\ktilde_2\cdot)W(\cdot)
\eta_{\rm near}(\cdot)}_{L^2(\Sigma)} = \thetasharp\widehat{\kappa\eta}_{{\rm near},+}(\xi) +
\sum_{j=1}^8I_{-}^j(\xi;\eta_{\rm near}),\ \label{full-ipminus}
\end{equation} 
where the terms $I_{-}^j(\xi;\eta_{\rm near})$ are defined analogously to $I_{+}^j(\xi;\eta_{\rm near})$.  We now substitute our results \eqref{full-ipplus}-\eqref{full-ipminus}, \eqref{Fb-def}-\eqref{Fdef} and \eqref{eta_far_affine}  into \eqref{near5}-\eqref{near6} to obtain the following:

%
%
\begin{proposition}
\label{near_freq_compact}
Let $\widehat{\beta}(\xi) = (\widehat{\eta}_{+,\rm near}(\xi) , \widehat{\eta}_{-,\rm near}(\xi))^T$. 
Equations  \eqref{near5}-\eqref{near6}, the closed system governing the near energy components, $\eta_{\rm near}$,  of the corrector, $\eta$,  is of the form:
\begin{equation}
 \label{compacterroreqn}
 \left(\widehat{\mathcal{D}}^{\delta}+\widehat{\mathcal{L}}^{\delta}(\mu) -\delta \mu\right)\widehat{\beta}(\xi) =
\mu\widehat{\mathcal{M}}(\xi;\delta) + \widehat{\mathcal{N}}(\xi;\delta) .
\end{equation} 
Here, $\mathcal{D}^\delta$ denotes the band-limited Dirac operator defined by 
\begin{equation}
 \widehat{\mathcal{D}}^{\delta}\widehat{\beta}(\xi)\ \equiv\ |\lamsharp|\ |\ktilde_2|\ \sigma_3\ \xi\ \widehat{\beta}(\xi)\
+\ \thetasharp\ \chi\left(\abs{\xi}\leq\delta^{\exponent-1}\right)\ \sigma_1\ \widehat{\kappa\beta}(\xi).
\label{bl-dirac-op}
\end{equation}
The linear operator, $\widehat{\mathcal{L}}^{\delta}(\mu)$,   acting on $\widehat{\beta}$, and the source terms  $\widehat{\mathcal{M}}(\xi;\delta)$ and $\widehat{\mathcal{N}}(\xi;\delta)$ are  defined   by:
{\small
\begin{equation}
\label{L_op}
\widehat{\mathcal{L}}^{\delta}(\mu)\widehat{\beta}(\xi) \equiv 
\chi\left(\abs{\xi}\leq\delta^{\exponent-1}\right) \sum_{j=1}^3 \widehat{\mathcal{L}}^{\delta}_j(\mu) \widehat{\beta}(\xi) , \quad \text{where}
\end{equation}
\begin{align*}
\widehat{\mathcal{L}}^{\delta}_1(\mu) \widehat{\beta}(\xi) & \equiv \delta \xi^2
\begin{pmatrix}
 E_{2,+}(\delta\xi)\ \widehat{\eta}_{+,\rm near}(\xi) \\E_{2,-}(\delta\xi)\ \widehat{\eta}_{-,\rm near}(\xi)
\end{pmatrix} , \quad
\widehat{\mathcal{L}}^{\delta}_2(\mu) \widehat{\beta}(\xi) \equiv
\sum_{j=1}^8
\begin{pmatrix}
 I_{+}^j(\xi;\widehat{\eta}_{\pm,\rm near}(\xi))\\ I_{-}^j(\xi;\widehat{\eta}_{\pm,\rm near}(\xi))
\end{pmatrix} , \\		
\widehat{\mathcal{L}}^{\delta}_3(\mu) \widehat{\beta}(\xi) & \equiv 
\begin{pmatrix}
\inner{\Phi_{+}(\cdot,\delta\xi),\kappa(\delta\ktilde_2\cdot)W(\cdot)
[A\eta_{\rm near}](\cdot;\mu,\delta)}_{L_\kparv^2} \\
\inner{\Phi_{-}(\cdot,\delta\xi),
\kappa(\delta\ktilde_2\cdot)W(\cdot)[A\eta_{\rm near}](\cdot;\mu,\delta)}_{L_\kparv^2}
\end{pmatrix},	
\end{align*}
\begin{equation}
\label{M_op}
\widehat{\mathcal{M}}(\xi;\delta) \equiv 
\chi\left(\abs{\xi}\leq\delta^{\exponent-1}\right) \sum_{j=1}^3\widehat{\mathcal{M}}_j(\xi;\delta) , \quad \text{where (inner products over ${L_\kparv^2}$)}
\end{equation}
\begin{align*}
\widehat{\mathcal{M}}_1(\xi;\delta) & \equiv
\begin{pmatrix}
 \inner{\Phi_{+}(\cdot,\delta\xi),\psi^{(0)}(\cdot,\delta\cdot)} \\
 \inner{\Phi_{-}(\cdot,\delta\xi),\psi^{(0)}(\cdot,\delta\cdot)}
\end{pmatrix} , \quad 
\widehat{\mathcal{M}}_2(\xi;\delta) \equiv
\delta
\begin{pmatrix}
 \inner{\Phi_{+}(\cdot,\delta\xi),\psi^{(1)}_p(\cdot,\delta\cdot)}\\
 \inner{\Phi_{-}(\cdot,\delta\xi),\psi^{(1)}_p(\cdot,\delta\cdot)}
\end{pmatrix} , \nn \\	
\widehat{\mathcal{M}}_3(\xi;\delta) & \equiv -
\begin{pmatrix}
 \inner{\Phi_{+}(\cdot,\delta\xi),\kappa(\delta\ktilde_2\cdot)W(\cdot)B(\cdot;\delta)}\\
 \inner{\Phi_{-}(\cdot,\delta\xi),\kappa(\delta\ktilde_2\cdot)W(\cdot)B(\cdot;\delta)}
\end{pmatrix},	\nn 
\end{align*}
%
%
\begin{equation}
\label{N_op}
\widehat{\mathcal{N}}(\xi;\delta) \equiv 
\chi\left(\abs{\xi}\leq\delta^{\exponent-1}\right) \sum_{j=1}^4\widehat{\mathcal{N}}_j(\xi;\delta) , \quad \text{where (inner products over ${L_\kparv^2}$)}
\end{equation}
\begin{align*}
\widehat{\mathcal{N}}_1(\xi;\delta) & \equiv
\begin{pmatrix}
\inner{\Phi_{+}(\bx,\delta\xi), \left(2\ktilde_2\cdot\nabla_\bx\  \D_\zeta-\kappa(\delta\ktilde_2\cdot\bx)W(\bx)\right)\psi^{(1)}_p(\bx,\delta\ktilde_2\cdot\bx)} \\
\inner{\Phi_{-}(\bx,\delta\xi), \left(2\ktilde_2\cdot\nabla_\bx\  \D_\zeta-\kappa(\delta\ktilde_2\cdot\bx)W(\bx)\right)\psi^{(1)}_p(\bx,\delta\ktilde_2\cdot\bx)}
\end{pmatrix} , \\ 
\widehat{\mathcal{N}}_2(\xi;\delta) & \equiv
\begin{pmatrix}
\inner{\Phi_{+}(\bx,\delta\xi),|\ktilde_2|^2\ \D_\zeta^2\psi^{(0)}(\bx,\zeta)\Big|_{\zeta=\delta \ktilde_2\cdot\bx}}\\
\inner{\Phi_{-}(\bx,\delta\xi),|\ktilde_2|^2\ \D_\zeta^2\psi^{(0)}(\bx,\zeta)\Big|_{\zeta=\delta \ktilde_2\cdot \bx}}
\end{pmatrix}, \\ 
\widehat{\mathcal{N}}_3(\xi;\delta) & \equiv
\begin{pmatrix}
\inner{\Phi_{+}(\bx,\delta\xi),|\ktilde_2|^2\ \D_\zeta^2\psi_p^{(1)}(\bx,\zeta)\Big|_{\zeta=\delta \ktilde_2\cdot\bx}} \\
\inner{\Phi_{-}(\bx,\delta\xi),|\ktilde_2|^2\ \D_\zeta^2\psi_p^{(1)}(\bx,\zeta)\Big|_{\zeta=\delta \ktilde_2\cdot \bx}}
\end{pmatrix} , \nn \\ 
\widehat{\mathcal{N}}_4(\xi;\delta) & \equiv -
\begin{pmatrix}
\inner{\Phi_{+}(\bx,\delta\xi),\kappa(\delta\ktilde_2\cdot\bx)W(\bx)C(\bx;\delta)} \\
\inner{\Phi_{-}(\bx,\delta\xi),\kappa(\delta\ktilde_2\cdot\bx)W(\bx)C(\bx;\delta)}
\end{pmatrix}. \nn 
\end{align*}
}
\end{proposition}

We conclude this section with the assertion that from an appropriate solution
$\left(\widehat{\beta}^\delta(\xi),\mu(\delta)\right)$
of the band-limited Dirac system \eqref{compacterroreqn} one can construct a bound state $\left(\Psi^\delta,E^\delta\right)$ of the Schr\"odinger eigenvalue
problem \eqref{EVP_2}. \
We say $f\in L^{2,1}(\R)$ if  $\|f\|_{L^{2,1}(\R)}^2\ \equiv\ \int (1+|\xi|^2)^{1/2}f(\xi)d\xi<\infty$.

\begin{proposition}\label{needtoshow}
Suppose, for $0<\delta<\delta_0$, the band-limited Dirac system \eqref{compacterroreqn} has a solution  
$\left(\widehat{\beta}^\delta(\xi),\mu(\delta)\right)$, $\widehat{\beta}^\delta=(\widehat{\beta}^\delta_+,\widehat{\beta}_-^\delta)^T$, where $\supp\widehat{\beta}^\delta\subset \{|\xi|\le\delta^{\exponent-1}\}$, satisfying:
\begin{align*}
 &\norm{\widehat{\beta}^\delta(\cdot;\mu,\delta)}_{L^{2,1}(\mathbb{R})} \lesssim \delta^{-1},\ 0<\delta<\delta_0
 \quad (\textrm{to be verified in Proposition~\ref{solve4beta}}), \\
 &\mu(\delta) \text{ bounded and } \mu(\delta) - \mu_0 \to 0 \text{ as } \delta\to0 \quad
  (\textrm{to be verified in Proposition~\ref{proposition3}}).  
 \end{align*} 
Define 
\begin{equation}
\widehat{\eta}^{\delta}_{\rm near,+}(\xi) = \widehat{\beta}^\delta_+(\xi),\ \ \ \widehat{\eta}^{\delta}_{\rm near,-}(\xi)=\widehat{\beta}^\delta_-(\xi) ,
\label{hat-eta}\end{equation}
and construct  $\eta^\delta\equiv \eta^\delta_{\rm near} + \eta^\delta_{\rm far}$ as follows:
\begin{align}
\eta^\delta_{\rm near}(\bx)&= \sum_{b=\pm}\int_{|\lambda|\le \delta^\exponent} \widehat{\eta}^\delta_{{\rm
near},b}\left(\frac{\lambda}{\delta}\right) \Phi_b(\bx;\lambda) d\lambda , \label{eta_def_beta1} \\ 
\widetilde{\eta}^\delta_{{\rm far},b}(\lambda)&= 
\widetilde{\eta}_{\rm far,b}[\eta_{\rm near},\mu,\delta](\lambda),\ \ b\ge1 ;
\quad \textrm{(see Proposition \ref{fixed-pt})} , \nn \\	
\eta^\delta_{\rm far}(\bx)&= \sum_{b=\pm} \int_{\delta^\exponent \le |\lambda |\le 1/2 } 
\widetilde{\eta}^\delta_{{\rm far},b}\left(\lambda\right) \Phi_b(\bx;\lambda) d\lambda
+ \sum_{b\ne\pm} \int_{|\lambda| \leq 1/2}
\widetilde{\eta}^\delta_{{\rm far},b}\left(\lambda\right) \Phi_b(\bx;\lambda) d\lambda  . \nn \\	
\eta^\delta(\bx)&\equiv \eta^\delta_{\rm near}(\bx) + \eta^\delta_{\rm far}(\bx),\ \ 
E^\delta\equiv E_\star+\delta^2\mu(\delta),\ \
0<\delta<\delta_0. \nn
\end{align}
Then, for all $0<|\delta|<\delta_0$, the following holds:
\begin{enumerate}
\item[(a)] $ \eta^\delta(\bx)\in H_\kparv^{2}(\Sigma)$.
\item[(b)] $\left(\eta^\delta,\mu(\delta)\right)$ solves the corrector equation \eqref{corrector-eqn1}.
\item[(c)] Theorem \ref{thm-edgestate} holds. The pair $(\Psi^\delta,E^\delta)$, defined by (see also
\eqref{eta-def}-\eqref{mu-def})
\begin{equation}
 \label{main_result_ansatz1}
 \begin{split}
 \Psi^\delta(\bx) &= \psi^{(0)}(\bx,\bX)+\delta\psi^{(1)}_p(\bx,\bX)+\delta\eta^\delta(\bx),\ \ \bX=\delta\ktilde_2\cdot \bx,\\
 E^\delta &= E_\star+\delta^2 \mu_0 +o(\delta^2) ,
 \end{split}
\end{equation}
is a solution of the eigenvalue problem \eqref{EVP_2} with corrector estimates asserted in the statement of Theorem \ref{thm-edgestate}.
\end{enumerate}
\end{proposition}

To prove Proposition \ref{needtoshow} we use the following lemma.

\begin{lemma}
 \label{beta_vs_eta}
 There exists a $\delta_0>0$ such that, for all $0<\delta<\delta_0$, the following holds:
 Assume $\beta \in L^2(\R)$ and let $\eta^\delta_{\rm near}(\bx)$ be defined by \eqref{hat-eta}-\eqref{eta_def_beta1}.
 Then, 
 \begin{equation*}
 \norm{\eta_{\rm near}}_{H_\kparv^2} \lesssim \delta^{1/2} \norm{\beta}_{L^2(\R)}.
 \end{equation*}
\end{lemma}

\nit The proof of Lemma \ref{beta_vs_eta} parallels that of Lemma 6.9 in \cites{FLW-MAMS:15}, and is not reproduced here.

\begin{proof}[Proof of Proposition \ref{needtoshow}]
From $\widehat{\beta}$ we construct
$\eta^\delta_{\rm near}$, such that:
$ \norm{\eta_{\rm near}}_{H_\kparv^2} \lesssim \delta^{1/2} \norm{\beta}_{L^2(\R)}$ (Lemma
\ref{beta_vs_eta}).
Next, part 2 of Proposition \ref{fixed-pt}, \eqref{eta-far-bound}, gives a bound on 
 $\eta_{\rm far}$:
 $\norm{\eta_{\rm far}[\eta_{\rm near};\mu,\delta]}_{H_\kparv^2} \le\ C''\left(\ 
\frac{\delta}{\omega(\delta^\exponent)}\norm{\eta_{\rm near}}_{L_\kparv^2}+\frac{\delta^\frac12}{\omega(\delta^\exponent)} \right)$.
These two bounds give the desired $H_\kparv^2(\Sigma)$ bound on $\eta^{\delta}$.
Note that all steps in our derivation of the band-limited Dirac
system \eqref{compacterroreqn} are reversible, in particular our application of the Poisson summation formula in
$L^2_{\rm loc}$. Therefore, $(\Psi^\delta,E^\delta)$, given by \eqref{main_result_ansatz1} is an $H_\kparv^2$ eigenpair of  \eqref{EVP_2} .
\end{proof}


\nit We focus then on constructing and estimating the solution of the band-limited Dirac system \eqref{compacterroreqn}.

\subsection{Analysis of the band-limited Dirac system}\label{analysis-blDirac}

The formal $\delta\downarrow0$ limit of the band-limited operator $\widehat{\mathcal{D}}^\delta$, displayed in \eqref{bl-dirac-op}, is a 1D Dirac operator ${\mathcal{D}}$ defined via:
\begin{equation}
 \label{diraclimit}
 \widehat{\mathcal{D}}\ \widehat{\beta}(\xi)\ \equiv\  |\lamsharp|\ |\ktilde_2|\ \sigma_3\ \xi \widehat{\beta}(\xi)\ +\ \thetasharp\ \sigma_1\ \widehat{\kappa\beta}(\xi).
\end{equation} 

Our goal is to solve the  system \eqref{compacterroreqn}.
We therefore rewrite the linear operator in equation \eqref{compacterroreqn} as a perturbation of $\widehat{\mathcal{D}}$ \eqref{diraclimit}, and seek $\widehat{\beta}$ as a solution to:
\begin{equation}
 \label{erroreqnfactored}
 \widehat{\mathcal{D}}\widehat{\beta}(\xi) + \left(\widehat{\mathcal{D}}^{\delta}-\widehat{\mathcal{D}} +
\widehat{\mathcal{L}}^{\delta}(\mu) -\delta \mu\right)\widehat{\beta}(\xi) =
\mu\widehat{\mathcal{M}}(\xi;\delta)
+ \widehat{\mathcal{N}}(\xi;\delta).
\end{equation} 
We next solve \eqref{erroreqnfactored} using a Lyapunov-Schmidt reduction strategy.
 By Proposition \ref{zero-mode}, the null space of $\widehat{\mathcal{D}}$ is spanned by $\widehat{\alpha}_{\star}(\xi)$, the Fourier transform of the zero energy eigenstate \eqref{Fstar}. Since  $\alpha_{\star}(\zeta)$ is Schwartz class, so too is $\widehat{\alpha}_{\star}(\xi)$ and $\widehat{\alpha}_{\star}(\xi)\in H^s(\mathbb{R})$ for any $s\ge1$. 

For any $f\in L^2{(\mathbb{R})}$, introduce the orthogonal projection operators,
\begin{equation*}
 \widehat{P}_{\parallel}f =
\inner{\widehat{\alpha}_{\star},f}_{L^2(\R)}\widehat{\alpha}_{\star},~~~\text{and}~~~\widehat{P}_{\perp}f =
(I-\widehat{P}_{\parallel})f.
\end{equation*} 
Since $\widehat{P}_{\parallel}\widehat{\mathcal{D}}\widehat{\beta}(\xi)=0$ and
$\widehat{P}_{\perp}\widehat{\mathcal{D}}\widehat{\beta}(\xi)=\widehat{\mathcal{D}}\widehat{\beta}(\xi)$, 
equation \eqref{erroreqnfactored} is equivalent to  the system
\begin{align}
&\widehat{P}_{\parallel}\left\{\left(\widehat{\mathcal{D}}^{\delta}-\widehat{\mathcal{D}} +
\widehat{\mathcal{L}}^{\delta}(\mu) -\delta
\mu\right)\widehat{\beta}(\xi) - \mu\widehat{\mathcal{M}}(\xi;\delta) -
\widehat{\mathcal{N}}(\xi;\delta)\right\} = 0, \label{pplleqn}\\
&\widehat{\mathcal{D}}\widehat{\beta}(\xi) +
\widehat{P}_{\perp}\left\{\left(\widehat{\mathcal{D}}^{\delta}-\widehat{\mathcal{D}}
+ \widehat{\mathcal{L}}^{\delta}(\mu) -\delta \mu\right)\widehat{\beta}(\xi)\right\} =
\widehat{P}_{\perp}\left\{\mu\widehat{\mathcal{M}}(\xi;\delta) +
\widehat{\mathcal{N}}(\xi;\delta)\right\}.
\label{pperpeqn}
\end{align} Our strategy will be to first solve \eqref{pperpeqn}  for $\widehat{\beta}=
\widehat{\beta}[\mu,\delta]$, for $\delta>0$ and sufficiently small. 
   We then  substitute
 $\widehat{\beta}[\mu,\delta]$ into \eqref{pplleqn} to obtain a closed scalar equation. This equation is then solved
   for $\mu=\mu(\delta)$ for $\delta$ small. 
 The first step in this strategy is accomplished in
\begin{proposition}\label{solve4beta}
Fix $M>0$. There exists $\delta_0>0$ and a  mapping 
$(\mu,\delta)\in R_{M,\delta_0}\equiv \{|\mu|<M\}\times (0,\delta_0)\mapsto \widehat{\beta}(\cdot;\mu,\delta)\in
L^{2,1}(\R)$
which is Lipschitz in $\mu$, such that $\widehat{\beta}(\cdot;\mu,\delta)$ solves
\eqref{pperpeqn} for $(\mu,\delta)\in R_{M,\delta_0}$. Furthermore, we have the bound
\begin{equation*}
 \norm{\widehat{\beta}(\cdot;\mu,\delta)}_{L^{2,1}(\mathbb{R})} \lesssim \delta^{-1},\ 0<\delta<\delta_0.
 \end{equation*}
\end{proposition}

The details of the proof of Proposition \ref{solve4beta} are  similar to those in proof of 
Proposition 6.10 in \cites{FLW-MAMS:15}; equation \eqref{pperpeqn} is expressed as $(I+C^\delta(\mu))\widehat{\beta}(\xi;\mu,\delta) = \widehat{\mathcal{D}}^{-1} \widehat{P}_{\perp}\left\{\mu\widehat{\mathcal{M}}(\xi;\delta) +
\widehat{\mathcal{N}}(\xi;\delta)\right\}$ and the operator $C^\delta(\mu)$ is proved to be bounded on $L^{2,1}(\R)$ and of norm less than one for all  $0<\delta<\delta_0$, with $\delta_0$ sufficiently small. In bounding $C^\delta(\mu)$ on $L^{2,1}(\R)$, we require $H^1(\R)$ bounds for wave operators associated with the Dirac operator, $\mathcal{D}$. These can derived from corresponding results for scalar Schr\"odinger operators, under the assumptions implied by $\kappa(\zeta)$ being a domain wall function in the sense of Definition \ref{domain-wall-defn}.

\subsection{Final reduction to an equation for $\mu=\mu(\delta)$ and its solution}\label{final-reduction}

Substituting the solution $\widehat{\beta}(\xi;\mu,\delta)$ (Proposition \ref{solve4beta}) into   \eqref{pplleqn},  yields the equation $\mathcal{J}_+[\mu,\delta]=0$, relating $\mu$ and $\delta$.
%
Here, $\mathcal{J}_{+}[\mu;\delta] $ is given by:
\begin{align*}
\mathcal{J}_{+}[\mu;\delta] &\equiv
\mu\ \delta\inner{\widehat{\alpha}_{\star}(\cdot),\widehat{\mathcal{M}}(\cdot;\delta)}_{L^2(\mathbb{R})}
+ \delta\inner{\widehat{\alpha}_{\star}(\cdot),\widehat{\mathcal{N}}(\cdot;\delta)}_{L^2(\mathbb{R})} \\
&~~~ -\delta\inner{\widehat{\alpha}_{\star}(\cdot),\left(\widehat{\mathcal{D}}^{\delta}-\widehat{\mathcal{D}}\right)
\widehat{\beta}(\cdot;\mu,\delta)}_{L^2(\mathbb{R})}
-\delta\inner{\widehat{\alpha}_{\star}(\cdot),\widehat{\mathcal{L}}^{\delta}(\mu)
\widehat{\beta}(\cdot;\mu,\delta)}_{L^2(\mathbb{R})}\nn\\
&~~~+\delta^2 \mu\inner{\widehat{\alpha}_{\star}(\cdot),\widehat{\beta}(\cdot;\mu,\delta)}_{L^2(\mathbb{R})}. \nn
\end{align*}
The mapping $(\mu,\delta)\in\{|\mu|<M,\ \delta\in(0,\delta_0)\}\mapsto\mathcal{J}_{+}(\mu,\delta)$ is well defined and Lipschitz continuous with respect to $\mu$. In the following proposition, we note that $\mathcal{J}_{+}[\mu,\delta]$ can
be extended to a continuous function on the half-open interval $[0,\delta_0)$.

\begin{proposition}
 \label{proposition3}
Let $\delta_0>0$ be as above. Define
 \begin{equation*}
  \mathcal{J}[\mu,\delta] \equiv \left\{
\begin{array}{cl}
\mathcal{J}_{+}[\mu,\delta]  & ~~~\text{for}~~ 0<\delta<\delta_0,\\
\mu-\mu_0 & ~~~\text{for}~~ \delta=0\ ,
\end{array} \right.\
 \end{equation*} 
where
$ \mu_0 \equiv -\inner{\alpha_{\star},\mathcal{G}^{(2)}}_{L^2(\mathbb{R})} = E^{(2)}$, 
 and $\mathcal{G}^{(2)}$ is given in \eqref{ipG}; see also \eqref{G2def} and \eqref{solvability_cond_E2}.
 Fix $M = \max\{2\abs{\mu_0},1\}$. 
  Then,   $(\mu,\delta)\in\{|\mu|<M,\ 0\le\delta<\delta_0\}\mapsto\mathcal{J}(\mu,\delta)$
  is well-defined and continuous. 
 \end{proposition} 

{\it Proof:} The proof parallels that of Proposition 6.16 of \cites{FLW-MAMS:15}. 
 The key is to establish the following asymptotic relations, for all $0<\delta<\delta_0$ with $\delta_0$
sufficiently small: 
\begin{align}
 \lim_{\delta\rightarrow0}\delta\inner{\widehat{\alpha}_{\star}(\cdot),
\widehat{\mathcal{M}}(\cdot;\delta)}_{L^2(\mathbb{R})} &=1; \label{sketch_limit1} \\
 \lim_{\delta\rightarrow0}\delta\inner{\widehat{\alpha}_{\star}(\cdot),
\widehat{\mathcal{N}}(\cdot;\delta)}_{L^2(\mathbb{R})} &= -\mu_0; \label{sketch_limit2} 
\end{align}
and the following bounds hold for some constant $C_M$:
\begin{align}
 \abs{\delta\inner{\widehat{\alpha}_{\star}(\cdot), \left(\widehat{\mathcal{D}}^{\delta}-\widehat{\mathcal{D}}\right)
\widehat{\beta}(\cdot;\mu,\delta)}_{L^2(\mathbb{R})}} &\le C_M \delta^{1-\exponent}; \label{sketch_bound1} \\
 \abs{\delta\inner{\widehat{\alpha}_{\star}(\cdot), \widehat{\mathcal{L}}^{\delta}(\mu)
\widehat{\beta}(\cdot;\mu,\delta)}_{L^2(\mathbb{R})}} &\le C_M \delta^{\exponent}; \label{sketch_bound2} \\
 \abs{\delta^2 \mu\inner{\widehat{\alpha}_{\star}(\cdot), \widehat{\beta}(\cdot;\mu,\delta)}_{L^2(\mathbb{R})}} 
&\le C_M \delta. \label{sketch_bound3}
\end{align}
The detailed verification of  \eqref{sketch_limit1}-\eqref{sketch_bound3} follows 
the approach taken in Appendix H of \cites{FLW-MAMS:15}.  We make a few remarks on the calculations. Each of the expressions in \eqref{sketch_limit1}-\eqref{sketch_limit2}
 consists of inner products of the form:
 {\small
 \begin{equation}
\mathfrak{J}(\delta) \equiv \delta \left\langle \widehat{\alpha}_{\star}(\xi) , \chi\left(\abs{\xi}\leq\delta^{\exponent-1}\right)
\begin{pmatrix}
\inner{\Phi_{+}(\bx,\delta\xi), J(\bx,\delta\ktilde_2\cdot\bx)}_{L_\kparv^2} \\
\inner{\Phi_{-}(\bx,\delta\xi ), J(\bx,\delta\ktilde_2\cdot\bx)}_{L_\kparv^2}
\end{pmatrix} \right\rangle_{L^2(\R_\xi)} .\label{sample-ip}
\end{equation}}
Here, $J(\bx,\zeta)=e^{i\bK\cdot\bx}\mathcal{K}(\bx,\zeta)$, where 
$\bx\mapsto\mathcal{K}(\bx,\zeta)$ is $\Lambda_h-$ periodic and $\zeta\mapsto \mathcal{K}(\bx,\zeta)$ is smooth and rapidly decaying on $\R$.  Consider, for example, the expression within:
$\inner{\Phi_{+}(\bx,\delta\xi), J(\bx,\delta\ktilde_2\cdot\bx)}_{L_\kparv^2}$. This may be rewritten and expanded, using Lemma \ref{poisson_exp}:
\begin{align}
&\delta\ \int_\Sigma e^{-i\delta\xi\ktilde_2\cdot\bx}\ \overline{p_+(\bx,\delta\xi)}\  \mathcal{K}(\bx,\delta\ktilde_2\cdot\bx)\ d\bx\ =\   \sum_{n\in\mathbb{Z}} \int_\Omega e^{in \ktilde_2\cdot\bx}\ \overline{p_+(\bx,\delta\xi)}\ \widehat{\mathcal{K}}\left(\bx,\frac{n}{\delta}+\xi\right)\ d\bx
\nn\\
&\ = \int_\Omega  \overline{p_+(\bx,\delta\xi)}\ \widehat{\mathcal{K}}\left(\bx,\xi\right)\ d\bx\ +\  \sum_{|n|\ge1} \int_\Omega e^{in \ktilde_2\cdot\bx}\ \overline{p_+(\bx,\delta\xi )}\ \widehat{\mathcal{K}}\left(\bx,\frac{n}{\delta}+\xi\right)\ d\bx . \nn
\end{align}
Since in \eqref{sample-ip} $\xi$ is localized to the set where $|\xi|\le\delta^{\exponent-1},\ \exponent>0$, for $|n|\ge1$, we have  $n/\delta+\xi\approx n/\delta$ and the decay of $\zeta\mapsto\widehat{\mathcal{K}}(\bx,\zeta)$ can be used to show that, as  $\delta$ tends to zero, the sum over $|n|\ge1$ tends to zero in $L^2(d\xi)$. It can also be shown, using the localization of $\xi$, that the $n=0$ contribution to the sum, tends to 
 $\left\langle P_+(\bx), \widehat{\mathcal{K}}\left(\bx,\xi\right)\ \right\rangle_{L^2(\Omega)}$.
 Therefore, uniformly in $|\xi|\le\delta^{\exponent-1}$, we have
 \begin{equation}
 \lim_{\delta\to0}\ \delta\ \inner{\Phi_{+}(\bx,\delta\xi), J(\bx,\delta\ktilde_2\cdot\bx)}_{L_\kparv^2}\ =\ \left\langle P_+(\bx), \widehat{\mathcal{K}}\left(\bx,\xi\right)\ \right\rangle_{L^2(\Omega)} .
 \nn\end{equation}
 Therefore,
 \begin{align}
\lim_{\delta\to0}\mathfrak{J}(\delta)\ & =
 \int_\R\ d\xi\ \left[\
  \overline{\widehat{\alpha}_{\star,+}(\xi)}\  \left\langle P_+(\bx), \widehat{ \mathcal{K} }\left(\bx,\xi\right)\ \right\rangle_{L^2(\Omega)} \right. \nn \\
 &\qquad\qquad\qquad \left. +   \overline{\widehat{\alpha}_{\star,-}(\xi)}\ 
\left\langle P_-(\bx),  \widehat{ \mathcal{K} }\left(\bx,\xi\right)\ \right\rangle_{L^2(\Omega)}
\ \right] .  \label{help-limit}
 \end{align}
 
 The principle contribution to the limit in \eqref{sketch_limit1} comes from the $\widehat{\mathcal{M}}_1(\xi;\delta)$ term in \eqref{M_op}.
We apply \eqref{help-limit} with the choice
 $J=J_{\mathcal{M}}(\bx,\zeta)= \psi^{(0)}(\bx,\zeta)$ and 
 \[ \mathcal{K}_{\mathcal{M}}(\bx,\zeta)\equiv e^{-i\bK\cdot\bx}J_{\mathcal{M}}(\bx,\zeta)=\alpha_{\star,+}(\zeta)P_+(\bx)+\alpha_{\star,-}(\zeta)P_-(\bx)\ .\]
 
  The principle contribution to the limit in \eqref{sketch_limit2} comes from the $\widehat{\mathcal{N}}_1(\xi;\delta)$ and $\widehat{\mathcal{N}}_2(\xi;\delta)$ terms in
  \eqref{N_op}. We
   apply  \eqref{help-limit} with the choice  
  \[J=J_{\mathcal{N}}(\bx,\zeta)= \left(2\ktilde_2\cdot\nabla_\bx\  \D_\zeta-\kappa(\zeta)W(\bx)\right)\psi^{(1)}_p(\bx,\zeta)+|\ktilde_2|^2\ \D_\zeta^2\psi^{(0)}(\bx,\zeta) \]
   and $\mathcal{K}_{\mathcal{N}}(\bx,\zeta)\equiv 
  e^{-i\bK\cdot\bx}J_{\mathcal{N}}(\bx,\zeta)$. The detailed computations are omitted since they
  are similar to those in  \cites{FLW-MAMS:15} .

By \eqref{sketch_limit1}-\eqref{sketch_bound3}, it follows that 
 $\mathcal{J}_+(\mu;\delta)=\mu-\mu_0+o(1)$ as $\delta\rightarrow0$ uniformly for
$|\mu|\leq M $. 
Therefore,  $\mathcal{J}[\mu,\delta]$ is well-defined on $ \{(\mu,\delta)\ :\ |\mu|<M,\ 0\leq\delta<\delta_0\}$, continuous at $\delta=0$ and Proposition \ref{proposition3} is  proved.

Summarizing, we have that given $\widehat{\beta}(\cdot,\mu,\delta)$, constructed in Proposition \ref{solve4beta}, to complete our construction of a solution to  \eqref{pplleqn}-\eqref{pperpeqn}, it suffices to solve \eqref{pplleqn} for
$\mu=\mu(\delta)$. Furthermore, we have just shown that   \eqref{pplleqn} holds if and only if $\mu=\mu(\delta)$ is a solution of $\mathcal{J}[\mu;\delta]=0$.
From Proposition \ref{proposition3} it follows that  $\mathcal{J}[\mu;\delta]=0$ has a transverse zero, $\mu=\mu(\delta)$, for all $\delta$ and  sufficiently small. The details are presented in Proposition 6.17
 of \cites{FLW-MAMS:15}:
\begin{proposition}\label{solveJeq0}
There exists $\delta_0>0$, and a function $\delta\mapsto\mu(\delta)$, defined for $0\le\delta<\delta_0$ such that:
 $|\mu(\delta)|\le M$, $\lim_{\delta\to0}\mu(\delta)=\mu(0)=\mu_0 \equiv E^{(2)}$ and 
 $\mathcal{J}[\mu(\delta),\delta]=0$ for all $0\le\delta<\delta_0$.
\end{proposition}

\nit We have constructed  a solution pair $\left(\widehat{\beta}^\delta(\xi),\mu(\delta)\right)$, with 
$\widehat{\beta}^\delta\in L^{2,1}(\R;d\xi)$,  of the band-limited Dirac system \eqref{compacterroreqn}. Now apply  Proposition \ref{needtoshow} and the proof of  Theorem \ref{thm-edgestate} is complete.

\section{Edge states for weak potentials and the \nofold condition for the zigzag slice} \label{zz-gap}
 
In Section \ref{thm-edge-state} we fixed an arbitrary edge, $\vtilde_1=a_1\bv_1+a_2\bv_2$ and proved the existence of  topologically protected $\vtilde_1-$ edge states under the spectral \nofold condition. In this section, we consider the special case of the zigzag edge, corresponding to the choice: $\vtilde_1=\bv_1$. We prove that the spectral \nofold condition holds in the weak potential regime, provided $\eps V_{1,1}>0$; this implies the existence of a topologically protected family zigzag edge states. 
%
%
%

We proceed in this section to prove the following:
\begin{enumerate}
\item Theorem \ref{SGC!}: The operator $-\Delta+\eps V$ satisfies the \nofold condition along the zigzag ($\bk_2$) slice at the Dirac point $(\bK,E^\eps_\star)$ ; see Definition \ref{SGC}.
\item  Theorem \ref{delta-gap}: $-\Delta+\eps V(\bx)+\delta W(\bx)$ acting in  $L^2(\Sigma_{\kpar=\bK\cdot\bv_1}$) has a spectral gap
about the energy $E=E_\star^\eps$.
\item Theorem \ref{NO-directional-gap!}: If $\eps V_{1,1}<0$, then the spectral \nofold condition for the zigzag slice does not hold.
\item Theorem \ref{Hepsdelta-edgestates}: For $0<|\delta|\ll\eps^2$ and $\eps$ sufficiently small, the zigzag edge state eigenvalue problem for $H^{(\eps,\delta)} =-\Delta+\eps V(\bx)+\delta\kappa(\delta\bk_2\cdot\bx) W(\bx)$ has topologically protected edge states. 
\end{enumerate}
 
\nit  We begin by stating our detailed
  {assumptions on $V(\bx)$ and $W(\bx)$.}\ There exists $\bx_0\in\R^2$ such that $\tilde{V}(\bx)=V(\bx-\bx_0)$
and $\tilde{W}(\bx)=W(\bx-\bx_0)$ satisfy the following:
\begin{equation}
\text{\rm{\bf Assumptions (V)}}
\label{V-assumptions}\end{equation}
\begin{enumerate}
\item[(V1)]  $\Lambda_h$- periodicity:\ \ $\tilde{V}(\bx+\bv)=  \tilde{V}(\bx)$ for all $\bv\in\Lambda_h$.
\item[(V2)]  Inversion symmetry:\  \ $\tilde{V}(\bx)=\tilde{V}(-\bx)$.  
\item[(V3)]  $2\pi/3$-rotational invariance:\ $\tilde{V}(R^*\bx)=\tilde{V}(\bx)$.
\item[(V4)]  Positivity of Fourier coefficient of $\eps V$, $\eps V_{1,1}$: $\eps \tilde{V}_{1,1}>0$, \\
where 
$\tilde{V}_{1,1}= \frac{1}{|\Omega|}\int_\Omega e^{-i(\bk_1+\bk_2)\cdot\by}\tilde{V}(\by) d\by$;
see \eqref{V11eq0} and \eqref{Omega-fourier}.
\end{enumerate}

\begin{equation}
\text{\rm{\bf Assumptions (W)}}
\label{W-assumptions}
\end{equation}
\begin{enumerate}
\item[(W1)~]   $\Lambda_h$- periodicity:\ \ $\tilde{W}(\bx+\bv)=  \tilde{W}(\bx)$ for all $\bv\in\Lambda_h$.
\item[(W2)~]  Anti-symmetry:\ \ $\tilde{W}(-\bx)= -\tilde{W}(\bx)$.
\item[(W3$^*$)]  Uniform nondegeneracy of $\tilde{W}$:\ \ Let $\Phi^\eps_j(\bx),\ j=1,2$ denote the $L^2_{\bK,\tau}$, respectively, $L^2_{\bK,\overline\tau}$, modes of the degenerate $L^2_\bK-$ eigenspace of $H^{(\eps,0)}=-\Delta+\eps V$. Then, there exists $\theta_0>0$, independent of $\eps$, such that  for all $\eps$ sufficiently small:
\begin{equation}
\left|\ \vartheta^\eps_\sharp\ \right|\ \equiv \left|\inner{\Phi_1^\eps,\tilde{W}\Phi_1^\eps}_{L^2_\bK} \right|\ge\theta_0>0\ .\label{uniform-nondegen}
\end{equation}
\end{enumerate}

\nit {N.B.\ \ Consistent with our earlier convention, in the following discussion, we shall drop the ``tildes'' on both $V$ and $W$. It will be  understood
that we have chosen coordinates with $\bx_0=0$.}

\begin{remark}\label{on-theta-sharp}
We claim that (W3*) (see \eqref{uniform-nondegen}) uniform non-degeneracy of $W$ is equivalent to the assumption:
\begin{equation}
\tag{W3*}
{W}_{0,1} + {W}_{1,0} - {W}_{1,1} \ne0 ,
\label{uniform-nondegen-W}
\end{equation}
where $\{W_\bfm\}_{ \bfm\in\Z^2}$ denote the Fourier coefficients of $W(\bx)$.
 To see this, note by
Proposition 3.1 of \cites{FW:12}, that for sufficiently small $\eps$, 
\[\Phi^\eps_1(\bx) = \frac{1}{\sqrt{3|\Omega|}} e^{i\bK\cdot\bx} \left[1+\overline{\tau}e^{i\bk_2\cdot\bx} + \tau e^{-i\bk_1\cdot\bx}\right] +\mathcal{O}(\eps); \]
see also \eqref{p_sigma}. 
Evaluation of $\vartheta_\sharp$ gives
\begin{align*}
\thetasharp^\eps &= \frac{1}{3} \int_\Omega \abs{1+\overline{\tau}e^{i\bk_2\cdot\bx} +\tau e^{-i\bk_1\cdot\bx}}^2 W(\bx)\ d\bx+\mathcal{O}(\eps) \nn\\
&= \frac{1}{3} \left[ ({W}_{0,1} + {W}_{1,0})\ (\tau - \overline{\tau}) + {W}_{1,1}(\tau^2 - \overline{\tau}^2) \right] + \mathcal{O}(\eps) \nn \\
&=i\frac{\sqrt3}3\ \left[{W}_{0,1} + {W}_{1,0} - {W}_{1,1} \right] + \mathcal{O}(\eps) ,
\end{align*}
which is nonzero if \eqref{uniform-nondegen-W} holds and $\eps$ is sufficiently small.
%
\end{remark}

Let $(\bK,E_\star^\eps)$ denote a Dirac point of $H^{(\eps,0)}=-\Delta + \eps V(\bx)$,
guaranteed to exist by Theorem \ref{diracpt-small-thm} for all $0<|\eps|<\eps_0$ and assume that $V(\bx)$ and $W(\bx)$ satisfy Assumptions (V), (W); see \eqref{V-assumptions}-\eqref{W-assumptions}.

\nit In our next result, we verify the spectral \nofold condition for the zigzag edge.  This is central to applying Theorem \ref{thm-edgestate} to prove our  result (Theorem \ref{Hepsdelta-edgestates}) on the existence of a family of  zigzag edge states. 

 \begin{theorem}\label{SGC!}\ \ 
 There exists a positive constant, $\eps_2\le\eps_0$, such that for any  $0<|\eps|<\eps_2$,  $H^{(\eps,0)}=-\Delta+\eps V$ satisfies the spectral \nofold condition at quasi-momentum $\bK$ along the zigzag slice; see Definition \ref{SGC}.\\  By Assumptions (V), $E_-(\lambda)\le 
E_+(\lambda)\le E_b(\lambda),\ b\ge3$. For any $\mathfrak{a}$ sufficiently small:
\begin{align*}
&b=\pm:\ \ \mathfrak{a} \le |\lambda|\le\frac12\ \implies\ \Big|\ E^{\eps,0}_b(\lambda) - E^\eps_\star\ \Big|\ \ge\ \frac{q^4}2\ |V_{1,1}\ \eps|\ \lambda^2\ge c_1\  |\eps|\ \mathfrak{a}^2 ,
\\
&b\ge3:\  \ |\lambda|\le1/2 \ \implies \ \Big| E^{\eps,0}_b(\lambda)-E^{\eps}_\star \Big|\ \ge\  c_2\ |\eps|\ (1+|b|)\ .
\end{align*}
\end{theorem}

\nit Theorem \ref{SGC!} controls the zigzag slice of the band structure at $\bK$, globally and outside a neighborhood of $(\bK,E^\eps_\star)$. Since a small perturbation of $V$ which breaks inversion symmetry opens a gap, locally about $\bK$ (see \cites{FW:12}), it can be shown that $H^{(\eps,\delta)}$ has a  $L^2_{\kpar=2\pi/3}-$ spectral gap about $E=E_\star^\eps$:
 \begin{theorem}\label{delta-gap}
 Assume $V(\bx)$ and $W(\bx)$ satisfy Assumptions (V) and  (W).
 Let $\eps_2$ be as in Theorem \ref{SGC!}. Then, there exists $c_\flat>0$ such that for all 
  $0<|\eps|<\eps_2$ and $0<\delta\le c_\flat\ \eps^2$, 
 the operator $-\Delta+\eps V(\bx)+\delta W(\bx)$ has a non-trivial $L^2_{\kpar=\bK\cdot\bv_1=2\pi/3}-$ spectral  gap about the energy $E=E^\eps_\star$. 
 \end{theorem}
 
\nit In the case where (V4) does not hold and the Fourier coefficient of $\eps V$, $\eps V_{1,1}$, is negative: $\eps \tilde{V}_{1,1}<0$, then we prove that the \nofold condition does not hold: 
 
 \begin{theorem}\label{NO-directional-gap!}[Conditions for non-existence of a zigzag spectral gap]
 Assume hypotheses (V1)-(V3) but, instead of hypothesis (V4), assume
 $\eps V_{1,1}<0$.
 Then, for any $\eps$ sufficiently small, the \nofold condition of Definition \ref{SGC} does not hold for the zigzag slice.
 \end{theorem}

\nit The proofs of Theorems \ref{SGC!} and \ref{delta-gap} are presented below. Section \ref{LSreduction} discusses a reduction to the lowest three spectral bands. The general strategy, based on an analysis of a $3\times3$ determinant, is given in Section \ref{sec:strategy}. Theorem \ref{NO-directional-gap!} is proved in Section \ref{no-directionalgap} as a consequence of Proposition \ref{Deq0}.

We prove the following theorem on zigzag edge states using results of Section \ref{thm-edge-state}. 
\begin{theorem}\label{Hepsdelta-edgestates}
 Let 
 $ H^{(\eps,\delta)} =-\Delta+\eps V(\bx)+\delta\kappa(\delta\bk_2\cdot\bx) W(\bx)$, 
  where $V(\bx)$ and $W(\bx)$ satisfy Assumptions (V) and (W), and $\kappa(\bX)$ is a domain wall function in the sense of Definition \ref{domain-wall-defn}. 
Let $\eps_2>0$ and $c_\flat>0$ be as in Theorem \ref{SGC!} and assume 
 $0<|\eps|<\eps_2$ and $0<|\delta|\le c_\flat \eps^2$. Then, 
there exist  edge states, $\Psi(\bx;k_\parallel)\in L^2_{k_\parallel}(\Sigma)$, with 
$|k_\parallel-2\pi/3|$ sufficiently small. 
Furthermore, continuous superposition in $\kpar$ yields wave-packets which are concentrated along the zigzag edge.
\end{theorem}

\begin{proof}[Proof of Theorem \ref{Hepsdelta-edgestates}]
We claim that the theorem is an immediate consequence of  the spectral \nofold condition for $-\Delta +\eps V$ for the zigzag edge, stated in Theorem \ref{SGC!}. This follows by an application of Theorem \ref{thm-edgestate} (and Corollary \ref{vary_k_parallel}).  
Since the details of the proof of Theorem \ref{thm-edgestate} are carried out for the case of $H^{(\eps,\delta)}$ with $\eps=1$, 
 we wish to point out how the proof applies with $\eps$ and $\delta$ varying as in the statement of Theorem \ref{Hepsdelta-edgestates}.

The proof of  Theorem \ref{thm-edgestate} uses a Lyapunov-Schmidt reduction strategy where the eigenvalue problem is reduced to an equivalent  eigenvalue problem (nonlinear in the eigenvalue parameter, $E$)  
for the Floquet-Bloch spectral components of the bound states in a neighborhood of the Dirac point $(\bK,E^\eps_\star)$. Stated for the relevant case of the zigzag edge, this reduction step requires the invertibility of an operator $(I-\mathcal{Q}_\delta)$ acting on $H_{\kpar=\bK\cdot\bv_1}^2(\Sigma)$,  where $\mathcal{Q}_\delta$ is defined in terms of Floquet-Bloch components in \eqref{tQ-def} for $\ktilde_2=\bk_2$.

It suffices to show that the $H_{\kpar=\bK\cdot\bv_1}^2(\Sigma)$ norm of $\mathcal{Q}_\delta$ is $o(1)$ as $\delta\downarrow0$. From \eqref{frak-e-bound} in Remark \ref{frak-e-remark}, the operator norm of $\mathcal{Q}_\delta$ is bounded by $\mathfrak{e}(\delta)$, given by:
\begin{equation*}
 \mathfrak{e}(\delta) \lesssim\ \frac{|\delta|}{\omega(\mathfrak{\delta^{\exponent}})c_1(V)}\ +\ \frac{|\delta|}{(1+|b|)c_2(V)} \ 
 \leq \frac{|\delta|^{1}}{\omega(\mathfrak{\delta^{\exponent}}) c_1(V)}\ +\ \frac{|\delta|}{c_2(V)}.
\end{equation*}
By Theorem \ref{SGC!}, the \nofold condition holds with modulus $\omega(\mathfrak{a})=\mathfrak{a}^2$ and constants 
\[ c_1(V)=\tilde{c}_1\frac{q^4}{2} |V_{1,1}\eps|\ ,\ \ \ \ c_2(V)=\tilde{c_2}|\eps|.\]
Therefore, with $\mathfrak{a}=\delta^\exponent$,  
 \begin{equation*}
  \mathfrak{e}(\delta) \lesssim\ \frac{|\delta|^{1-2\exponent}}{c_1(V)}\ +\ \frac{|\delta|}{c_2(V)} \
 \lesssim\ \frac{2}{q^4\tilde{c}_1 |V_{1,1}|}\cdot \frac{|\delta|^{1-2\exponent}}{|\eps|}\ +\ \frac{1}{\tilde{c_2}}\cdot \frac{|\delta|}{|\eps|}.
\end{equation*}
By hypothesis, $|\delta|\le c_\flat\eps^2$. Hence, $0\le \mathfrak{e}(\delta)
 \lesssim |\eps|^{1-4\exponent}+|\eps|\to0$ as $\eps\to0$ if we choose $\exponent\in(0,1/4)$.
 This completes the proof of Theorem \ref{Hepsdelta-edgestates}.
 \end{proof}

To prove Theorems \ref{SGC!} - \ref{NO-directional-gap!} we introduce a tool to localize the analysis about the lowest three bands. Throughout the analysis, without loss of generality, we take $\eps>0$ and satisfy the cases $\eps V_{1,1}>0$ and $\eps V_{1,1}<0$ by varying the sign of $V_{1,1}.$

\subsection{Reduction to the lowest three bands}\label{LSreduction}

In this subsection we show, via a Lyapunov-Schmidt reduction argument, that the proofs of Theorems \ref{SGC!} and \ref{delta-gap} can be reduced to the study of the lowest three spectral bands. To achieve this, we consider several parameter space regimes separately.

We start by considering $\eps=\delta=\lambda=0$. In this case, $H^{(0,0,0)}(\bK)=\left(\nabla+i\bK\right)^2$ has a triple eigenvalue, $E^0_\star=|\bK|^2$, with corresponding $3-$ dimensional $L^2(\R^2/\Lambda_h)-$ eigenspace spanned by the eigenfunctions $p_\sigma$, for $\sigma=1,\tau,\overline{\tau}$; see  Section \ref{dpts-generic-eps}.

Next, we turn on $\eps$, keeping $\delta=\lambda=0$.
From Theorem \ref{diracpt-small-thm}, there exists $\eps_0>0$ such that for $\eps\in I_{\eps_0}\equiv (-\eps_0,\eps_0)\setminus\{0\}$, the operator $H^{(\eps,0,0)}(\bK)=-\left(\nabla+i\bK\right)^2+\eps V(\bx)$ has a double $L^2(\R^2/\Lambda_h)$ eigenvalue, $E^\eps_\star$.
Let $E^0_\star=|\bK|^2$.
The maps $\eps\mapsto E^\eps_\star$ and $\eps\mapsto \widetilde{E}^\eps_\star$ are real analytic for $ \eps\in I_{\eps_0}$
with expansions \eqref{E*expand}, \eqref{tildeE*expand}:
\begin{align*}
E^\eps_\star\ &=  E_\star^0 + \eps(V_{0,0}-V_{1,1})+\mathcal{O}(\eps^2),\\
\widetilde{E}^\eps_\star &= E_\star^0 + \eps(V_{0,0}+2V_{1,1})+\mathcal{O}(\eps^2). 
\end{align*}


More generally, we may study the eigenvalue problem 
\begin{equation}
 (H^{(\eps,\delta,\lambda)}-E)p=0,\ \ p\in L^2(\R^2/\Lambda_h) , \text{\ \ with\ \ }
 E\ =\ E_\star^\eps + \mu ,
 \label{evp-eps-delta-lam}
\end{equation}
and seek, via Lyapunov-Schmidt reduction, to localize \eqref{evp-eps-delta-lam} about the three lowest zigzag slices.
Written out, the eigenvalue problem has the form:
\begin{equation}
\Big[ -\left(\nabla+i[\bK+\lambda\bk_2] \right)^2-E_\star^\eps - \mu\ +\ \eps V\ +\delta W\ \Big]p(\bx) = 0,\  \ p\in L^2(\R^2/\Lambda_h) .
\label{evp-p}
\end{equation}

Since $\eps$ and $\delta$ will be chosen to be small,  we shall expand $p(\bx)$ relative to the natural basis of the $L^2(\R^2/\Lambda_h)-$ eigenspace of the free operator, $H^{(0,0,0)}=-(\nabla+i\bK)^2$, displayed explicitly in \eqref{p_sigma}. Let $P^\parallel$ denote the projection onto 
${\rm span}\left\{p_\sigma:\sigma=1,\tau,\overline{\tau}\right\}$ and $P^\perp=I-P^\parallel$.

We seek a solution of \eqref{evp-p} in the form  $p(\bx) = p_\parallel(\bx) + p_\perp(\bx)$, 
where  
\begin{align*}
p_\parallel\in{\rm span}\left\{p_\sigma:\sigma=1,\tau,\overline{\tau}\right\},\ \ 
p_\perp\in \Big[\ {\rm span}\left\{p_\sigma:\sigma=1,\tau,\overline{\tau}\right\}\ \Big]^\perp .
\end{align*}
Then, we have that \eqref{evp-p} is equivalent to the coupled system of equations:
\begin{align}
&\Big[ -\left(\nabla+i[\bK+\lambda\bk_2] \right)^2-E_\star^\eps - \mu \Big]p_\parallel \label{p-parallel} \\
&\quad +
\Big[ \eps P^\parallel V + \delta P^\parallel W \Big]p_\parallel + 
\Big[ \eps P^\parallel V + \delta P^\parallel W \Big]p_\perp = 0, \nn \\
&\Big[ -\left(\nabla+i[\bK+\lambda\bk_2] \right)^2-E_\star^\eps - \mu \Big]p_\perp \label{p-perp} \\
&\quad + \Big[ \eps P^\perp V + \delta P^\perp W \Big]p_\perp + 
\Big[ \eps P^\perp V + \delta P^\perp W \Big]p_\parallel = 0 . \nn
\end{align}

\begin{proposition}\label{invert-on-Pperp}
There exists a constant $c>0$ such that if $|\eps|+|\mu|<c$ then
\begin{equation*}
H^{(\eps,0,\lambda)}-\mu\ \equiv\ \Big[ -\left(\nabla+i[\bK+\lambda\bk_2] \right)^2-E_\star^\eps - \mu \Big]\ \  \textrm{is invertible on}\ P^\perp L^2(\R^2/\Lambda_h) .
\end{equation*}
\end{proposition}

\begin{proof}[Proof of Proposition  \ref{invert-on-Pperp}]
This follows from the lower bound:
\begin{equation*}
0<g_0 \equiv \inf_{|\lambda |\le1/2}\ \inf_{ \bfm\notin \{ (0,0),(0,1),(-1,0) \} }
\Big|\ |\bK+\bfm\vec{\bk}+\lambda\bk_2|^2-|\bK|^2\ \Big| ,
\end{equation*}
where $\bfm\vec{\bk}=m_1\bk_1+m_2\bk_2$.
\end{proof}

By Proposition \ref{invert-on-Pperp}, for all $\eps$ and $\mu$ sufficiently small  equation \eqref{p-perp} is equivalent to:
\begin{align*}
&\Big[\ I\ +\ \left[H^{(\eps,0,\lambda)}-\mu \right]^{-1}\
 \left[ \eps P^\perp V + \delta P^\perp W \right]\ \Big] p_\perp \\
&\qquad =\ -\left[H^{(\eps,0,\lambda)}-\mu\right]^{-1}\ \left[ \eps P^\perp V + \delta P^\perp W \right]p_\parallel  .
\end{align*}
Suppose now $|\eps|+|\delta|+|\mu|<d_1$, where $0<d_1\le c$ is chosen sufficiently small. Then we have 
\begin{align}
p_\perp\ 
&= -\Big[\ I\ +\ \left[H^{(\eps,0,\lambda)}-\mu \right]^{-1}\
 \left[ \eps P^\perp V + \delta P^\perp W \right]\ \Big]^{-1} \times \nn \\
&\qquad\qquad \left[H^{(\eps,0,\lambda)}-\mu\right]^{-1}\ \left[ \eps P^\perp V + \delta P^\perp W \right]p_\parallel 
 =\ \mathcal{P}(\eps,\delta,\lambda,\mu)\ p_\parallel  . \label{p-perp2}
 \end{align}
 Equation \eqref{p-perp2} defines a bounded linear mapping on $P^\parallel L^2(\R^2/\Lambda_h)$ with operator norm $\lesssim 1$:
 \[ p_\parallel\mapsto p_\perp[\eps,\delta,\lambda;p_\parallel]=\mathcal{P}(\eps,\delta,\lambda,\mu)\ p_\parallel  :\  Ran\left(P^\parallel\right)\to Ran\left(P^\perp\right) , \]
which is analytic in $\eps, \delta, \lambda$ and $\mu$ for $|\eps|+|\delta|+|\mu|<d_1$ and $|\lambda|<1/2$.
Consequently, equation \eqref{p-parallel} becomes a closed equation for $p_\parallel$
 which we write as:
\begin{equation*}
\mathcal{M}(\eps,\delta,\lambda,\mu) p =\ 0 ,
\end{equation*}
where
\begin{align}
\mathcal{M}(\eps,\delta,\lambda,\mu) &\equiv 
 P^\parallel\Big[ -\left(\nabla+i[\bK+\lambda\bk_2] \right)^2-E_\star^\eps - \mu 
+  \eps V + \delta W  \nn\\
&\quad
  + ( \eps  V + \delta W ) \mathcal{P}(\eps,\delta,\lambda,\mu)  
  \Big]  P^\parallel \label{calMdef1}\\
&= 
P^\parallel \Big[ -\left(\nabla+i[\bK+\lambda\bk_2] \right)^2-E_\star^\eps - \mu + \eps V + \delta W\nn\\
& \quad
  + ( \eps  V + \delta W  )  P^\perp \left[ H^{(\eps,0,\lambda)}-\mu+ \eps  V + \delta W \right]^{-1}P^\perp  ( \eps  V + \delta W ) \Big] P^\parallel \ . \label{calMdef2}
\end{align}
 The operator $\mathcal{M}(\eps,\delta,\lambda,\mu)$ acts on the three-dimensional 
  space $P^\parallel L^2(\R^2/\Lambda_h)={\rm span}\{p_\sigma:\sigma=1,\tau,\overline{\tau}\}$.
Moreover, from the expression \eqref{calMdef2} it is clear that for $\eps,\delta,\lambda,\mu$ complex, $\mathcal{M}(\eps,\delta,\lambda,\mu)$ is self-adjoint. 
We shall study it via its matrix representation relative to the basis $\{p_\sigma:\sigma=1,\tau,\overline{\tau}\}$:
\begin{equation}
M_{\sigma,\tsigma}(\eps,\delta,\lambda,\mu) \equiv \left\langle p_\sigma,\mathcal{M}(\eps,\delta,\lambda,\mu) p_\tsigma\right\rangle_{L^2(\R^2/\Lambda_h)},\ \ \sigma=1,\tau,\overline{\tau}\ .
\label{Msts}\end{equation}
Clearly, $M(\eps,\delta,\lambda,\mu)$ is Hermitian for  $\eps,\delta,\lambda,\mu$ real:
$ M_{\sigma,\tsigma}(\eps,\delta,\lambda,\mu) = \overline{M_{\tsigma,\sigma}(\eps,\delta,\lambda,\mu)}$.

Summarizing the above discussion we have the following:

\begin{proposition}\label{Msigtsig}
\begin{enumerate}
\item The entries of the $3\times3$ 
matrix $M(\eps,\delta,\lambda,\mu)$ are analytic functions of complex
$\eps, \delta, \lambda$ and $\mu$ for $|\eps|+|\delta|+|\mu|<d_1$ and $|\lambda|<1/2$.
\item For $|\eps|+|\delta|+|\mu|<d_1$ and $|\lambda|<1/2$, we have that
$E=E^\eps_\star+\mu$ is an $L^2(\R^2/\Lambda_h)-$ eigenvalue of $H^{(\eps,\delta,\lambda)}$
if and only if
$\det M(\eps,\delta,\lambda,\mu)=0$.
\end{enumerate}
\end{proposition}

We now study the roots of $\det M(\eps,\delta,\lambda,\mu)=0$ (eigenvalues of 
$H^{(\eps,\delta,\lambda)})$ for $\eps$ and $\delta$ small.
First let  $\eps=\delta=0$. By the formulae \eqref{psig-nablaK-tpsig} and \eqref{J-matrix-computed}, derived and also applied in Section \ref{Mexpansion}, we have:
\begin{align*}
M(0,0,\lambda,\mu) &= \Big(\ \left\langle p_\sigma\ ,\ \mathcal{M}(0,0,\lambda,\mu)  p_\tsigma\right\rangle_{L^2(\R^2/\Lambda_h)}\ \Big)_{ \sigma,\tsigma=1,\tau,\overline\tau}\ ,
\nn \\
&= \Big(\ \left\langle p_\sigma\ ,\  \left(-(\nabla+i(\bK+\lambda\bk_2))^2\right), p_\tsigma\right\rangle_{L^2(\R^2/\Lambda_h)}\ -\ (E_\star^0+\mu)\delta_{\sigma,\tsigma}\ \Big)_{ \sigma,\tsigma=1,\tau,\overline\tau}\ ,
\nn \\
 &= \begin{pmatrix}
  \lambda^2q^2-\mu & \alpha\lambda & \overline{\alpha}\lambda\\
  \overline{\alpha}\lambda &\lambda^2q^2-\mu & \alpha\lambda
   \\ \alpha\lambda & \overline{\alpha}\lambda & \lambda^2q^2-\mu
 \end{pmatrix} , \ \ \alpha= \frac{q^2}{\sqrt3}\ i\tau.
 \end{align*} 
 %
 %
 Thus,  $M(0,0,\lambda,\mu)$ is singular if $\mu$ is a root of the polynomial:
 \begin{align*}
 &\det M(0,0,\lambda,\mu) \\
 &\quad = -\mu^3 + \mu^2 (3q^2\lambda^2) + \mu(3\lambda^2|\alpha|^2-3\lambda^4q^4)
 +\lambda^6q^6 -3\lambda^4q^2|\alpha|^2+2\lambda^3\Re(\alpha^3) \nn \\
 &\quad = -\mu^3 + \mu^2(3\lambda^2q^2) + \mu(\lambda^2q^4-3\lambda^4q^4) + \lambda^6q^6 + \lambda^4q^6,  
 \end{align*}
 where we have used that $|\alpha|^2=\frac{q^4}{3}$ and $\Re(\alpha^3)=0$.
The roots, $\mu^{(0)}_j(\lambda),\ j=1,2,3$, defined by the ordering
 $\mu^{(0)}_1(\lambda)\le\mu^{(0)}_2(\lambda)\le\mu^{(0)}_3(\lambda)$,  
 listed with multiplicity, are given by:
  \begin{align}
{\bf 0\le\lambda\le\frac12:}\ 
 &\mu^{(0)}_1(\lambda)=q^2\lambda (\lambda-1),\ \mu^{(0)}_2(\lambda)=q^2\lambda^2,\ \mu^{(0)}_3(\lambda)=q^2\lambda (\lambda+1)  , \label{mulam>0}\\
{\bf -\frac12\le\lambda\le0:}\ 
 &\mu^{(0)}_1(\lambda)=q^2\lambda (\lambda+1),\ \mu^{(0)}_2(\lambda)=q^2\lambda^2,\ \mu^{(0)}_3(\lambda)=q^2\lambda (\lambda-1)  . \label{mulam<0}
 \end{align}

The roots $\mu^{(0)}_j(\lambda), j=1,2,3$, are eigenvalues of $-\left(\nabla+i[\bK+\lambda\bk_2] \right)^2-E_\star^0$. They are uniformly bounded away from all other $L^2(\R^2/\Lambda_h)-$ spectrum. This distance is given by
\begin{equation*}
g_0 \equiv \inf_{|\lambda |\le1/2}\ \inf_{ \bfm\notin \{ (0,0),(0,1),(-1,0) \} }
\Big|\ |\bK+\bfm \vec{\bk}+\lambda\bk_2|^2-|\bK|^2\ \Big|\ > 0  ,\ \ (E^0_\star=|\bK|^2), 
\end{equation*}
where $\bfm \vec{\bk}=m_1\bk_1+m_2\bk_2$.

Note that $\det M(0,0,0,\mu)=-\mu^3$ has a root of multiplicity three: $\mu=0$. By Proposition \ref{Msigtsig} and analyticity of $\det M(\eps,\delta,\lambda,\mu)$, we have that   for $\eps, \delta$ and $\lambda$ in a small  neighborhood of the origin in $\C^3$, there are three solutions of 
$\det M(\eps,\delta,\lambda,\mu)=0$. We label them $\mu_j(\eps,\delta,\lambda),\ j=1,2,3$. By self-adjointness of $M(\eps,\delta,\lambda,\mu)$, these roots are real for real values of $\eps$ and $\delta$. Correspondingly, 
for small and real $\eps, \delta$ and $\lambda$ there are three real $L^2_{k_\parallel=2\pi/3}-$ eigenvalues of $H^{(\eps,\delta,\lambda)}$,  denoted
\[ E_j^{(\eps,\delta)}(\lambda)\equiv E_j^{(\eps,\delta)}(\bK+\lambda\bk_2)\equiv E^\eps_\star+\mu_j(\eps,\delta,\lambda),\  \ \ j=1,2,3.\]
The ordering of the $E_j$ implies 
\[ \mu_1(\eps,\delta,\lambda)\le \mu_2(\eps,\delta,\lambda)\le \mu_3(\eps,\delta,\lambda) .\]
A mild extension of Proposition \ref{directional-bloch} yields
\begin{proposition}\label{roots-delta0}
For each $|\lambda|\le1/2$, there exist orthonormal $L^2_{\bK+\lambda\bk_2}-$ eigenpairs $(\Phi_\pm(\bx;\lambda),E_\pm(\lambda))$ and  $(\widetilde{\Phi}(\bx;\lambda),\widetilde{E}(\lambda))$,  real analytic in $\lambda$,  such that 
\[ \textrm{span}\ \{\Phi_{-}(\bx;\lambda), \Phi_{+}(\bx;\lambda), \widetilde{\Phi}(\bx;\lambda)\}
 =\ \textrm{span}\ \{\Phi_j(\bx;\bK+\lambda\bk_2) : j=1,2,3 \}\ .
 \]
   For fixed $\eps$ small and $\lambda$ tending to $0$ we have:
\begin{enumerate}
\item 
\begin{align}
E_\pm^{\eps,\delta=0}(\lambda)\ &=\ E_\star^\eps\ \pm\ |\lambda^\eps_\sharp|\ |z_2|\ \lambda\ +\ \lambda^2\ e_\pm(\lambda,\eps) , \label{Epm-expand}\\
\widetilde{E}^{\eps,\delta=0}(\lambda)\ &=\ \widetilde{E}_\star^\eps\   +\ \lambda^2\ \widetilde{e}(\lambda,\eps) , \label{Etilde-expand}
\end{align}
where $e_\pm(\lambda,\eps),\  \widetilde{e}(\lambda,\eps)\ =\ \mathcal{O}(1)$  as $\lambda, \eps\to0$, and $z_2=k_2^{(1)}+ik_2^{(2)}$. These functions can be represented in a convergent power series in $\eps$ and $\lambda$ in a fixed $\C^2$ neighborhood of the origin. Furthermore, $e_\pm(\lambda,\eps),\  \widetilde{e}(\lambda,\eps)$ are 
real-valued for real $\lambda$ and $\eps$.
\item Therefore, for all small $\eps$ and $\lambda$,  the three roots of $\det M(\eps,\delta=0,\lambda,\mu)$ may be labeled:
\begin{align}
\mu_+(\lambda,\eps)\ &=\ |\lambda^\eps_\sharp|\ |z_2|\ \lambda\ +\ \lambda^2\ e_+(\lambda,\eps)\label{mu+expand}\\
\mu_-(\lambda,\eps)\ &= -|\lambda^\eps_\sharp|\ |z_2|\ \lambda\ +\ \lambda^2\ e_-(\lambda,\eps)\label{mu-expand}\\
\tilde{\mu}(\lambda,\eps)\ &=\ \widetilde{E}_\star^\eps-E_\star^\eps\ +\  \lambda^2\ \widetilde{e}(\lambda,\eps)\ ,\ \ {\rm where}\ \widetilde{E}_\star^\eps-E_\star^\eps= 3\eps V_{1,1}\ +\ \mathcal{O}(\eps^2) . \label{mutilde-expand}
\end{align}
\end{enumerate}
\end{proposition}

\begin{proof}[Proof of Proposition \ref{roots-delta0}]
Part 1 can be proved as follows. The expansion \eqref{Epm-expand} and analyticity follow by perturbation theory as in Proposition \ref{directional-bloch}; see also  \cites{Friedrichs:65,kato1995perturbation}.  The expansion \eqref{mutilde-expand} also follows by perturbation theory of the simple eigenvalue $\widetilde{E}^\eps_\star $ for $\lambda=0$.
 It is easy to see that 
 \begin{equation}
 \widetilde{E}^{\eps,\delta=0}(\lambda)\ 
 =\ \widetilde{E}_\star^\eps\ +\
 \lambda\ \times\left(\  -2i\bk_2\cdot \left\langle\widetilde{\Phi}^\eps,\nabla_\bx\widetilde{\Phi}^\eps
\right\rangle_{L^2_\bK}\ \right)\ +\ \lambda^2\ \widetilde{e}(\eps,\lambda).
\end{equation}
We claim that $\left\langle\widetilde{\Phi}^\eps,\nabla_\bx\widetilde{\Phi}^\eps
\right\rangle_{L^2_\bK}=0$. This follows since $\widetilde{\Phi}^\eps\in L^2_{\bK,1}$ as in the proof of Proposition 4.1 of \cites{FW:12}. This proves the expansion \eqref{Etilde-expand}.
\end{proof}

Finally we discuss a useful symmetry of $\det M(\eps,\delta,\lambda,\mu=0)$.

\begin{proposition}\label{symmetry-detM}
Assume $V$ satisfies Assumptions (V) and $W$ satisfies Assumptions (W).
Recall $M(\eps,\delta,\lambda,0)$, defined in \eqref{calMdef1}-\eqref{Msts}.  
Then, $\det M(\eps,\delta,\lambda,0)$ is real-valued  for real $\eps, \delta, \lambda$ and analytic in a small neighborhood of the origin in $\C^3$. Furthermore, 
\begin{align}
\det M(\eps,\delta,\lambda,0)&=\det M(\eps,-\delta,\lambda,0)
\label{delta-minusdelta}
\end{align}
and therefore $\det M(\eps,\delta,\lambda,0)$ is a real analytic in $\eps$, $\delta^2$ and $\lambda$, and we write
\begin{align*}
D(\eps,\delta^2,\lambda)\equiv \det M(\eps,\delta,\lambda,0) &
\end{align*}
\end{proposition}

\begin{proof}[Proof of Proposition \ref{symmetry-detM}]
Let $\mathcal{I}[f](\bx)=f(-\bx)$ and $\mathcal{C}[f](\bx)=\overline{f(\bx)}$.  Using that $\eps, \delta$ and $\lambda$ are real, $V(\bx)$ is even and $W(\bx)$ is odd, one can check directly that 
\begin{align}
\mathcal{C}\circ\mathcal{I}\circ H^{\eps,-\delta,\lambda} &= H^{\eps,\delta,\lambda}\circ
 \mathcal{I}\circ\mathcal{C}\label{sym1}
 \end{align}
  Furthermore, note that
  \begin{align}
  &(\mathcal{C}\circ\mathcal{I}) p_1 = p_1,\quad  
  (\mathcal{C}\circ\mathcal{I}) p_\tau = p_{\overline\tau},\quad  
  (\mathcal{C}\circ\mathcal{I}) p_{\overline\tau} = p_\tau.
 \label{CIp} \end{align}
It follows that 
 \begin{align}
 \mathcal{C}\circ\mathcal{I}\circ P^\parallel &= P^\parallel\circ \mathcal{C}\circ\mathcal{I} .
 \label{sym2}\end{align}
Using the symmetry relations \eqref{sym1}-\eqref{sym2}
 to rewrite $M(\eps,-\delta,\lambda,0)$ in terms of $H^{\eps,\delta,\lambda}$, 
 we find that  $M(\eps,-\delta,\lambda,0)$ can be transformed into
    $M(\eps,\delta,\lambda,0)$ by interchanging its second and third rows, and then  interchanging its second and third columns. 
      Therefore,  $\det M(\eps,\delta,\lambda,0)=\det M(\eps,-\delta,\lambda,0)$ and the proof of Proposition \ref{symmetry-detM} is complete.
\end{proof}

\subsection{Strategy for analysis of $\det M(\eps,\delta,\lambda,\mu)$ in the case $\eps V_{1,1}>0$} \label{sec:strategy}

We first observe that for a positive constant, $d_1$,  if $|\eps|+|\delta|<d_1$ then 
\begin{equation}
|E_j^{(\eps,\delta)}(\lambda)-E_\star^\eps|\  \ge c_4(1+|j|),\ \ j\ge4 .
\label{outer-lams-bbig}\end{equation}
Indeed, by the   discussion following Proposition \ref{Msigtsig}, we have that there exists $d_2>0$ such that for $j\geq4$,  
$|\mu_j(\eps,\delta,\lambda)|\equiv |E_j^{(\eps,\delta)}(\lambda)-E_\star^\eps|\ge d_2$; 
the lower bound \eqref{outer-lams-bbig} now follows from the Weyl asymptotics for eigenvalues of second order elliptic operators in two space dimensions.  
%
%

Hence, we restrict our attention to $E_j^{(\eps,\delta)}(\lambda)=E_\star^\eps+\mu_j(\eps,\delta),\ j=1,2,3$, which we study by a detailed analysis of $\det M(\eps,\delta,\lambda,\mu=0)$. The analysis consists of verifying two steps, which we now outline.

\nit {\bf Step 1:}\  
 Fix $c_\flat>0$ and arbitrary. We will prove that there exists $C_\flat>0$, such that the following holds. There exists  $\eps_1>0$ and constant $c_3$, depending on $V$ and $W$,  such that
 for all $0<|\eps|<\eps_1$ and $0\le|\delta|\le c_\flat\ \eps^2$:
\begin{align}
C_\flat\sqrt{\eps}\le |\lambda|\le\frac12\ \implies \  |E_j^{(\eps,\delta)}(\lambda)-E_\star^\eps|\ & \ge c_3 \eps,\ \ j=1,2,3 \label{outer-lams} .
\end{align}
\nit Furthermore, by \eqref{outer-lams} and \eqref{outer-lams-bbig},  it follows that 
 \begin{align}
 & C_\flat\sqrt{\eps}\le |\lambda|\le\frac12\ \implies \ L^2(\R/\Lambda_h)-{\rm spec}( H^{(\eps,\delta,\lambda)})\ \cap  \Big[E_\star^\eps-c\eps,E_\star^\eps+c\eps\Big]\ \ \textrm{is empty}. \label{Step1}
 \end{align}

\nit {\bf Step 2:}\  Let $c_\flat$ and $C_\flat$ be as in Step 1. We will prove that there exists $0<\eps_2\le \eps_1$, such that  if $(\lambda,\delta)$ are in the set:
\begin{equation}
 |\lambda|\le C_\flat\ \eps^{\frac12}\ \ \textrm{and}\ \ 0\le |\delta|\ \le c_\flat\ \eps^2,\ \ {\rm where} \ 0<|\eps|<\eps_2,\ \ {\rm then}\label{smiley}
 \end{equation}
\begin{align}
\det M(\eps,\lambda,\delta,0)\ =\ 
 \left(q^2\lambda^2 + \eps V_{1,1}\right) \left( q^4\lambda^2 + \delta^2\ \left|W_{0,1}+W_{1,0}-W_{1,1} \right|^2 \right)\times\left(\ 1+o(1)\ \right).
 \label{Dlower}
\end{align}
A simple lower bound on the three eigenvalues $|\mu_j(\eps,\lambda)|=|E_j^\eps(\lambda)-E^\eps_\star|$ is then obtained as follows.
By \eqref{Dlower}, for some positive constants: $C_1$ and $C_2$, we have
{\small{ 
\begin{align*}
&C_1(\lambda^2+\eps)\cdot (\lambda^2+\delta^2)\ \le\ \left| \det M(\eps,\lambda,\delta,0) \right| \\
 &= \left| \det M(\eps,\lambda,\delta,0) - \det M (\eps,\lambda,\delta,\mu_j(\lambda,\eps,\delta))\right|\ \le C_2\ |0-\mu_j(\eps,\delta,\lambda)| = C_2\ |\mu_j(\eps,\delta,\lambda)|,
\end{align*}
}}
for $j=1,2,3$.  
Therefore, with $C_3=C_1/C_2$, 
\begin{equation}
\left| E_j^{\eps,\delta}(\lambda)-E_\star^\eps\right|  = |\mu_j(\eps,\delta,\lambda)| \geq C_3\ \eps\ (\lambda^2+\delta^2) \ \ge\ C_3\ \eps\ \delta^2,\ \ j=1,2,3.
\label{Ediff-lb}\end{equation}
 It follows from \eqref{Ediff-lb} and \eqref{outer-lams-bbig} that
\begin{equation}
\textrm{the $L^2(\R^2/\Lambda_h)-$\ spec( $H^{(\eps,\delta,\lambda)}$ )$\ \cap\  [E_\star^\eps-\eta,E_\star^\eps+\eta]$\ \ is empty}
\nn\end{equation}
with $\eta= \frac{1}{2}C_3\ \eps\ \delta^2$, whenever  $\eps$, $(\lambda,\delta)$ satisfy the constraints \eqref{smiley}.

Theorem \ref{SGC!} and Theorem \ref{delta-gap} are immediate consequences of \eqref{Step1} (Step 1) and \eqref{Ediff-lb} (Step 2) and the representation:
\begin{equation*} \left. H^{(\eps,\delta)}\ \right|_{L^2(\Sigma_{k_\parallel=2\pi/3})}\ =\
  \oplus\ \int_{|\lambda|\le\frac12}\  \left.H^{(\eps,\delta,\lambda)}\ \right|_{L^2(\R^2/\Lambda_h)}\ d\lambda .
  \end{equation*}

Hence, we now turn to their proofs. We first carry out Step 1  by a simple perturbation analysis about the free Hamiltonian, $H^{(0,0,\lambda)}$. We then turn to Step 2, which is much more involved.

{\bf Verification of \eqref{outer-lams} from Step 1.}
 Let $C_\flat$ denote a positive constant, which we will specify shortly, and consider the range $C_\flat\sqrt{\eps}\le\lambda\le1/2$. Using the expressions for $\mu^{(0)}_j(\lambda),\ j=1,2,3$,  in \eqref{mulam>0} we have that if $\eps\le \eps'(C_\flat)\equiv(4C_\flat^2)^{-1}$, then 
$  |E_1^{(0)}(\lambda)-E_\star^0|= q^2|\lambda|\ |1-\lambda|\ge C_\flat q^2 \sqrt{\eps}/2$,\ 
  $|E_2^{(0)}(\lambda)-E_\star^0|= q^2|\lambda|^2\ \ge C_\flat^2 q^2 \eps$,\  
and $|E_3^{(0)}(\lambda)-E_\star^0|= q^2|\lambda|\ |1+\lambda|\ge C_\flat^2 q^2 \sqrt{\eps}/2$.
Note that the eigenvalues $(\eps,\delta,\lambda)\mapsto E^{\eps,\delta}(\lambda)$ are Lipschitz continuous functions;
see Chapter XII of \cites{RS4} or Appendix A of \cites{FW:14}.
Therefore, we have
\begin{equation*}
|E_2^{(\eps,\delta)}(\lambda)-E_\star^\eps|\ \ge C_\flat^2\ q^2\ \eps\ -\ |\eps|\ \|V\|_\infty\ -\ |\delta|\ ||W\|_\infty\ge \frac{1}{2}\ C_\flat^2\ q^2\ \eps ,
\end{equation*}
for some $C_\flat$ positive, finite and sufficiently large. With this choice of $C_\flat$, we also have
\begin{align*}
|E_1^{(\eps,\delta)}(\lambda)-E_\star^\eps|&\ \ge C_\flat\ q^2\ \sqrt{\eps}/2\ -\ |\eps|\ \|V\|_\infty\ -\ |\delta|\ ||W\|_\infty , \nn\\
 |E_3^{(\eps,\delta)}(\lambda)-E_\star^\eps&|\ge C_\flat\ q^2\ \sqrt{\eps}/2\ -\ |\eps|\ \|V\|_\infty\ -\ |\delta|\ ||W\|_\infty .
\end{align*}
Therefore, there exists $\eps_1>0$, such that for all $0<|\eps|<\eps_1$ we have
\begin{equation*}
|E_1^{(\eps,\delta)}(\lambda)-E_\star^\eps|\ +\ |E_3^{(\eps,\delta)}(\lambda)-E_\star^\eps|\ \ge\ C_\flat\ q^2\ \sqrt{\eps}/4 .
\end{equation*}
This completes the proof of the assertions in Step 1.

\subsection{Expansion of $M(\eps,\delta,\lambda,0)$ and its determinant for $\eps V_{1,1} \neq  0$}
\label{Mexpansion}

The key to verifying Step 2 and proving Theorem \ref{SGC!} is the following:

\begin{proposition}\label{detM-expansion}
Let $C_\flat$ be as chosen in Step 1; see \eqref{Step1}.
 Then, there exist constants $\eps_2>0$ and $ c>0$ such that for all $0\le\eps<\eps_2$, if 
 \begin{equation}
 0\le|\lambda|\ \le\ C_\flat\ \eps^{1/2}\ \ {\rm and}\ \ 0\le |\delta| \le c\ \eps^2,\ \ {\rm  then}
 \label{smiley3}
 \end{equation}
\begin{align}
-\det M(\eps,\delta,\lambda,0)\
&=
 \pi(\eps,\delta^2,\lambda) \ + \  o\left( \pit \right) ,
\label{detM-exp}
\end{align}
where
\begin{align}
\pi(\eps,\delta^2,\lambda)&\equiv \left(q^2\lambda^2 + \eps V_{1,1}\right) \left( q^4\lambda^2 + \delta^2\ \left|W_{0,1}+W_{1,0}-W_{1,1} \right|^2 \right).
\label{pi-def}\end{align}
Here, $W_\bfm$, $ \bfm\in\Z^2$, denote Fourier coefficients of $W(\bx)$ and, by  \eqref{smiley3},  the correction term in \eqref{detM-exp} divided by $(\lambda^2+\eps)(\lambda^2+\delta^2)$ tends to zero as $\eps$ tends to zero. 
Thus,  
\begin{align*}
\eps V_{1,1}>0\ \implies\ -\det M(\eps,\delta,\lambda,0)\
&=
 \pi(\eps,\delta^2,\lambda)\ \left(1 +  o(1) \right)\ \textrm{in the region \eqref{smiley3}.}
\end{align*}

\end{proposition}
%

\nit We now embark on an expansion of $\det M(\eps,\delta,\lambda,0)$ and  the proof of Proposition \ref{detM-expansion}. $M(\eps,\delta,\lambda,0)$, see \eqref{calMdef1}-\eqref{Msts},
may be written as the sum of matrices:
\begin{equation}
M(\eps,\delta,\lambda,0)=\ \left[\ M^0+M^V+M^W+M^{\mathcal{P}}\ \right](\eps,\delta,\lambda,0),
 \label{Msum}
\end{equation}
where the $\sigma,\tsigma=1,\tau,\overline{\tau}$ entries are given by:
\begin{align}
M^0_{\sigma,\tsigma}(\eps,\delta,\lambda,0)\ &=\ \left\langle p_\sigma,\left(\ -\left(\nabla+i[\bK+\lambda\bk_2] \right)^2-E_\star^\eps \ \right) p_\tsigma\right\rangle_{L^2(\R^2/\Lambda_h)} , \label{M0}\\
M^V_{\sigma,\tsigma}(\eps)\ &=\ \eps \left\langle p_\sigma, V p_\tsigma\right\rangle_{L^2(\R^2/\Lambda_h)} , \nn \\ 
M^W_{\sigma,\tsigma}(\delta) &=\  \delta \left\langle p_\sigma, W p_\tsigma\right\rangle_{L^2(\R^2/\Lambda_h)} , \nn \\ 
M^{\mathcal{P}}_{\sigma,\tsigma}(\eps,\delta,\lambda,0)\ &=\   \left\langle p_\sigma,(\eps  V + \delta W)\ 
\mathcal{P}(\eps,\delta,\lambda,0)\ p_\tsigma\right\rangle_{L^2(\R^2/\Lambda_h)} . \nn 
\end{align}
For $\eps$ and $\delta$ small, 
\begin{align}
&M^0_{\sigma,\tsigma}(\eps,\delta,\lambda,0)\ =\ 
M^{0,approx}_{\sigma,\tsigma}(\eps,\lambda)\ +\ \mathcal{O}(\eps^2),  \label{M0-M0app} 
\end{align}
where (inner products over $L^2(\R^2/\Lambda_h)$)
{\small{
\begin{align}
&M^{0,approx}_{\sigma,\tsigma}(\eps,\lambda) =\left\langle p_\sigma,\left( -\left(\nabla+i[\bK+\lambda\bk_2] \right)^2- [E_\star^0 + \eps(V_{0,0}-V_{1,1})] \right) p_\tsigma\right\rangle, \ {\rm and}  \label{M0-expanded}\\ 
 & M^{\mathcal{P}}_{\sigma,\tsigma}(\eps,\delta,\lambda,0) 
 \label{MP-expanded}\\
  &\quad = 
\left\langle (\eps V + \delta W) p_\sigma , P^\perp\left(-(\nabla+i[\bK+\lambda\bk_2])^2-E_\star^\eps +\eps V +\delta W\right)^{-1}P^\perp  (\eps V + \delta W) p_\tsigma\right\rangle \nn \\
&\quad = \left\langle (\eps V + \delta W) p_\sigma , P^\perp\left(-(\nabla+i[\bK+\lambda\bk_2])^2-E_\star^\eps \right)^{-1}P^\perp  (\eps V + \delta W) p_\tsigma\right\rangle \nn \\
&\quad\qquad +\mathcal{O}(\delta^3 + \delta^2\eps + \delta \eps^2 + \eps^3)   = \mathcal{O}(\eps^2 + \eps\delta + \delta^2) .
\nn
\end{align}}}

\nit We next explain that to calculate the determinant of $ M(\eps,\delta,\lambda,0)$ to the desired order in the region \eqref{smiley}, it suffices to calculate the determinant of the approximate matrix:
\begin{equation}
 M^{approx}(\eps,\delta,\lambda)\ \equiv \  M^{0,approx}(\eps,\lambda)+M^V(\eps)+M^W(\delta) . \label{detMapprox}
\end{equation}
That is,  we show that the omitted terms  in  $M(\eps,\delta,\lambda,0)-M^{approx}(\eps,\delta,\lambda,0)$ contribute negligibly to the determinant of $ M(\eps,\delta,\lambda,0)$,  when compared with the polynomial, $\pi(\eps,\delta^2,\lambda)$,  in \eqref{pi-def}, provided $\lambda$ and $\delta$ are in the region \eqref{smiley}: 
\begin{equation}
  |\lambda|\ \le C_\flat\ \eps^{\frac12}\ \ \textrm{and}\ \ |\delta|\le c_\flat\ \eps^2 ,
  \label{smiley2}
  \end{equation}
 where $C_\flat$ and $c_\flat$ are appropriately chosen constants.

  Recall that $D(\eps,\delta^2=0,\lambda=0)= \det M (\eps,0,0,0) =0$ for all $\eps$, since $\mu=0$ corresponds to $E=E^\eps_\star$, which is an eigenvalue of $H^{(\eps,\delta=0,\lambda=0)}=-\Delta+\eps V$.
   Thus, $D(\eps,\delta^2,\lambda)=\det M(\eps,\delta,\lambda,0)$ is a convergent power series in $\eps$, $\delta^2$ and $\lambda$ with no ``pure $\eps$'' terms. On the other hand,  
   the entries of the matrix $M(\eps,\delta,\lambda,0)$ are convergent power series in $\eps, \delta$ and $\lambda$. 
   
   \begin{proposition}\label{ld2}
 For $(\lambda,\delta)$ in the region \eqref{smiley2} we have
\[ |\lambda|\ \delta^2 = o \left( \pit \right) .\] 
 Therefore we may drop the $\mathcal{O}(\lambda\delta^2)$ terms, a further simplification.
 %
 \end{proposition}
 
\begin{proof}[Proof of Proposition \ref{ld2}]
Consider separately the two regimes: (a) $|\lambda|\le \eps^{1.1}$ and (b) $|\lambda|\ge \eps^{1.1}$. For $|\lambda|\le \eps^{1.1}$, we have $|\lambda|\delta^2\le \eps^{1.1}\delta^2=
 \eps^{0.1}\eps\delta^2\le \eps^{0.1}(\lambda^2+\eps)\cdot (\lambda^2+\delta^2)$. And for $|\lambda|\ge \eps^{1.1}$, note that $(\lambda^2+\eps)\cdot (\lambda^2+\delta^2)\ge\eps\lambda^2\ge\eps\cdot\eps^{2.2}=\eps^{3.2}$. On the other hand, 
  $|\lambda|\delta^2\le \frac12\delta^2\lesssim\eps^4\le \eps^{.8}\cdot (\lambda^2+\eps)\cdot (\lambda^2+\delta^2)$. This completes the proof of Proposition \ref{ld2}.
\end{proof}
   
   Let us now attach ``weight'' $1$ to the variable $\eps$ and ``weight'' $1/2$ to the variables $\delta$ and $\lambda$. A monomial of the form $\eps^a \lambda^b \delta^c$ carries the weight $a+b/2+c/2$, where $a,b,c\in\mathbb{N}$. In Proposition \ref{neglect} we show that all terms in the power series of $\det M(\eps,\lambda,\delta)$ which are of weight strictly larger than two introduce negligible corrections 
    to $\det M^{approx}(\eps,\lambda,\delta)$ for $\lambda$ and $\delta$ in the region \eqref{smiley}.  
    
    Recalling that there are no pure $\eps$ terms, we see that a monomial in the power series of $D(\eps,\delta^2,\lambda)$ of weight larger than $2$ must have one of the following two forms ($a,b,c\in\mathbb{N}$): 
  \begin{align*}
& (I)\quad \lambda\times \eps^a\lambda^{b} (\delta^{2})^c,\quad {\rm with}\quad a+b/2+c>3/2\implies 2a+b+2c\ge4 \\
& (II)\quad  \delta^2\times \eps^a\lambda^b(\delta^2)^c,\quad {\rm with}\quad a+b/2+c>1\implies 2a+b+2c\ge3 .
   \end{align*}

\begin{proposition}\label{neglect} Terms of form (I) and form (II) that may appear in $D(\eps, \delta^2, \lambda)$ are $o\left( \pit \right)$ as $\eps\to0$,
  for $(\lambda,\delta)$ in the region \eqref{smiley2}:
  $
  |\lambda|\le C_\flat\ \eps^{\frac12}\ \ \textrm{and}\ \ |\delta|\lesssim\eps^2.
$
  We may therefore neglect all terms in the power series of $D(\eps,\delta^2,\lambda)$ which are of weight strictly larger than $2$ for  $(\lambda,\delta)$ in the region \eqref{smiley2}.
  \end{proposition}
  
\nit We prove  Proposition \ref{neglect} by estimating all terms of the form (I) or (II), and we begin with the following lemma, which is a consequence of part 2 of Proposition \ref{roots-delta0}:

\begin{lemma}\label{det-delta0} 
Let $\mu_-(\eps,\lambda),\ \mu_+(\eps,\lambda)$ and $ \widetilde{\mu}(\eps,\lambda)$ denote the three roots of $\det M(\eps,\delta=0,\lambda,\mu)$, defined and analytic in $(\eps,\lambda)$ in a $\C^2$ neighborhood of the origin and displayed in \eqref{mu+expand}-\eqref{mutilde-expand}. Then, 
\begin{align*}
\det M(\eps,\delta=0,\lambda,\mu)= (\mu-\mu_-(\eps,\lambda)) \times\ (\mu-\mu_+(\eps,\lambda))\times (\mu-\widetilde{\mu}(\eps,\lambda))\times\Omega(\eps,\lambda,\mu),
\end{align*}
where $\Omega(\eps,\lambda,\mu)$ is bounded.  In particular, for all $\eps$ such that $0<|\eps|<\eps_0$ there exists a constant $C_\eps$ such that  
\begin{align}\label{bound-det}
\left| \det M(\eps,\delta=0,\lambda,0) \right| &\le C_\eps\ |\lambda|^2 .
\end{align}
\end{lemma}

\nit {\it Proof:} For fixed $\eps$ and $\lambda$, the mapping $\mu\mapsto\det M(\eps,\delta=0,\lambda,\mu)$ is analytic for $|\mu|<\mu_0$ with zeros at $\mu_\pm(\eps,\lambda), \widetilde{\mu}(\eps,\lambda)$; see Proposition \ref{roots-delta0} . Fix $\eps'$ and $\lambda'$ small such that if $|\eps|<\eps'$ and $|\lambda|<\lambda'$,
 the roots all satisfy  $|\mu_+(\eps,\lambda)|$, $ |\mu_-(\eps,\lambda)|$, 
  $|\widetilde{\mu}(\eps,\lambda)|<\mu_0/2$. We claim that for such $\eps$ and $\delta$,
  \[ \Omega(\eps,\lambda,\mu)\equiv 
 \frac{ \det M(\eps,\delta=0,\lambda,\mu)}
  {(\mu-\mu_-(\eps,\lambda)) \times\ (\mu-\mu_+(\eps,\lambda))\times (\mu-\widetilde{\mu}(\eps,\lambda))}\]
  is uniformly bounded for all $|\mu|\le\mu_0$, $|\eps|\le\eps'$ and $|\lambda|\le\lambda'$. Indeed, since the roots are bounded in magnitude by $\mu_0/2$, we have 
 $\max_{|\mu|=\mu_0} \left| \Omega(\eps,\lambda,\mu)  \right|\le (2/\mu_0)^3\ \max_{|\mu|=\mu_0}\left|\det M(\eps,\delta=0,\lambda,\mu)\right|$. Applying the maximum principle we have 
 \[ \max_{|\mu|\le\mu_0}\left| \Omega(\eps,\lambda,\mu)  \right|\le (2/\mu_0)^3\ \max_{|\mu|=\mu_0}\left|\det M(\eps,\delta=0,\lambda,\mu)\right|\le (2/\mu_0)^3 C(\mu_0,\eps',\lambda'),\] where $C(\mu_0,\eps',\lambda')$ is a constant. The bound \eqref{bound-det} now follows
  from the expansions of the roots.

\nit Proof of Proposition \ref{neglect}:\\
{\bf (I) Terms of the form  $\lambda\times \eps^a\lambda^{b} (\delta^{2})^c, \ {\rm with}\ 2a+b+2c\ge4,\ a,b,c\in\mathbb{N}$:}

\nit (i) Suppose first that $c=0$. Then, we consider $\lambda\times\eps^a\lambda^b$ with $2a+b\ge4$. By Lemma \ref{det-delta0} we must have $b\ge1$. 
Thus, $\lambda\times\eps^a\lambda^b=\lambda^2\eps^a \lambda^{b-1}$. If $a\ge2$, then 
$\lambda\times\eps^a\lambda^b=\lambda^2\eps^2 \eps^{a-2} \lambda^{b-1}\lesssim
 (\lambda^2+\delta^2) \times (\lambda^2+\eps)^2= o\left( \pit \right)$ for $(\lambda,\delta)$ in the region \eqref{smiley2}. Otherwise, $a=0$ or $a=1$. If $a=0$, then $b\ge4$ and 
 $\lambda\times \eps^a\lambda^b=\lambda\cdot \lambda^2\times\lambda^2\cdot\lambda^{b-4}
 \lesssim\lambda\cdot (\lambda^2+\eps)\times(\lambda^2+\delta^2)=o\left( \pit \right)$ for $(\lambda,\delta)$ in the region \eqref{smiley}. Now suppose $a=1$. Then, $b\ge2$ and we have
 $\lambda\times\eps^a\lambda^b=\lambda\cdot \lambda^2 \times\eps \cdot\lambda^{b-2}
 \lesssim \lambda\cdot (\lambda^2+\delta^2)\times(\lambda^2+\eps) = 
 o\left( \pit \right)$ for $(\lambda,\delta)$ in the region \eqref{smiley2}.
 
\nit (ii) Suppose now that $c\ge1$. If $b=0$, then 
 $\lambda\times \eps^a\lambda^{b} (\delta^{2})^c= \lambda\delta^2 \eps^a(\delta^2)^{c-1}$
 with $2a+2c\ge4$, which is $=o\left( \pit \right)$ for $(\lambda,\delta)$ in the region by Proposition \ref{ld2}. Finally, if $b\ge1$, then  $\lambda\times\eps^a\lambda^b(\delta^2)^c= \lambda\delta^2\cdot \eps^a\lambda^{b}(\delta^2)^{c-1}=o\left( \pit \right)$ for $(\lambda,\delta)$ in the region \eqref{smiley2}. 

{\bf (II) Terms of the form\  $\delta^2\times \eps^a\lambda^{b} (\delta^{2})^c, \quad {\rm with}\ 
2a+b+2c\ge3,\ a,b,c\in\mathbb{N}$:}

\nit  (i) Suppose first that $a=0$. Then, $\delta^2\times \eps^a\lambda^{b} (\delta^{2})^c=\delta^2\times \lambda^{b} (\delta^{2})^c,\ b+2c\ge3$. If $b\ge1$, then 
we rewrite this as $\delta^2\lambda \times \lambda^{b-1} (\delta^{2})^c=o\left( \pit \right)$ for $(\lambda,\delta)$ in the region \eqref{smiley2}, by Proposition \ref{ld2}. And if $b=0$, then
 $\delta^2\times \eps^a\lambda^{b} (\delta^{2})^c=(\delta^2)^{c+1}$, with $c\ge2$,
  which is $\lesssim \delta^2\delta^2 (\delta^2)^{c-1}\lesssim\delta^2\eps\cdot\eps^3(\delta^2)^{c-1} 
  \lesssim(\lambda^2+\delta^2)(\lambda^2+\eps)\cdot\eps^3(\delta^2)^{c-1}=
 o\left( \pit \right)$ for $(\lambda,\delta)$ in the region \eqref{smiley2}.
 
\nit  (ii) If now $a\ge1$, then  $\delta^2\times \eps^a\lambda^{b} (\delta^{2})^c=\delta^2\eps\cdot \eps^{a-1}\lambda^b(\delta^2)^c\le(\lambda^2+\delta^2)\cdot(\lambda^2+\eps)\cdot \eps^{a-1}\lambda^b(\delta^2)^c=o\left( \pit \right)$ for $(\lambda,\delta)$ in the region \eqref{smiley2}. This completes the proof of Proposition \ref{neglect}. 

 %
 %
 Now each entry in the $3\times3$ matrix, $M(\eps,\delta,\lambda,0)$ is a sum of terms of weight $\ge1/2$; this is a consequence of the expansion  \eqref{Msum}, \eqref{M0}, \eqref{M0-M0app}, \eqref{p-perp2} and the explicit expansion of $M^{approx}$ displayed in \eqref{M0app-expanded}.  If we change any one entry by a term of weight $\ge3/2$, then the effect on the $3\times3$ determinant $D(\eps,\delta^2,\lambda)$ will be a sum of terms of weight $\ge 3/2+1/2+1/2>2$. By Propositions \ref{ld2} and \ref{neglect}, such terms are $o\left( \pit \right)$ in the region \eqref{smiley2}. Therefore, we may compute each entry of $M(\eps,\delta,\lambda,0)$, retaining only terms of weight strictly smaller than $3/2$ and discarding the rest. The resulting determinant will differ from $D(\eps,\delta^2,\lambda)$ by
  terms which are  $o\left( \pit \right)$ in the region \eqref{smiley2}.
 
 We next study the power series of the  $3\times3$ matrix, $M(\eps,\delta,\lambda,0)$, keeping in mind that the relevant monomials are those of weight $\ge1/2$ but strictly less than $3/2$. The complete list of such monomials is: $\eps, \delta, \lambda, \lambda^2, \delta^2$ and $ \lambda\delta$.
 
 Before proceeding further we show that in fact that a monomial of type $\delta^2$ can be neglected. Indeed, the weight $\le2$ contributions of such a monomial to $D(\eps,\lambda,\delta^2,0)=\det M(\eps,\delta,\lambda,0)$ will be a sum of monomials of the type: (i) $\delta^2\times \delta\cdot\delta$, (ii) $ \delta^2\times \lambda\cdot\lambda$ and 
  (iii) $\delta^2\times \lambda\cdot\delta$. Terms of type (ii) and (iii)  are clearly 
   $o(\  \lambda\delta^2 )=o\left( \pit \right)$ as $\eps\to0$ for $(\lambda,\delta)$ in the region \eqref{smiley2}, by Proposition \ref{ld2}. For the type (i) term we have, for $(\lambda,\delta)$ in the region \eqref{smiley2},
    $\delta^2\times \delta\cdot\delta\lesssim \delta^2\ \eps\ \eps\delta\lesssim (\lambda^2+\delta^2)\cdot (\lambda^2+\eps)\ \eps\delta=o\left( \pit \right)$ as $\eps\to0$. Hence, we may strike $\delta^2$ from our list. In particular, we may neglect the contribution from the matrix $\mathcal{M}^{\mathcal{P}}$; see \eqref{MP-expanded}.
  
  \nit{\bf Relevant monomials in the expansion of 
  $M(\eps,\delta,\lambda,0)$ :}\ We shall call the  monomials: $\eps, \delta, \lambda, \lambda^2$ and $ \lambda\delta$ \underline{relevant}. All others are called \underline{irrelevant}.
 
\nit  Stepping back, we have shown above that 
$M=M^{0}+M^V+M^W+M^{\mathcal{P}}$  (\eqref{Msum}),
 where $M^{0}=M^{0,approx}+\mathcal{O}(\eps^2)$ (\eqref{M0-M0app}) and 
  $M^{\mathcal{P}}=\mathcal{O}(\eps^2+\eps\delta+\delta^2)$ (\eqref{MP-expanded}).  We have further shown that the relevant monomial contributions for calculation of $\det M(\eps,\delta,\lambda,0)$   are all contained in 
 $M^{approx}(\eps,\delta,\lambda,0)\equiv M^{0,approx}(\eps,\lambda)+M^V(\eps)+M^W(\delta)$. We consider each of these matrices individually, and  explicitly extract the relevant terms in each; see Proposition  
 \ref{Det-approx} below. 
 %

\nit{\bf Expansion of $M^{0, approx}$:} The entries of $M^{0,approx}$ are
  \begin{align}
&M^{0,approx}_{\sigma,\tsigma}(\eps,\delta,\lambda,0) \nn \\
&\quad \equiv\ \left\langle p_\sigma, \left(\ -\left(\nabla+i[\bK+\lambda\bk_2] \right)^2- [E_\star^0 + \eps(V_{0,0}-V_{1,1})]\ \right) p_\tsigma\right\rangle_{L^2(\R^2/\Lambda_h)} \nn\\
&\quad  = \left\langle p_\sigma,-\left(\nabla+i[\bK+\lambda\bk_2] \right)^2 p_\tsigma\right\rangle_{L^2(\R^2/\Lambda_h)}
-[E_\star^0 + \eps(V_{0,0}-V_{1,1})]\ \delta_{\sigma,\tsigma}\ ,
\label{M0approx1}\end{align}
where we have used that $\left\langle p_\sigma,p_\tsigma\right\rangle=\delta_{\sigma,\tsigma}$.
 The first term in \eqref{M0approx1} may be written, using \eqref{p_sigma}-\eqref{orthon},  $E_\star^0=|\bK|^2$ and $|\bk_2|^2=q^2$  as:
 \begin{align}
 &\left\langle p_\sigma,-\left(\nabla+i[\bK+\lambda\bk_2] \right)^2 p_\tsigma\right\rangle_{L^2(\R^2/\Lambda_h)} \label{psig-nablaK-tpsig}\\
 &\quad =  \left\langle p_\sigma,-(\nabla+i\bK)^2 p_\tsigma\right\rangle
 \ -\ 2i\lambda\bk_2\cdot \left\langle p_\sigma,(\nabla+i\bK) p_\tsigma\right\rangle\
  +\ \left\langle p_\sigma,\lambda^2q^2 p_\tsigma\right\rangle\nn\\
  &\quad = \left(E_\star^0\ +\ \lambda^2q^2 \right)\delta_{\sigma,\tsigma} + \lambda\ J_{\sigma,\tsigma} .\nn
  \end{align}
  Consider now the matrix 
  \begin{align}
  J_{\sigma,\tsigma} &= -2i\bk_2\cdot\left\langle \Phi_\sigma,\nabla \Phi_\tsigma\right\rangle_{L^2_\bK} =-2i\bk_2\cdot\int_\Omega\ \overline{\Phi_\sigma}\ \nabla\Phi_\tsigma d\bx\nn\\
  &=2\ \bk_2\cdot \frac13\left(I+\sigma\overline{\tsigma}R+\overline{\sigma}\tsigma R^2\right)\bK . \label{J-matrix}
  \end{align}
  
  We pause to collect some properties that will enable the evaluation of $J_{\sigma,\tsigma}$;
  see also \cites{FW:12}.
  Recall that $R$ has eigenpairs: $(\tau,\  \zeta)$ 
  and $(\overline{\tau},\ \overline{\zeta}),$ where $\zeta=\frac{1}{\sqrt2}(1,i)^T$. Then, $\overline{\tau}R$ has eigenpairs: $[1,\zeta]$ and $[\tau,\overline{\zeta}]$. Furthermore, 
    $\tau R$ has eigenpairs $[1,\overline{\zeta}]$ and $[\overline{\tau},\zeta]$, and
  \begin{align*}
   \frac13\left[\ I +\overline{\tau}R+(\overline{\tau}R)^2 \right]\zeta=\zeta,\ \ 
     \ \left[\ I +\overline{\tau}R+(\overline{\tau}R)^2 \right]\overline{\zeta}=0, \nn\\
      \frac13\left[\ I + \tau R+(\tau R)^2 \right]\overline{\zeta}=\overline{\zeta},\ \ 
      \left[\ I +\tau R+(\tau R)^2 \right]\zeta=0 .
      \end{align*}
      Hence, $\frac13\left[\ I +\overline{\tau}R+(\overline{\tau}R)^2 \right]$ and $ \frac13\left[\ I + \tau R+(\tau R)^2 \right]$ are, respectively, projections onto $\rm{span}\{\ \zeta\ \}$ and $\rm{span}\{\ \overline{\zeta}\ \}$.
    For any $\bw\in\C^2$, we have $\bw=\left\langle\zeta,\bw\right\rangle_{\C^2}\zeta +
    \left\langle\overline{\zeta},\bw\right\rangle_{\C^2}\overline{\zeta}$, where 
     $\left\langle \bx,\by\right\rangle_{\C^2}=\overline{\bx}\cdot\by$.  Therefore
    \begin{equation}
    \frac13\left[\ I +\overline{\tau}R+(\overline{\tau}R)^2 \right]\bw = \left\langle\zeta,\bw\right\rangle\zeta,\quad\ 
    \frac13\left[\ I + \tau R+(\tau R)^2 \right]\bw = \left\langle\overline{\zeta},\bw\right\rangle\overline{\zeta} .\label{Rprojections}
    \end{equation}
    Also,      $(I-R)(I + R+R^2)=I-R^3=0$ and  therefore
    \begin{equation}
    I + R + R^2\ =\ 0 .
    \label{Rsum0}\end{equation}
    We next calculate  $J_{\sigma,\tsigma} $ using \eqref{Rprojections}-\eqref{Rsum0}.
 Note that $J$ is Hermitian, and by 
 \eqref{Rsum0} its diagonal elements   $J_{\sigma,\sigma} $ all vanish:
$  J_{\sigma,\tsigma} \ =\ \overline{J_{\tsigma,\sigma}},\ J_{\sigma,\sigma}=0,\ \sigma=1,\tau,\overline{\tau}$ .
  It suffices therefore to compute the three entries $J_{1,\tau},\ J_{1,\overline{\tau}}$ and 
  $J_{\tau,\overline{\tau}}$:
  \begin{align*}
  J_{1,\tau}=2\ \frac13\left[\ I +\overline{\tau}R+(\overline{\tau}R)^2 \right]\bK\cdot\bk_2=
  2\ (\overline{\zeta}\cdot\bK)\ (\zeta\cdot\bk_2)\equiv \ \alpha , \\ 
  J_{1,\overline\tau}=2\ \frac13\left[\ I +\tau R+(\tau R)^2 \right]\bK\cdot\bk_2=
  2\ (\zeta\cdot\bK)\ (\overline{\zeta}\cdot\bk_2)=\ \overline{\alpha} , \\ 
  J_{\tau,\overline\tau}=2\  \frac13\left[\ I +\overline\tau R+(\overline\tau R)^2 \right]\bK\cdot\bk_2= 2\ (\overline\zeta\cdot\bK)\ (\zeta\cdot\bk_2)=\ \alpha .
  \end{align*}
  Thus,
  \begin{equation}
  J\ =\  \left(\ J_{\sigma,\tsigma}\ \right)\ =\ 
  \begin{pmatrix} 
  0&\alpha&\overline{\alpha}\\
  \overline{\alpha}&0&\alpha\\
  \alpha&\overline{\alpha}&0
  \end{pmatrix},\ \ \  {\rm where}\ \  \alpha=2\ (\overline{\zeta}\cdot\bK)\ (\zeta\cdot\bk_2) .
  \label{J-matrix-computed}
  \end{equation}
  It follows that 
  \begin{equation}
  M^{0,approx}_{\sigma,\tsigma}(\eps,\delta,\lambda,0)
  =\ \left(\ - \eps(V_{0,0}-V_{1,1}) + \lambda^2q^2 \ \right)\delta_{\sigma,\tsigma}\
   +\ \lambda\ J_{\sigma,\tsigma} . \label{M0app-expanded}
  \end{equation}
  
  \nit{\bf Expansion of $ M^V(\eps)$:} 
  $M_{\sigma,\tsigma} ^V(\eps)= \eps\ \left\langle p_\sigma, V\ p_\tsigma\right\rangle_{L^2(\R^2/\Lambda_h)}\ =\ \eps\ \mathcal{V}_{\sigma,\tsigma},$
where
\begin{align}
\mathcal{V}_{\sigma,\tsigma}& =  \left\langle p_\sigma, V p_\tsigma\right\rangle_{L^2(\R^2/\Lambda_h)}\nn\\
&= \frac{1}{3} \Big[ (1+\sigma \overline{\tsigma}+\overline\sigma \tsigma) V_{0,0}\ +\ \sigma V_{0,1}
\ +\ \overline\tsigma V_{0,-1}\ +\ \tsigma V_{1,0}\ \nn\\
&\qquad\qquad + \ \overline{\sigma} V_{-1,0} \ +\ \sigma\tsigma  V_{1,1}\ +\ \overline{\sigma  \tsigma} V_{-1,-1} \Big] \nn  .
\end{align}
 Since $V$ is real-valued and even, it follows that $V_{-\bfm}=V_{\bfm}$. Furthermore, $V$ is also $R-$ invariant and therefore $V_{0,1}=V_{1,0}=V_{1,1}$. Hence
 \begin{align*}
\mathcal{V}_{\sigma,\tsigma}&= \frac13
(1+\sigma\ \overline{\tsigma}+\overline\sigma\ \tsigma)\ V_{0,0}\ +\  \frac13
(\sigma+\overline\sigma + \tsigma+ \overline\tsigma + \sigma\tsigma+ \overline{\sigma\ \tsigma})\  V_{1,1} .
\end{align*}
$\mathcal{V}$ is clearly symmetric and using that $1+\tau+\tau^2=1+\tau+\overline\tau=0$, we obtain $M^V(\eps)=\eps\ \mathcal{V}$, where
\begin{equation*}
\mathcal{V}\ =\ 
\begin{pmatrix}
V_{0,0}\ +\ 2\ V_{1,1}&0&0\\
0&V_{0,0}-V_{1,1}&0\\
0&0&V_{0,0}-V_{1,1}
\end{pmatrix} .
\end{equation*}

 \nit{\bf Expansion of $M^W(\delta)$:} 
 $M^W_{\sigma,\tsigma}(\eps) = \delta\ \left\langle p_\sigma, W\ p_\tsigma\right\rangle_{L^2(\R^2/\Lambda_h)}\ =\ \delta\ \mathcal{W}_{\sigma,\tsigma},$
where
%
%
\begin{equation}
 \label{MW-step1}
 \mathcal{W}_{\sigma,\tsigma} = 
\frac13\Big[\sigma W_{0,1}\ +\ \overline\tsigma W_{0,-1}\ +\ \overline\sigma W_{-1,0}
+\tsigma W_{1,0}\ +\ \sigma \tsigma W_{1,1}\ +\ \overline{\sigma \tsigma} W_{-1,-1} \Big]  .
\end{equation}
Since $W$ is real and odd, we have that 
$W_{-\bfm}= -W_{\bfm}$ and $W_\bfm$ is purely imaginary.  
Therefore,
\begin{equation*}
\mathcal{W}_{\sigma,\tsigma} = \ \frac{1}3\
\ \Big[\ (\sigma-\overline\tsigma)\ {W}_{0,1}\ +\ (\tsigma-\overline\sigma)\ {W}_{1,0}\
 +\  (\sigma\ \tsigma\ -\ \overline{\sigma\ \tsigma})\ {W}_{1,1}\ \Big] \ ,\ \sigma,\tsigma=1,\tau,\overline{\tau}.
\end{equation*}
It follows that $M^W(\delta)=\delta\ \mathcal{W}$, where 
\begin{equation*}
\mathcal{W} = w_{01} \begin{pmatrix} 0& \tau & -\overline{\tau}\\ \overline{\tau}&-1&0\\ -\tau&0&1\end{pmatrix} +
w_{10} \begin{pmatrix} 0& \overline{\tau}& -\tau\\ \tau &-1&0\\ -\overline{\tau}&0&1\end{pmatrix} +
w_{11} \begin{pmatrix} 0&-1&1\\-1&1&0\\ 1&0&-1\end{pmatrix}  ,
\end{equation*}
and $w_{ij}\equiv -i\ {W}_{i,j} /\sqrt3 \in \R$.

Now assembling all relevant terms (weights $\ge1/2$ and less than $3/2$) we obtain

\begin{proposition}\label{Det-approx} For  $\sigma,\tsigma = 1,\tau,\overline{\tau} $,
\begin{align*}
M_{\sigma,\tsigma}(\eps,\delta,\lambda,0) &\approx\ M^{approx}_{\sigma,\tsigma}(\eps,\delta,\lambda,0)\nn\\
&\approx\ \left(\ - \eps(V_{0,0}-V_{1,1}) + \lambda^2q^2 \ \right)\delta_{\sigma,\tsigma}\
   +\ \lambda\ J_{\sigma,\tsigma}\ + \eps\ \mathcal{V}_{\sigma,\tsigma}\ +\ 
    \delta\ \mathcal{W}_{\sigma,\tsigma}.
    \end{align*}
 Here,  $A_{\sigma,\tsigma}\approx B_{\sigma,\tsigma}$, means that their difference  is a matrix with entries having weight $\ge 3/2$. Hence, the contribution of such terms to the determinant consists of  terms of weight strictly larger than $2$, for $(\lambda,\delta)$ in the region \eqref{smiley2}. Hence, these terms can be  neglected,  by Proposition \ref{neglect}. 
    \end{proposition}
    
\nit   So the calculation of $\det M(\eps,\delta,\lambda,0)$ boils down to the calculation of 
    $\det M^{approx}(\eps,\delta,\lambda,0)$.
    
\nit {\bf Calculation of $\det M^{approx}(\eps,\delta,\lambda,0)$:} 
Assembling the above computations, we have that
{\small
\begin{equation*}
M^{approx} =
\begin{pmatrix}
\lambda^2q^2 + 3 \eps V_{1,1} & 
\alpha \lambda - \delta\widetilde{w} & 
\overline{\alpha} \lambda + \delta\overline{\widetilde{w}} \\
\overline{\alpha} \lambda - \delta\overline{\widetilde{w}} & 
\lambda^2q^2 - \delta (w_{01} +w_{10} - w_{11}) & 
\alpha \lambda \\
\alpha \lambda + \delta\widetilde{w} & 
\overline{\alpha} \lambda & 
\lambda^2q^2 + \delta (w_{01} +w_{10} - w_{11}) \\
\end{pmatrix} ,
\end{equation*}}
where
$\widetilde{w} = w_{11} -w_{01}\tau -w_{10}\overline{\tau}$ and  $\alpha=2(\overline\zeta\cdot\bK)\ (\zeta\cdot\bk_2)$. Note that:
\begin{equation}
\label{alpha-quantities}
\alpha=\frac{q^2}{\sqrt3}\ i\tau,\ \ \Re(\alpha)=-\frac{q^2}2,\ \ \Re(\alpha^3)=0  .
\end{equation}
Calculating the determinant of $M^{approx}$, and using \eqref{alpha-quantities} and that 
$w_{ij}= -i W_{i,j}/\sqrt3$ yields:

\begin{align}
&\det M^{approx} (\eps, \delta, \lambda, 0) \ =\ 
-\left(q^2\lambda^2 + \eps V_{1,1}\right) \left( q^4\lambda^2 + 3\delta^2 (w_{01}^2+w_{10}^2+w_{11}^2) \right)  \nn \\
&\qquad +6 \eps V_{1,1} \delta^2 (w_{11}w_{01}+w_{10}w_{11}-w_{01}w_{10}) 
+\  \mathcal{O}( \lambda \delta^2) + \mathcal{O}(\eps\lambda^4) + \mathcal{O}(\lambda^6) \nn \\
&\quad\ = -\left(q^2\lambda^2 + \eps V_{1,1}\right) \left(\ q^4\lambda^2 + 3\delta^2 (w_{01}+w_{10}-w_{11})^2\ \right)  \nn \\
&\qquad \ \ \ 
+\ \mathcal{O}( \lambda^2\delta^2)\ +\ \mathcal{O}( \lambda \delta^2) + \mathcal{O}(\eps\lambda^4) + \mathcal{O}(\lambda^6)\ ,
\nn\\
 &\quad\ =\  -\left(q^2\lambda^2 + \eps V_{1,1}\right) \left( q^4\lambda^2 + \delta^2 \left| W_{0,1}+W_{1,0}-W_{1,1}\ \right|^2 \right) \nn \\
  &\qquad\ \    
+\ \mathcal{O}( \lambda^2\delta^2)\ +\ \mathcal{O}( \lambda \delta^2) + \mathcal{O}(\eps\lambda^4) + \mathcal{O}(\lambda^6)\nn\\
 &\quad\ =\ -\pi(\eps,\delta^2,\lambda)\ +\ o\left( \pit \right) ,
\label{detMapprox2}
\end{align}
for $(\lambda,\delta)$ in the region \eqref{smiley2}. 
This completes the proof of Proposition \ref{detM-expansion}.

 \subsection{If $\eps V_{1,1}<0$, the zigzag slice does not satisfy the \nofold condition}\label{no-directionalgap}

Recall that to satisfy the case $\eps V_{1,1}<0$ we assume, without loss of generality, that $\eps>0$ and $V_{1,1}<0$. Theorem \ref{NO-directional-gap!} follows from:
 
 \begin{proposition}\label{Deq0}
Assume 
\begin{equation}
0<|\eps|<\eps_2,\ \ {\rm and}\ \ 0\le\delta\le c_\flat\ \eps^2.
\label{smiley4}
\end{equation}
There exists $\theta_0>0$ and $\lambda_\eps>0$ satisfying  $\eps<\lambda_{\eps}<\theta_0\sqrt{\eps}$ such that for all $\eps$ sufficiently small,
\begin{equation*}
\det M(\eps,\delta,\lambda_\eps,\mu=0) = 0.
 \end{equation*}
Thus, $E_\star^\eps$ is an interior point of the $L^2_{k_\parallel}(\Sigma)-$ spectrum of $H^{(\eps,\delta)}$.
 \end{proposition}
\nit It follows that for $\eps V_{1,1}<0$,  the operator $H^{(\eps,\delta)}$ does not have a spectral gap about $E=E_\star^\eps$ along the zigzag slice. Referring to the middle panel of Figure \ref{fig:eps_V11_neg}, we see that, for $\delta\ne0$ and small, a {\it local in $\lambda$} gap opens, for $\lambda$ small, about the energy $E=E_\star^\eps$. But since the no-fold property is not satisfied (by Proposition \ref{Deq0}), this is not a true (global in $\lambda\in[-1/2,1/2]$) spectral gap. 

\begin{proof}[Proof of  Proposition \ref{Deq0}]

 Let $C_\flat$ denote the constant in Proposition 
  \ref{detM-expansion}. Note that $C_\flat$ was chosen to be sufficiently large in the proof of Proposition \ref{detM-expansion} and can be arranged to be taken so that $C_\flat>\theta_0$, where $\theta_0$ is defined by $\theta_0^2=2|V_{1,1}|/q^2$. Also, choose a constant $\zeta_0$ such that  $\zeta_0^2= |V_{1,1}|/2q^2$. Note $\zeta_0<\theta_0$;
 below we shall see why we make these choices.
  For $(\lambda,\delta)$ in the region \eqref{smiley4} we have:
 %
 \begin{equation}-\det M(\eps,\delta,\lambda,0) \ =\  \pi(\eps,\delta^2,\lambda) + o\left( \pit \right)=\pi(\eps,\delta^2,\lambda) + o\left( \eps^2 \right) , \label{-detMonsmiley}
 \end{equation}
 where 
 \begin{align*}
\pi(\eps,\delta^2,\lambda)&\equiv \left(q^2\lambda^2 + \eps V_{1,1}\right) \left( q^4\lambda^2 + \delta^2\ \left|W_{0,1}+W_{1,0}-W_{1,1} \right|^2 \right).
\end{align*}
We now show that there exists  $\lambda^{\eps,\delta}\in (\zeta_0\sqrt\eps,\theta_0\sqrt{\eps})$ such that 
$\pi(\eps,\delta^2,\lambda^{\eps,\delta}) = 0$.
Note first that  $\eps V_{1,1}<0$, $\eps^2\ll\eps$ and the choice of $\zeta_0$ implies, upon evaluation of $\pi(\eps,\delta^2,\lambda)$ at $\lambda=\zeta_0\sqrt\eps$, that:
\begin{align*}
\pi(\eps,\delta^2,\zeta_0\sqrt\eps) 
& =\left(q^2\zeta_0^2\eps + \eps V_{1,1}\right) \left( q^4\zeta_0^2\eps + \delta^2\ \left|W_{0,1}+W_{1,0}-W_{1,1} \right|^2 \right)<0
\end{align*}
and by \eqref{-detMonsmiley} $-\det M(\eps,\delta,\zeta_0\sqrt\eps,0)<0$. 
 On the other hand, the choice of $\theta_0$ implies, upon  evaluation at $\lambda=\theta_0\sqrt{\eps}$ that:
  \begin{align*}
 \pi(\eps,\delta^2,\theta_0\sqrt{\eps})  &=  
  \left(q^2\theta_0^2\eps + \eps V_{1,1}\right) \left( q^4\theta_0^2\eps + \delta^2\ \left|W_{0,1}+W_{1,0}-W_{1,1} \right|^2 \right)>0
  \end{align*}
  and hence, by \eqref{-detMonsmiley} $-\det M(\eps,\delta,\theta_0\sqrt\eps,0)>0$.
  Now $\det M(\eps,\delta^2,\lambda,0)$ is, for all $0<\eps<\eps_1$, a continuous function of $\lambda$. Hence, there exists $\lambda^{\eps,\delta}\in (\zeta_0\sqrt\eps,\theta_0\sqrt{\eps})$ such that $\det M(\eps,\delta,\lambda^{\eps,\delta},\mu=0)=0$. Hence, $E^{\eps,\delta}(\lambda^{\eps,\delta})=E_\star^\eps
\in\ L^2_{\kparpi}-\ {\rm spec}(H^{(\eps,\delta)})$.
This completes the proof of Proposition \ref{Deq0}.
\end{proof}

\appendix
\section{Evaluation of $\eps V_{1,1}$ for two examples}\label{V11-section}

Recall the bases $\{\bv_1,\bv_2\}$ of $\Lambda_h=\Z\bv_1\oplus\Z\bv_2$ and $\{\bk_1,\bk_2\}$ of $\Lambda_h^*=\Z\bk_1\oplus\Z\bk_2$, introduced in Section \ref{sec:honeycomb}. More generally, introduce a lattice spacing parameter, $a>0$, and define the scaled lattices: $\Lambda^{(a)}_h$ and $(\Lambda^{(a)}_h)^*$ with bases:
$
\bv^a_1 = a ( \frac{\sqrt{3}}{2}, \frac{1}{2} )^T, \ 
\bv^a_2 = a ( \frac{\sqrt{3}}{2}, -\frac{1}{2} )^T
$
and 
$
{\bf k}^a_1 = \frac{q}{a} ( \frac{1}{2}, \frac{\sqrt{3}}{2} )^T,\ 
{\bf k}^a_2 = \frac{q}{a} ( \frac{1}{2}, -\frac{\sqrt{3}}{2} )^T
$, 
where 
$q \equiv \frac{4\pi}{\sqrt{3}}$ and  $  \bk_l^{(a)}\cdot \bv_j^{(a)}=2\pi\delta_{lj}$.
Now  introduce the base points:
${\bf A}^{(a)}=(0,0) {\rm \ and\ } {\bf B}^{(a)} = a(\frac{1}{\sqrt{3}},0)$
 and the honeycomb structure with general lattice spacing parameter, $a>0$:
$
{\bf H}^{(a)} =  ({\bf A}^{(a)} + \Lambda^{(a)}_h) \cup   ( {\bf B}^{(a)} + \Lambda^{(a)}_h).
$
To be consistent with previous notation, we write:
$
 {\bf A}^{(1)}={\bf A}= (0,0),\ {\bf B}^{(1)} = {\bf B} = (\frac{1}{\sqrt{3}},0), \ 
 \Lambda_h^{(1)}=\Lambda,\ (\Lambda^{(1)}_h)^*=(\Lambda_h)^*, \ 
 {\bf H}^{(1)}={\bf H} .
$
 
Let $g_0(\bx)$ denote a smooth, real-valued, radially symmetric ($g_0(\bx)=g_0(|\bx|)$) and rapidly decaying function on $\R^2$.  
 Below we shall use the 2D Poisson Summation formula:
 \begin{equation*}
 \sum_{\bn\in\Z^2} f (\bx+\bn\vec\bv) = \frac{(2\pi)^2}{|\Omega_h|}\  \sum_{\bfm\in\Z^2}  \widehat{f} (\bfm\vec\bk) e^{i \bfm\vec\bk \cdot \bx}.
\end{equation*}

 We next present two examples: sums of translates of $g_0$ over the scaled triangular lattice, $\Lambda_h^{(a)}$, and  honeycomb structure, ${\bf H}^{(a)}$. In both cases, $
V^{(a)}_{1,1}$ is expressible in terms of $\widehat{g_0}\left(\frac{4\pi}{\sqrt3}\frac{1}{a}\right) $. Therefore, if $\widehat{g_0}(\xi)$ changes sign, then the sign of $V^{(a)}_{1,1}$ can be changed by varying the lattice constant, $a$. 

\subsection{\bf Example 1:\ Evaluation of $\eps V_{1,1}$ for $V$ equal to  a sum of translates over the scaled triangular lattice, $\Lambda^{(a)}_h$}

Define
$V(\bx;a) = \sum_{\bv\in \Lambda^{(a)}_h} g_0\left(\bx+\bv\right) .$
 The potential $V(\bx;a)$ is a honeycomb potential; it is $\Lambda^{(a)}_h-$ periodic, inversion symmetric and $\mathcal{R}-$ invariant with respect to  the origin of coordinates $\bx_0=0$.
 \medskip
 
\nit {\bf Claim 1.} 
 $
V^{(a)}_{1,1}= \frac{(2\pi)^2}{|\Omega_h|}\ \widehat{g_0}\left(\frac{\bk_1+\bk_2}{a}\right)= \frac{(2\pi)^2}{|\Omega_h|}\ 
\widehat{g_0}\left(\frac{4\pi}{\sqrt3}\frac{1}{a}\right) . 
$

\nit{\it Proof of Claim 1.}  
The Poisson summation gives
\begin{equation*}
V(\bx;a) = \frac{(2\pi)^2}{|\Omega_h|}\ \sum_{\bfm\in\Z^2}
  \widehat{ g_0}(\bfm\vec\bk^{(a)})
   e^{i (\bfm\vec\bk^{(a)})\cdot \bx} .
\end{equation*}
Recall also that
$\bfm\vec\bk^{(a)} = m_1\bk_1^{(a)}+m_2\bk_2^{(a)}=\frac{1}{a} (m_1\bk_1+m_2\bk_2).$
Claim 1 now follows from:
 \begin{equation*}
  V(\bx;a) =  \frac{(2\pi)^2}{|\Omega_h|}\ \sum_{\bfm\in\Z^2}\
 \widehat{g_0}\left(\frac{m_1\bk_1+m_2\bk_2}{a} \right)
 \  e^{i \bfm\vec\bk^{(a)}\cdot \bx} 
  \ = \sum_{\bfm\in\Z^2}\ V^{(a)}_\bfm\  e^{i \bfm\vec\bk^{(a)}\cdot \bx} .
  \end{equation*}
 
\subsection{
Evaluation of $\eps V_{1,1}$ for $V$ equal to a sum of translates over the  scaled honeycomb structure 
${\bf H}^{(a)}$ }
 
As earlier, take ${\bf A}^{(a)}=(0,0)^T$ and ${\bf B}^{(a)}=a(\frac{1}{\sqrt3},0)^T$.
The point at the center of hexagon, immediately northeast of ${\bf A}^{(a)}$ is ${\bf \tau_0}^{(a)}=\frac{a}{2}(\frac{1}{\sqrt{3}},1)^T$.
 Define
 \begin{equation*}
V(\bx;a) = \sum_{\bv\in \Lambda^{(a)}_h} g_0\left(\bx - {\bf A}^{(a)} + {\bf \tau_0}^{(a)} +\bv\right)
 \ +\  \sum_{\bv\in \Lambda^{(a)}_h} g_0\left(\bx-{\bf B}^{(a)} + {\bf \tau_0}^{(a)} +\bv\right) .
 \end{equation*}
$V(\bx;a)$ is a honeycomb potential; it is $\Lambda_h-$ periodic, inversion symmetric and $\mathcal{R}-$ invariant with respect to the origin coordinates, $\bx_0=0$, located at the center of a hexagon.
 \medskip

\nit {\bf Claim 2.} 
 $
V^{(a)}_{1,1}= -\frac{(2\pi)^2}{|\Omega_h|}\ \times\widehat{g_0}\left(\frac{4\pi}{\sqrt3}\frac{1}{a}\right) . 
$

\nit{\it Proof of Claim 2.}  
 Poisson summation yields
 {\footnotesize
 \begin{equation*}
 V(\bx;a) = \frac{(2\pi)^2}{|\Omega_h|} \sum_{\bfm\in\Z^2}
  \Big[ 
  \widehat{ g_0}\left( \cdot-{\bf A}^{(a)} + {\bf \tau_0}^{(a)}\right)(\bfm\vec\bk^{(a)})
+  \widehat{g_0}\left(\cdot-{\bf B}^{(a)} + {\bf \tau_0}^{(a)}\right)(\bfm\vec\bk^{(a)})
  \Big]
    e^{i (\bfm\vec\bk^{(a)})\cdot \bx} .
 \end{equation*}}
Now, $ \widehat{ g_0}\left( \cdot - {\bf A}^{(a)} + {\bf \tau_0}^{(a)}\right)(\xi)=\exp\left( i {\bf \xi}\cdot {\bf \tau_0}^{(a)} \right) \widehat{g_0}({\bf \xi})$ and
 $ \widehat{ g_0}\left( \cdot - {\bf B}^{(a)} + {\bf \tau_0}^{(a)}\right)(\xi)=\exp\left( i {\bf \xi}\cdot \left(-{\bf B}^{(a)} + {\bf \tau_0}^{(a)}\right) \right) \widehat{g_0}({\bf \xi})$, 
 and therefore
 \begin{align*}
  V(\bx;a) &=  \frac{(2\pi)^2}{|\Omega_h|}\ \sum_{\bfm\in\Z^2}
  \left[ e^{i \bfm\vec\bk^{(a)}\cdot {\bf \tau_0}^{(a)}} +  e^{i \bfm\vec\bk^{(a)}\cdot \left(-{\bf B}^{(a)} + {\bf \tau_0}^{(a)}\right)} \right]\ 
  \widehat{g_0}(\bfm\vec\bk^{(a)}) \ e^{i \bfm\vec\bk^{(a)}\cdot \bx} .
  \end{align*}
Noting that
$\bfm\vec\bk^{(a)}\cdot {\bf \tau_0}^{(a)} = \bfm\vec\bk\cdot {\bf \tau_0}$ and 
$\bfm\vec\bk^{(a)}\cdot (-{\bf B}^{(a)} + {\bf \tau_0}^{(a)}) = \bfm\vec\bk\cdot (-{\bf B} + {\bf \tau_0})$
(independent of the lattice constant, $a$), we have
\begin{align*}
\bfm\vec\bk^{(a)}\cdot {\bf \tau_0}^{(a)} &= (m_1\bk_1+m_2\bk_2)\cdot {\bf \tau_0} = \frac{2\pi}{3} (2m_1 - m_2), \quad \text{and} \\
\bfm\vec\bk^{(a)}\cdot (-{\bf B}^{(a)} + {\bf \tau_0}^{(a)}) &= (m_1\bk_1+m_2\bk_2)\cdot (-{\bf B} + {\bf \tau_0}) = -\frac{2\pi}{3} (2m_1 - m_2).
\end{align*}
 Recall again that $\bfm\vec\bk^{(a)} = m_1\bk_1^{(a)}+m_2\bk_2^{(a)}=\frac{1}{a} (m_1\bk_1+m_2\bk_2)$.
Hence,
 \begin{align*}
  V(\bx;a) &=  \frac{(2\pi)^2}{|\Omega_h|} \sum_{\bfm\in\Z^2} 
  \left[ e^{\frac{2\pi i}{3} (2m_1 - m_2)} +  e^{-\frac{2\pi i}{3} (2m_1 - m_2) } \right]
  \widehat{g_0}\left(\frac{m_1\bk_1+m_2\bk_2}{a} \right)
  e^{i \bfm\vec\bk^{(a)}\cdot \bx} \\
 &= \frac{(2\pi)^2}{|\Omega_h|} \sum_{\bfm\in\Z^2} 
   2 \cos\left( \frac{2\pi}{3} (2m_1 - m_2) \right)\ \widehat{g_0}\left(\frac{m_1\bk_1+m_2\bk_2}{a} \right)
 \  e^{i \bfm\vec\bk^{(a)}\cdot \bx} \\
  &= \sum_{\bfm\in\Z^2} V^{(a)}_\bfm\  e^{i \bfm\vec\bk^{(a)}\cdot \bx} .
  \end{align*}
Therefore, 
$V^{(a)}_{1,1}= \frac{(2\pi)^2}{|\Omega_h|}\ 2\cos\Big(\frac{2\pi}{3}\Big)\times  \widehat{g_0}\left(\frac{\bk_1+\bk_2}{a}\right)=
-\frac{(2\pi)^2}{|\Omega_h|}\ \times\widehat{g_0}\left(\frac{4\pi}{\sqrt3}\frac{1}{a}\right) . 
$

\bibliographystyle{amsplain}
\bibliography{honey-edge}

\end{document}